\documentclass[12pt,a4paper,english]{article}

\usepackage[latin9]{inputenc}
\usepackage[a4paper]{geometry}
\usepackage{amsmath}
\usepackage{amssymb}
\usepackage{float}
\usepackage{esint}
\usepackage{authblk}
\usepackage[authoryear]{natbib}
\usepackage{longtable}
\usepackage{multirow}
\usepackage{hyperref}
\usepackage{xcolor}
\usepackage{booktabs}
\usepackage{rotating,capt-of}
\usepackage{caption}
\usepackage{setspace}
\usepackage{bm}

\makeatletter

\usepackage{amsfonts}
\usepackage[english]{babel}
\usepackage{graphicx}
\usepackage{color}
\usepackage{colordvi}
\usepackage{multirow}
\usepackage{multicol}
\usepackage{babel}
\usepackage{amsthm}
\usepackage{subcaption}
\theoremstyle{definition}
\newtheorem{exmp}{Example}
\newtheorem{rem}{Remark}[section]
\newtheorem{prop}{Proposition}[section]

\newdimen\defpicwidth
\newdimen\defepswidth
\def\SKparam#1#2{}

\newcommand{\IF}{\boldsymbol{1}}    

\long\def\@makecaption#1#2{%
  \vskip\abovecaptionskip
  \sbox\@tempboxa{#1. #2}%
  \ifdim \wd\@tempboxa >\hsize
    {\centerline{\renewcommand{\baselinestretch}{0.8}%
\small\normalsize\parbox{0.8\textwidth}{#1. #2}}}\par
  \else
    \global \@minipagefalse
    \hbox to\hsize{\hfil\box\@tempboxa\hfil}%
  \fi
  \vskip\belowcaptionskip}

{\begin{longtable}{|p{1mm}p{3mm}p{0.8\textwidth}p{1mm}|}
\hline\hline
&&&\\[-4mm]
&& \hskip-20pt \centerline{\bf \summaryname}&\\[1mm]
\hline\hline\endfirsthead
\hline\hline
&&&\\[-4mm]
&& \hskip-20pt \centerline{\bf \summarycontname }&\\[1mm]
\hline\hline\endhead
\hline\hline\endfoot
}
{\end{longtable}}

\providecommand{\pdf}[1]{}

\newcommand{\dirxplore}[1]{../#1}
\newcommand{\dircodes}[1]{\wwwebook/codes/#1}

\def\wwwebook{http://www.quantlet.de}

\def\definitionname{Definition}

\def\summarycontname{Summary (continued)}
\def\summaryname{Summary}

\def\bookname{}

\ifx\sfb\undefined
  \newcommand{\sfb}{}
  \fi

  \def\bookname{sfs}
\SKparam{BEGIN}{xpl ps pdf}
\SKparam{INPUT}{titlesv}
\SKparam{END}{}
\SKparam{BEGIN}{html java}
\SKparam{INPUT}{matitcs}
\SKparam{END}{}
\SKparam{LIB}{sfs}
\def\wwwebook{http://www.quantlet.com/mdstat}

\newcommand{\E}{\mathop{\mbox{\sf E}}} 

\SKparam{BEGIN}{ps pdf}
\SKparam{INPUT}{dedic}
\SKparam{INPUT}{pref}
\SKparam{END}{}
\newcommand{\argmax}{\mbox{arg}\max}

\def\defeq{\stackrel{\mathrm{def}}{=}}  

\newtheorem{theorem}{Theorem}[section]

\newtheorem{lemma}{Lemma}[section]

\newtheorem{definition}{\definitionname}[section]

\newcommand{\vps}{{\varepsilon}}
\newtheorem{assumption}{Assumption}[section]
\def\defeq{\stackrel{\mathrm{def}}{=}}  

\def\beq{\begin{equation}}
	\def\eeq{\end{equation}}
\def\beqq{\begin{equation*}}
	\def\eeqq{\end{equation*}}

\def\bea{\begin{eqnarray}}
	\def\eea{\end{eqnarray}}

\def\beaa{\begin{eqnarray*}}
	\def\eeaa{\end{eqnarray*}}
\makeatother

\renewcommand{\baselinestretch}{1.3}
\setcitestyle{authoryear,open={(},close={)}} 

\title{\Large{\textbf{Policy Choice in Time Series by Empirical Welfare Maximization}\Large \thanks{%
We are grateful to Karun Adusumilli, Isaiah Andrews, Joshua Angrist, Giuseppe Cavaliere, Yoosoon Chang, Wei Cui, Whitney Newey, Joon Park, Mikkel Plagborg-M{\o}ller, Ashesh Rambachan, Barbara Rossi, Frank Schorfheide, Neil Shephard, Christian Wolf, and  the participants at the EEA-ESEM 2021, IAAE Annual Meeting 2022, and seminars at CeMMAP, Glasgow, Harvard/MIT, Indiana, KU Leuven, LMU Munich, Michigan, Notre Dame, OCIS, Oxford, Philadelphia Fed, SFU, TEDS, and UC3M for beneficial comments. We are grateful to Thomas Carr for his excellent research assistance.  Financial support from
the ESRC through the ESRC CeMMAP
(grant number RES-589-28-0001), the ERC (grant number 715940), and the NSF (grant number FAIN 234361) is gratefully acknowledged.
}}}
\author{Toru Kitagawa\thanks{Department of Economics, Brown University and Department of Economics, University College London. Email: toru\_kitagawa@brown.edu} \hspace{0.5cm} Weining Wang\thanks{Department of Economics, University of Bristol. Email: weining.wang@bristol.ac.uk} \hspace{0.5cm} Mengshan Xu\thanks{Department of Economics, University of Mannheim. Email: mengshan.xu@uni-mannheim.de}}
\begin{document}

\maketitle

\begin{abstract}
This paper develops a novel method for policy choice in a dynamic setting where the available data is a multivariate time series. Overcoming challenges unique to the time-series setting such as time-varying environments, history-dependent welfare, dynamic causal effects, and statistical dependence, we propose Time-series Empirical Welfare Maximization (T-EWM) methods. We characterize conditions for T-EWM to consistently learn optimal policies conditional or unconditional on the time-series history, and derive
nonasymptotic upper bounds for the  welfare regret. We illustrate the use of T-EWM for optimal restriction rules against Covid-19. 
\end{abstract}
\thispagestyle{empty}
\pagebreak

\section{Introduction}
\setcounter{page}{1}
\onehalfspacing

A central topic in economics is the nature of the causal relationships between economic outcomes and government policies, both within and across time periods. To investigate this, empirical research makes use of time-series data, with the aim of finding  desirable policy rules. 
For instance, a monetary policy authority may wish to use past and current macroeconomic data  to learn an interest rate policy that is optimal in terms of a social welfare criterion. 
Building on the recent development in potential outcome time series (\cite{white2010granger}, \cite{angrist2018semiparametric}, \cite{bojinov2019time}, and \cite{rambachan2021when}), this paper proposes a novel method to inform policy choice when the available data is a multivariate time series.

In contrast to the structural and semi-structural approaches that are common in macroeconomic policy analysis, such as dynamic stochastic general equilibrium (DSGE) models and structural vector autoregressions (SVAR), we set up the policy choice problem from the perspective of the statistical treatment choice proposed by \cite{manski2004statistical}. The existing statistical treatment choice literature typically focuses on microeconomic applications in a static setting, and the applicability of these methods to a time-series setting has yet to be explored. In this paper, we propose a novel statistical treatment choice framework for time-series data and study how to learn an optimal policy rule. Specifically, we consider extending the conditional empirical success (CES) rule of \cite{manski2004statistical} and the empirical welfare maximization rule of \cite{kitagawa2018should} to time-series policy choice, and characterize the conditions under which these approaches can inform welfare-optimal policy. These conditions do not require functional form specifications for structural equations or the exact temporal dependence of the time-series observations, but can be 
connected to the structural approach under certain conditions.

In the standard microeconometric setting considered in the treatment choice literature, the planner has access to a random sample of cross-sectional units, and it is often assumed that the populations from which the sample was drawn and to which the policy will be applied are the same. These assumptions are not feasible or credible in the time-series context, which leads to  several non-trivial challenges. First, the economic environment and the economy's causal response to it may be time-varying. Assumptions are required to make it possible to learn an optimal policy rule for future periods based on available past data.
In addition, 
the outcomes and policies observed in the available data can be statistically and causally dependent in a complex manner, 
and accordingly, the identifiability of social welfare under counterfactual policies becomes non-trivial and requires some conditions on how past policies were assigned over time. 
Second, in defining an optimal policy in a time-series context, it is natural to consider both unconditional welfare and conditional  welfare given the history of observables available at the time the policy is chosen. The conditional perspective differs from the unconditional one in that it evaluates welfare given the realized history, whereas the unconditional criterion averages these conditional measures across all possible histories. These welfare criteria are unique in the time-series setting, contrasting with the standard settings with cross-sectional microdata. Whether the unconditional or the conditional welfare criterion is more relevant depends on applications and the data environment. Third, when past data are used to inform policy, we have only a single realization of a time series in which the observations are dependent across periods and possibly nonstationary. Such statistical dependence complicates the characterization of the statistical convergence of the welfare performance of an estimated policy. Fourth, if the planner wishes to learn a dynamic assignment policy, which prescribes a policy in each period over multiple periods on the basis of observable information available at the beginning of every period, the policy learning problem becomes substantially more involved. This is because a policy choice in the current period may affect subsequent policy choices through the current policy assignment and a realized outcome under the assigned treatment. 

Taking into account these challenges, we propose time-series empirical welfare maximization (T-EWM) methods that construct an empirical welfare criterion based on a historical average of the outcomes and obtain a policy rule by maximizing the empirical welfare criterion over a class of policy rules. We then clarify the conditions on the causal structure and data-generating process under which T-EWM methods consistently estimate a policy rule that is optimal in terms of unconditional or conditional welfare. Extending the regret bound analysis of \cite{manski2004statistical} and \cite{kitagawa2018should} to time-series dependent observations, we obtain a finite-sample uniform bound for welfare regret. Moreover, we characterize the convergence of welfare regret and establish the nonasymptotic upper bounds  of the T-EWM rule.

Our development of T-EWM builds on the recent potential outcome time-series literature including \cite{white2010granger}, \cite{angrist2018semiparametric}, \cite{bojinov2019time}, and \cite{rambachan2021when}. In particular, to identify the counterfactual welfare criterion, we employ the sequential exogeneity restriction considered in \cite{bojinov2019time}. Our framework of T-EWM is also closely related to the literature of switchback experiments in which experimenters repeatedly and randomly assign treatments on a fixed unit over multiple periods; see \cite{Bojinov_etal2023}, \cite{Hu_Wager2023}, \cite{Liang_Recht2025}, and references therein.     This line of literature  focuses on either retrospective evaluation of the causal impact of policies observed in historical data or estimation and inference for various average treatment effects taking into account dynamic causal (carryover) effects. In contrast, our focus is on how to perform ex ante policy choices based on the historical evidence.

Since the seminal works of \cite{manski2004statistical} and \cite{Dehejia2005}, statistical treatment choice and empirical welfare maximization have been active topics of research, e.g., \cite{stoye2009minimax, stoye2012minimax}, \cite{QianMurphy2011}, \cite{tetenov2012statistical}, \cite{BhattacharyaDupas2012}, \cite{Zhao2012JASA}, \cite{kitagawa2018should, kitagawa2021equality}, \cite{kallus2021more}, \cite{athey2021policy}, \cite{MT17}, \cite{KST21}, among others. These works focus on a setting where the available data is a cross-sectional random sample obtained from an experimental or observational study with randomized treatment, possibly conditional on observable characteristics. \cite{Viviano21} and \cite{Abhishek2020OptimalTreatment} consider EWM approaches for treatment allocations where the training data features cross-sectional dependence due to network spillovers, while to our knowledge, this paper is the first to consider policy choice with time-series data. As a related but distinct problem, there is a large literature on the estimation of dynamic treatment regimes, \cite{murphy2003optimal}, \cite{zhao2015new}, \cite{han2021optimal},  \cite{Sakaguchi21}, and \cite{ida2024dynamic}, among others. The problem of dynamic treatment regimes assumes that training data is a short panel in which treatments have been randomized both among cross-sectional units and across time periods. Recently, \cite{adusumilli2019dynamically} consider an optimal policy in a dynamic treatment assignment problem with a budget constraint where the planner allocates treatments to subjects arriving sequentially. 
The T-EWM framework, in contrast, assumes observations are drawn from a single time series as is common in empirical macroeconomics. 

A large literature on multi-arm bandit algorithms analyzes learning and dynamic allocations of treatments when there is a trade-off between exploration and exploitation. See \cite{lattimore2020bandit} and references therein,  and \cite{Dimakopoulou_et_al_2017}, \cite{Kock2020functional}, \cite{kasy2019adaptive}, \cite{adusumilli2021risk}, and \cite{Kitagawa_Rowley_2024} for recent works in econometrics. The setting in this paper differs from the standard multi-arm bandit setting in the following three respects. First, our framework treats the available past data as a training sample and focuses on optimizing short-run welfare. We are hence concerned with the performance of the method in terms of short-term regret rather than cumulative regret over a long horizon. 
Second, in the standard multi-arm bandit problem, subjects to be treated are assumed to differ across rounds, which implies that the outcome-generating process is independent over time. This is not the case in our setting, and we include the realization of outcomes and policies in the past periods as contextual information for the current decision. {Third, even if bandit algorithms can be adjusted to take into account the dependence of observations, our method is then analogous to the ``pure exploration'' class, 
involving a long exploration phase followed by a one-period exploitation at the very end. However, a major difference is that the bandit algorithm concerns data in random experiments while our method is aimed at data in quasi-random experiments. }

The analysis of welfare regret bounds is similar to the derivation of risk bounds in empirical risk minimization, as reviewed by \cite{Vapnik1998Book} and \cite{Lugosi2002}. Risk bounds studied in the empirical risk minimization literature typically assume independent and identically distributed (i.i.d.) training data. A few exceptions, \cite{jiang2010risk}, \cite{brownlees2021performance}, and \cite{brownlees2021empirical}, obtain risk bounds for empirical risk minimizing predictions with time-series data, but they do not consider welfare regret bounds for causal policy learning.

The rest of the paper is organized as follows. Section \ref{model_example}  describes the setting using a simple illustrative model with a single discrete covariate. Section \ref{continuous} discusses the general model with continuous covariates and presents the main theorems. In Section \ref{ext}, we discuss extensions to our proposed framework, including  a case of multi-period welfare functions, how T-EWM is related to the Lucas critique, as well as T-EWM's links with structural vector autoregressive models, Markov Decision Processes, and reinforcement learning. In Section \ref{applica}, we present  an empirical application. Technical proofs, simulation studies, and other details are presented in Appendices.

\section{Model and illustrative example} \label{model_example}

In this section, we introduce the basic setting, notation, the welfare criterion we aim to maximize, and conditions on the data-generating process that are important for the learnability of an optimal policy. Then, we illustrate the main analytical tools used to bound welfare regret through a heuristic model with a simple dynamic structure. 


\subsection{Notation, timing, and welfare} \label{sec:notation_timing_welfare}
We suppose that the social planner is at the beginning of time $T$. Let $W_t \in \{0, 1 \}$ denote a treatment or policy (e.g., nominal interest rate) implemented at time $t = 0,1,2,\dots$ To simplify the analysis, we assume that $W_t$ is binary (e.g., a high or low interest rate). The planner sets $W_T \in \{0, 1\}$, $T \geq 1$, making use of the history of observable information up to period $T-1$ to inform her decision. This observable information consists of an economic outcome (e.g., GDP, unemployment rate, etc.), $Y_{0:T-1}=(Y_0,Y_1,Y_2, \cdots, Y_{T-1})$, the history of implemented policies, $W_{0:T-1}=(W_0,W_1,W_2, \cdots, W_{T-1})$, contextual covariates (e.g., inflation), $Z_{0:T-1}=(Z_0, Z_1,Z_2, \cdots, Z_{T-1})$, and potential mediators between treatment and outcome, $M_{0:T-1}=(M_0, M_1,M_2, \cdots, M_{T-1})$. $Z_t$ and $M_t$ can be  multidimensional vectors, but $Y_t$ and $W_t$ are assumed to be univariate. 


Following \cite{bojinov2019time}, we refer to a sequence of policies $w_{0:t} = (w_0, w_1,$ $ \dots, w_t) \in \{0,1\}^{t+1}$, $t \geq 0$, as a treatment path. A realized treatment path observed in the data $0 \leq t \leq T-1$ is a stochastic process $W_{0:T-1}$ drawn from the data-generating process. 
The timing of realizations is therefore: 
\begin{figure}[H]
	\centering
	\includegraphics[scale=0.5]{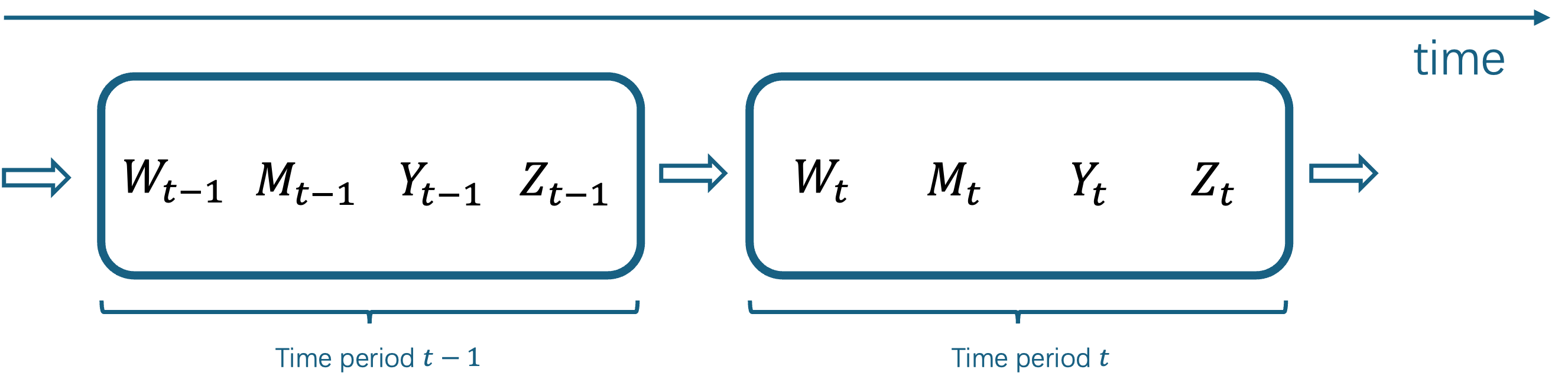}
\end{figure}
  \noindent That is, period $t$ starts after  $Z_{t-1}$ is realized and before $W_t$ is realized.
Furthermore, let
\begin{equation} \label{covariate}
X_{t}=\{W_{t},M_{t}^{\prime}, Y_{t}, Z_{t}^{\prime}\}^{\prime}
\end{equation}
collect the observable variables for period $t$. For $t = 0,1,2,\dots$, let $\mathcal{X}$ denote the sample space of $X_t$. Furthermore,  we define the filtration
$$\mathcal{F}_{t-1}=\sigma (X_{0:t-1}),$$ 
where $\sigma(\cdot)$ denotes the Borel $\sigma$-algebra generated by the variables specified in the argument. The filtration $\mathcal{F}_{T-1}$ corresponds to the planner's information set at the time of making her decision in period $T$. 

Following the framework of \cite{bojinov2019time}, we introduce potential outcome time series. At each $t = 0, 1, 2, \dots,$ and for every treatment path $w_{0:t} \in \{0 ,1 \}^{t+1}$, let $Y_t(w_{0:t}) \in \mathbb{R}$ be the realized period $t$ outcome, and the treatment path from $0$ to period $t$ is $w_{0:t}$. Hence, we have a collection of potential outcome paths indexed by treatment path, 
$$
\{ Y_t(w_{0:t}) : w_{0:t} \in \{0, 1 \}^{t+1}, t=0,1,2,\dots \},
$$
which defines $2^{t+1}$ potential outcomes in each period $t$. This is an extension of the Neyman-Rubin causal model originally developed for cross-sectional causal inference. As maintained in \cite{bojinov2019time}, the potential outcomes for each $t$ are indexed by the current and past treatments $w_{0:t}$ only. This imposes the restriction that any future treatment $w_{t+p}$, $p \geq 1$, does not causally affect the current outcome, i.e., an exclusion restriction for future treatments. 

For a realized treatment path $W_{0:t}$, the observed outcome $Y_t$ and the potential outcomes satisfy
$$
Y_t = \sum_{w_{0:t} \in \{ 0, 1\}^{t+1} } 1 \{W_{0:t} = w_{0:t} \} Y_t(w_{0:t})  
$$
for all $t \geq  0$.

The baseline setting of the current paper considers the choice of policy $W_T$ for a single period $T$.\footnote{Section \ref{sec_multi} discusses how to extend the single-period policy choice problem to multi-period settings.} We denote the policy choice based on observations up to period $T-1$ by 
\beq \label{eq:policy_f} 
g: \mathcal{X}^{T} \to \{0,1\},
\eeq
where $\mathcal{X}^T$ denotes the sample space of $X_{0:T-1}$. The period-$T$ treatment is $W_T = g(X_{0:T-1})$, and we refer to $g(\cdot)$ as a \emph{decision rule}.

We assume that the planner's preferences for policies in period $T$ are embodied in a social welfare criterion. In particular, we define one-period \emph{welfare conditional on $\mathcal{F}_{T-1}$} (conditional welfare, for short) to be\footnote{Throughout the paper, the expectation, $\E$, and the probability, $\Pr$, are understood as the outer measure whenever a measurability issue arises.} 
\bea \label{wel:condi}
\mathcal{W}_T(g|\mathcal{F}_{T-1})
= \E \left[ Y_{T}(W_{0:T-1},1) g(X_{0:T-1})+ Y_{T}(W_{0:T-1},0)(1- g(X_{0:T-1}))|\mathcal{F}_{T-1} \right].
\eea


As an alternative welfare criterion, we can define the \emph{unconditional welfare} by averaging out the conditioning information in \eqref{wel:condi},
\bea \label{eq:uncon_g}
    \mathcal{W}_{T}(g) = \E \left[ Y_{T}(W_{0:T-1},1) g(X_{0:T-1})+ Y_{T}(W_{0:T-1},0)(1- g(X_{0:T-1}))\right].
\eea
Whether the planner's objective function should be conditional or unconditional welfare depends on applications and the goal of analysis. If the planner's goal is to perform policy choice in a potentially nonstationary environment given the realized path of scenarios, it is natural to set the conditional welfare as the criterion to maximize. This case contrasts with the setting of the cross-sectional treatment choice with micro data where the unconditional welfare that averages out the observable characteristics of a unit is common.\footnote{\cite{manski2004statistical} also considers a conditional welfare criterion in the cross-sectional setting.}

 On the other hand, there are cases in which studying optimal policy from the perspective of unconditional welfare is suitable. First, an optimal policy in terms of the unconditional welfare can inform an optimal policy action that the planner would choose under a counterfactual history. Such analysis could be of interest in the practice of \textit{ex post} counterfactual policy analysis.
 Second, in settings where multiple units exist, learning an unconditionally optimal policy from the time-series data of a particular unit informs the planner of optimal policies for other similar units but with different histories. 
 Third, an unconditionally optimal policy can also inform an optimal policy in the setting where the planner implements the same Markovian policy over multiple time periods, assuming the data-generating process is stationary.
 
Note that the T-EWM based on unconditional welfare yields the same empirical average as the EWM in cross-sectional data. The two, however, differ markedly in substance. In time-series settings, observed outcomes and policies are subject to statistical and causal dependence across periods, whereas cross-sectional data are typically assumed to be i.i.d. across the units. Due to this difference, external-validity type assumptions become more debatable in the time-series setting than in the cross-sectional setting. 

 Learning conditionally optimal policies is generally more challenging than learning unconditionally optimal ones since estimation of the unconditional welfare can use all the observations in the sample, whereas estimation of the conditional welfare may have to rely on a subset of observations. 
In Section \ref{sec:uncon_con}, however, we show that under certain assumptions, the unconditional welfare regret can bound the conditional welfare regret so that one can learn the conditionally optimal policy as fast as the unconditionally optimal one.

For the welfare targets \eqref{wel:condi} and \eqref{eq:uncon_g}, define the planner's optimal  conditional and unconditional policies that maximize her one-period welfare as
\beaa
g^{\ast}_{X_{0:T-1}} &\in & \arg \max_g \mathcal{W}_T(g| \mathcal{F}_{T-1}),\label{p:opti_p}\\
g^{\ast} &\in & \arg \max_g \mathcal{W}_T(g),\nonumber
\eeaa
respectively. The former, $g^{\ast}_{X_{0:T-1}}$, can be regarded as the optimal policy $g^{\ast}$ in the latter evaluated at the conditioning value of $X_{0:T-1}$. The planner does not know $g^{\ast}$, so she instead seeks a statistical treatment choice rule (\citeauthor{manski2004statistical}, \citeyear{manski2004statistical}) $\hat{g}$, which is a decision rule selected on the basis of the available data $X_{0:T-1}$.

Our goal is to develop a way of obtaining $\hat{g}$ that performs well in terms of the conditional and unconditional welfare criteria (\ref{wel:condi}) and \eqref{eq:uncon_g}. Specifically, we assess the statistical performance of an estimated policy rule $\hat{g}$ in terms of the convergence of conditional and unconditional welfare regrets,
\beaa 
\mathcal{R}(\hat{g}|x_{0:T-1})&=&\mathcal{W}_T(g^{\ast}|X_{0:T-1} = x_{0:T-1}) - \mathcal{W}_T(\hat{g}|X_{0:T-1} = x_{0:T-1}), \\
\mathcal{R}(\hat{g})&=&\mathcal{W}_T(g^{\ast}) -\mathcal{W}_T(\hat{g}),\label{cond_regret_general}
\eeaa
as well as their convergence rate with respect to the sample size $T$. When examining convergence, we accommodate statistical uncertainty over $\hat{g}$ by focusing on convergence with probability approaching one uniformly over a class of sampling distributions for $X_{0:T-1}$.

\subsection{An illustrative model with a discrete covariate} \label{one-p}
We begin our analysis with a simple illustrative  model. Our aim is to provide  a heuristic exposition of the main idea of T-EWM and its statistical properties, rather than to focus on realistic scenarios. We will cover more general settings and extensions in Sections \ref{continuous} and \ref{ext}. 

Suppose that the data consist of a bivariate time series $X_{0:T-1} = ((Y_t,W_t) \in \mathbb{R} \times \{0, 1 \} : t=0,1, \dots, T-1)$ with no other covariates. To simplify exposition for the illustrative model, we impose the following restrictions on the dynamic causal structure and dependence of the observations.

\begin{assumption}  [Markov properties] \label{ass:toy_example}The time series of potential outcomes and observable variables satisfy the following conditions:

(i) \emph{Markovian exclusion}: for $t=2, \dots, T$ and for arbitrary treatment paths $(w_{0:t-2},w_{t-1},w_{t})$ and $(w_{0:t-2}^{\prime},w_{t-1},w_{t})$, where $ w_{0:t-2} \neq w_{0:t-2}^{\prime}$, 
\bea \label{exclu}
Y_t(w_{0:t-2}, w_{t-1},w_{t} ) = Y_t(w_{0:t-2}^{\prime}, w_{t-1},w_{t}) := Y_t(w_{t-1},w_t)
\eea
holds with probability one.

(ii) \emph{Markovian exogeneity}: for $t=1, \dots, T$ and any treatment path $w_{0:t}$,
\bea
Y_t(w_{0:t}) \perp X_{0:t-1} | W_{t-1},
\eea
and for $t=1, \dots, T-1$,
\bea
W_t \perp X_{0:t-1} | W_{t-1}.
\eea
\end{assumption}

These assumptions significantly simplify the dynamic structure of the problem. Markovian exclusion, Assumption \ref{ass:toy_example}(i), says that only the current treatment $W_{t}$ and treatment in the previous period $W_{t-1}$ can have a causal impact on the current outcome. This allows the indices of the potential outcomes to be compressed to the latest two treatments $(w_{t-1},w_t)$, as in (\ref{exclu}). Markovian exogeneity, Assumption \ref{ass:toy_example}(ii), states that once one conditions on the policy implemented in the previous period $W_{t-1}$, the potential outcomes and treatment for the current period are statistically independent of any other past variables. 

It is important to note that these assumptions do not impose stationarity: we allow the distribution of potential outcomes to vary across time periods. In addition, under Assumption \ref{ass:toy_example},
we can reduce the class of policy rules to those that map from $W_{T-1} \in \{0,1 \} $ to $W_T \in \{0,1 \}$, i.e., $$g : \{0 , 1\} \to \{0, 1\}.$$
Consequently, welfare conditional on $\mathcal{F}_{T-1}$ can be simplified to
\bea 
\mathcal{W}_T(g|\mathcal{F}_{T-1})
&=& \E  [ Y_{T}(W_{T-1},1) g(X_{0:T-1})+ Y_{T}(W_{T-1},0)(1- g(X_{0:T-1}))|X_{0:T-1} ]. \nonumber \\
&=& \E [ Y_{T}(W_{T-1},1) g(W_{T-1})+ Y_{T}(W_{T-1},0)(1- g(W_{T-1}))|W_{T-1} ]\nonumber\\ &=:&\mathcal{W}_T(g|W_{T-1}),\label{obj_simple}
\eea
where the first equality follows from the definition of $\mathcal{F}_{T-1}$ and Assumption \ref{ass:toy_example}(i); the second equality follows from Assumption \ref{ass:toy_example}(ii).

To make sense of Assumption \ref{ass:toy_example} and illustrate the relationship between the potential outcome time series and the standard structural equation modeling, we provide a toy example.

\begin{exmp} \label{example_q1_Markov}
Suppose the planner (monetary policy authority) is interested in setting a low or high interest rate at period $T$. Let $W_t$ denote the indicator for whether the interest rate in period $t$ is high ($W_t = 1$) or low ($W_t = 0$). $Y_t$ denotes a measure of social welfare, which can be a function of aggregate output, inflation, and other macroeconomic variables. 
Let $\varepsilon_t$ be an i.i.d.\ shock that is statistically independent of $X_{0:t-1}$, and we assume the following structural equation for the causal relationship of $Y_t$ on $W_t$ (and its lag) and the  dependence of $W_t$ on its lag.\footnote{The distribution of $W_{t}$ follows \cite{hamilton1989new}. However, the latent regime variable in \cite{hamilton1989new} is unobserved, which contrasts with our example where $W_t$ is observed.} 
\begin{align}
Y_{t} & =\beta_0 + \beta_{1}W_{t}+\beta_{2}W_{t-1}+\varepsilon_{t},\label{eq:Yt}\\
W_{t} & =(1-q)+\lambda W_{t-1}+V_{t},\quad \lambda  =p+q-1,\label{eq:Wt}\\
 \varepsilon_t & \perp (W_t, X_{0:t-1}) \qquad \forall 1 \leq t \leq T-1  \qquad \text{and} \qquad  \varepsilon_T  \perp X_{0:T-1}.& \label{eq:et}
 \end{align}
 \begin{align}
\text{If }W_{t-1} & =1,\text{ }\begin{cases}
V_{t}=1-p & \text{with probability }p\\
V_{t}=-p & \text{with probability }1-p.
\end{cases} \label{eq:Vt_1} \\
\text{If }W_{t-1} & =0,\text{ }\begin{cases}
V_{t}=-(1-q) & \text{with probability }q\\
V_{t}=q & \text{with probability }1-q \label{eq:Vt_2}.
\end{cases}
\end{align}
\end{exmp}
Compatibility with Assumption \ref{ass:toy_example} can be seen as follows. Assumption \ref{ass:toy_example}(i) is implied by \eqref{eq:Yt}, where the structural equation of $Y_t$ involves only $(W_t,W_{t-1})$ as the factors of direct cause. Assumption \ref{ass:toy_example}(ii) is implied by \eqref{eq:Yt}, \eqref{eq:et}, and the fact that the distribution of $V_{t}$ depends solely on $W_{t-1}$, i.e., under \eqref{eq:Wt},\eqref{eq:Vt_1}, and \eqref{eq:Vt_2}, we have $\Pr(W_t|\mathcal{F}_{t-1})=\Pr(W_t|W_{t-1}).$\\


To examine the learnability of the optimal policy rule, we further restrict the data-generating process. First, we impose a strict overlap condition on the propensity score.
\begin{assumption}[Strict overlap]  \label{bound1} Let $e_t(w) := \Pr(W_t = 1|W_{t-1} = w)$ be the period-$t$ propensity score. There exists a constant $\kappa \in (0,1/2)$, such that for any $t = 1, 2, \dots, T-1$ and $w\in \{0,1\}$,
        \begin{equation*}
                \kappa \leq e_t(w) \leq 1-\kappa.
        \end{equation*}
\end{assumption}
The next assumption imposes an unconfoundedness condition on observed policy assignment.
\begin{assumption} [Sequential unconfoundedness]\label{unconf} 
        For any $t = 1,2,\dots, T-1$ and $w\in \{0,1\}$,
        \[
        Y_{t}(W_{t-1},w)\perp W_t|X_{0:t-1}.
        \]
\end{assumption}
This assumption states that the treatments observed in the data are sequentially randomized conditional on lagged observable variables. This is a key assumption to make unbiased estimation of the welfare feasible in each period in the sample, as employed in \cite{bojinov2019time} and others. It is worth noting that the above assumption together with Assumption \ref{ass:toy_example}(ii) implies $Y_{t}(W_{t-1},w)\perp W_t|W_{t-1}.$

Under Assumption \ref{ass:toy_example}, we have that, for any measurable function $f$ of the outcome $Y_{t}(W_{0:t})$ and treatment $W_t$, it holds
\bea \label{simple_cond_exp}
\E (f(Y_{t}(W_{0:t}),W_t)|\mathcal{F}_{t-1})=\E (f(Y_{t}(W_{t-1},W_t),W_t)|W_{t-1}) = \E (f(Y_{t},W_t)|W_{t-1}).
\eea

\noindent {\bf Example \ref{example_q1_Markov} continued.} Assumption \ref{bound1} is satisfied if $0<p<1$ and $0<q<1$; Assumption \ref{unconf} is implied by \eqref{eq:Yt} and \eqref{eq:et}.\\

Imposing Assumption \ref{bound1} and \ref{unconf} and assuming propensity scores are known, we consider constructing a sample analogue of \eqref{obj_simple} conditional on $W_{T-1}=w$ based on the historical average of the inverse propensity score weighted outcomes,
\bea \label{sample_simple}
\widehat{\mathcal{W}}(g|W_{T-1}=w) 
= \frac{1}{T(w)}\sum_{1 \leq t \leq T-1: W_{t-1} = w}\left [\frac{Y_{t}W_{t}g(W_{t-1})}{e_t(W_{t-1})}+ \frac{Y_{t}(1-W_{t})\{1- g(W_{t-1})\}}{1- e_t(W_{t-1})} \right],
\eea
where $T(w) = \#\{1\leq t \leq T-1: W_{t-1} = w\} $ is the number of observations where the policy in the previous period took value $w$, i.e., the subsample corresponding to $W_{t-1} = w$. Unlike the microeconometric setting considered in, e.g., \cite{kitagawa2018should}, we do not necessarily have $\widehat{\mathcal{W}}(g|W_{T-1}=w)$ as a direct sample analogue for the planner's social welfare objective, since we allow a nonstationary environment in which the historical average of the conditional welfare criterion can diverge from the conditional welfare in the current period. Nevertheless, we refer to $\widehat{\mathcal{W}}(g|W_{T-1}=w)$ as the empirical welfare of the policy rule $g$.

Denoting $(\cdot|W_{T-1}=w)$ by $(\cdot|w)$, we define the true optimal conditional policy and its empirical analogue to be,
\bea 
g^*(w) &\in& \mbox{argmax}_{g: \{w\}\to  \{0,1\}} \mathcal{W}_T(g|w),\label{best_g} \\
\hat{g}(w) &\in& \mbox{argmax}_{g: \{w\}\to  \{0,1\}} \widehat{\mathcal{W}}(g|w),\label{best_g_hat}
\eea
where $\hat{g}$ is constructed by maximizing empirical welfare over a class of policy rules (four policy rules in total). We call a policy rule constructed in this way the \textit{Time-series Empirical Welfare Maximization} (T-EWM) rule. The construction of the T-EWM rule $\hat{g}$ is analogous to the conditional empirical success rule with known propensity scores considered by \cite{manski2004statistical} in the i.i.d. cross-sectional setting. In the time-series setting, however, the assumptions imposed so far do not guarantee that $\widehat{\mathcal{W}}(g|w)$ is an unbiased estimator of the true conditional welfare $\mathcal{W}_{T}(g|w)$.

\subsection{Bounding the conditional welfare regret of the T-EWM rule} \label{discrete_bound}
A major contribution of this paper is characterizing conditions that justify the T-EWM rule $\hat{g}$ in terms of the convergence of conditional welfare. This section clarifies these points in the context of our illustrative example. 
 
To bound conditional welfare regret, our strategy is to decompose empirical welfare $\widehat{\mathcal{W}}(g|w)$ into a conditional mean component and a deviation from it. The deviation is the sum of a martingale difference sequence (MDS), and this allows us to apply concentration inequalities for sums of MDS.  Define an intermediate welfare function,
\bea \label{E_sum_c}
\bar{\mathcal{W}}(g|w)
=T(w)^{-1}\sum_{1 \leq t \leq T-1 : W_{t-1} = w}\E\left[ Y_{t}(W_{t-1},1)g(W_{t-1})+Y_{t}(W_{t-1},0)\left[1-g(W_{t-1})\right]|W_{t-1}\right]. 
\eea
Under the strict overlap and unconfoundedness assumptions (i.e., Assumptions \ref{bound1} and \ref{unconf}), the difference between empirical welfare and  \eqref{E_sum_c} is a sum of MDS. A concentration inequality for an average of MDS then implies that the empirical welfare concentrates around $\bar{\mathcal{W}}(g|w)$. Since $\bar{\mathcal{W}}(g|w)$ is not guaranteed to inform an optimal policy in terms of $\mathcal{W}_T(g|w)$, we impose the assumption:

 \begin{assumption}  [Invariance of the welfare ordering]  \label{equiv_W}
      Given $w\in \{0,1\}$, let $g^*=g^*(w)$ as defined in \eqref{best_g}.   There exists a positive constant $c>0$, such that for any  $g\in\{0,1\}$,
                \bea \label{eq:invariance_toy}
        \mathcal{W}_T(g^*|w)- \mathcal{W}_T(g|w) \leq c\big[\bar{\mathcal{W}}(g^*|w)-\bar{\mathcal{W}}(g|w)\big],
        \eea
        with probability one, i.e., $P_T \left( \text{inequality (\ref{eq:invariance_toy}) holds} \right) = 1$, 
        where $P_T$ is the probability distribution for $X_{0:T-1}$.
\end{assumption}

Noting that the left-hand side of (\ref{eq:invariance_toy}) is nonnegative for any $g$ by construction and $c>0$, this assumption implies that $\bar{\mathcal{W}}(g^*|w)-\bar{\mathcal{W}}(g|w) > 0$ must hold whenever $\mathcal{W}_T(g^*|w)- \mathcal{W}_T(g|w) > 0$. That is, the optimality of $g^{\ast}$ in terms of the conditional value of welfare at $T$ is maintained in the historical average of the conditional values of welfare. Under this assumption, having an estimated policy $\hat{g}$ that attains a convergence of $\bar{\mathcal{W}}(g^*|w)-\bar{\mathcal{W}}(\hat{g}|w)$ to zero guarantees that $\hat{g}$ is also consistent for the optimal policy $g^{\ast}$ in terms of the conditional welfare at $T$. 

\begin{rem} \label{rem:invar_order}
Assumption \ref{equiv_W} can be restrictive in a situation where the dynamic causal structure of the current period is believed to be different from the past, but is weaker than stationarity. In the current example, Assumption \ref{equiv_W} is implied by the following condition.
        \textit{\begin{description}
                        \item [{A\ref{equiv_W}' }] The stochastic process
        $$S_t(w) \equiv Y_{t}(W_{t-1},1)g(W_{t-1})+Y_{t}(W_{t-1},0)\left[1-g(W_{t-1})\right]|_{W_{t-1}=w}$$ is weakly stationary.
        \end{description}}
   Under A\ref{equiv_W}',  $\E\left[ Y_{t}(W_{t-1},1)g(W_{t-1})+Y_{t}(W_{t-1},0)\left[1-g(W_{t-1})\right]|W_{t-1}=w\right] $
    is invariant for $2\leq t\leq T$. Then, Assumption \ref{equiv_W} will hold naturally.
                
  Furthermore,  Assumption \ref{equiv_W} can be satisfied by many classic nonstationary processes in linear time-series models, including series with deterministic or stochastic time trends.
  \end{rem}

    \noindent{\bf  Example \ref{example_q1_Markov} continued.}\\ 
    $\varepsilon_{t}$ remains an i.i.d. noise in the following settings.\\
(i) By Remark \ref{rem:invar_order}, Assumption \ref{equiv_W} holds
for Example 1 since
$S_{t}(w)=\beta_{0}+\beta_{1}\cdot g+\beta_{2}\cdot w+\varepsilon_{t}$
is weakly stationary.

\noindent(ii) If we replace \eqref{eq:Yt} by
\[
Y_{t}=\delta_{t}+\beta_{1}W_{t}+\beta_{2}W_{t-1}+\varepsilon_{t},
\]
where $\delta_{t}$ is an arbitrary deterministic time trend. The
process $Y_{t}$ is trend stationary (and hence nonstationary), but Assumption
\ref{equiv_W} still holds with $c=1$ since those
deterministic trends are canceled out by differences, i.e.,
\[
\mathcal{W}_{T}(g^{*}|w)-\mathcal{W}_{T}(g|w)=\beta_{1}\left(g^{*}-g\right)=\bar{\mathcal{W}}(g^{*}|w)-\bar{\mathcal{W}}(g|w).
\]

\noindent(iii) If we replace \eqref{eq:Yt} by
\[
Y_{t}=\beta_0 + \beta_{1}W_{t}+\beta_{2}W_{t-1}+\sum_{i=0}^{t}\varepsilon_{i},
\]
the process $Y_{t}$ is nonstationary with stochastic trends, but
Assumption \ref{equiv_W} still holds with $c=1$
since the stochastic trends are canceled out by differences.

\noindent(iv) If we replace \eqref{eq:Yt} by
\[
Y_{t}=\delta_t + \beta_{1,t}W_{t}+\beta_{2,t}W_{t-1}+\varepsilon_{t}
\]
to allow heterogeneous treatment effect, then
\begin{align*}
\mathcal{W}_{T}(g^{*}|w)-\mathcal{W}_{T}(g|w) & =\beta_{1,T}\left(g^{*}-g\right),\\
\bar{\mathcal{W}}(g^{*}|w)-\bar{\mathcal{W}}(g|w) & =\bar{\beta}_{w}\left(g^{*}-g\right),
\end{align*}
where $\bar{\beta}_{w}:=\frac{1}{T(w)}\sum_{1\leq t\leq T-1:W_{t-1}=w}\beta_{1,t}$.
Since $\mathcal{W}_{T}(g^{*}|w)-\mathcal{W}_{T}(g|w)$ is non-negative
by the definition of $g^{*}$ in  \eqref{best_g}, Assumption \ref{equiv_W}
holds if $\beta_{1,T}$ and $\bar{\beta}_{w}$ have the same
sign and $\frac{|\beta_{1,T}|}{|\bar{\beta}_{w}|}\leq c.$

Without loss of generality, we can assume that both $\beta_{1,T}$
and $\bar{\beta}_{w}$ are positive. Then, a sufficient condition for Assumption \ref{equiv_W} is: There
are positive numbers $l$ and $u$, such that $0<l\leq\beta_{1,t}\leq u$
holds for all $t$. In this case, $c=\frac{u}{l}$.

Assumption \ref{equiv_W} implies that $g^*$ also maximizes $\bar{\mathcal{W}}$. Hence, if empirical welfare $\widehat{\mathcal{W}}(\cdot|w)$ can approximate $\bar{\mathcal{W}}(\cdot|w)$ well, intuitively the T-EWM rule $\hat{g}$ should converge to $g^{\ast}$.  The motivation for Assumption \ref{equiv_W} is to create a bridge between $\mathcal{W}_{T}(g^*|w)-\mathcal{W}_{T}(\hat{g}|w)$, the population regret,  and $\widehat{\mathcal{W}}(\cdot|w) - \bar{\mathcal{W}}(\cdot|w) $, which is a sum of MDS  with respect to filtration $\{ \mathcal{F}_{t-1} : t = 1,2, \dots, T \}$.\footnote{Note that by \eqref{simple_cond_exp},  $\E(\cdot|\mathcal{F}_{t-1})=\E(\cdot|W_{t-1} )$.}
 Specifically, 
  \begin{align} \label{ineq_s}
        & \mathcal{W}_{T}(g^*|w)-\mathcal{W}_{T}(\hat{g}|w)\nonumber
        \leq  c\left[\bar{\mathcal{W}}(g^*|w)-\bar{\mathcal{W}}(\hat{g}|w)\right]\nonumber\\
        = & c\left[\bar{\mathcal{W}}(g^*|w)-\widehat{\mathcal{W}}(\hat{g}|w)+\widehat{\mathcal{W}}(\hat{g}|w)-\bar{\mathcal{W}}(\hat{g}|w)\right]\nonumber
        \leq  c\left[\bar{\mathcal{W}}(g^*|w)-\widehat{\mathcal{W}}(g^*|w)+\widehat{\mathcal{W}}(\hat{g}|w)-\bar{\mathcal{W}}(\hat{g}|w)\right]\nonumber\\
        \leq & 2c\sup_{g:\{w\}\to  \{0,1\}}|\bar{\mathcal{W}}(g|w)-\widehat{\mathcal{W}}(g|w)|,
 \end{align}
 where the first inequality follows from Assumption \ref{equiv_W}. The second inequality follows from the definition of T-EWM rule $\hat{g}$ in \eqref{best_g_hat}.

 To bound the right-hand side of \eqref{ineq_s}, define
 \begin{align*}
 \widehat{\mathcal{W}}_t(g|w)=&1(W_{t-1}=w)\left[\frac{Y_{t}W_{t}g(W_{t-1})}{e_t(W_{t-1})}+ \frac{Y_{t}(1-W_{t})\{1- g(W_{t-1})\}}{1- e_t(W_{t-1})}\right],\\
\bar{\mathcal{W}}_t(g|w)=&1(W_{t-1}=w)\E \{ Y_{t}(W_{t-1},1) g(W_{t-1})+ Y_{t}(W_{t-1},0)(1- g(W_{t-1}))|W_{t-1}=w\}.
\end{align*}
 Then, we can express \eqref{sample_simple} and \eqref{E_sum_c} as
 \beaa
 \widehat{\mathcal{W}}(g|w)
 =&\frac{T-1}{T(w)}\cdot \frac{1}{T-1}\sum_{t  = 1}^{T-1}\widehat{\mathcal{W}}_t(g|w),\\
\bar{\mathcal{W}}(g|w)
 =&\frac{T-1}{T(w)}\cdot \frac{1}{T-1}\sum_{t  = 1}^{T-1}\bar{\mathcal{W}}_t(g|w),
 \eeaa
 and
 \bea
 \widehat{\mathcal{W}}(g|w) - \bar{\mathcal{W}}(g|w) = \frac{T-1}{T(w)} \cdot \frac{1}{T-1} \sum_{t=1}^{T-1} \left[ \widehat{\mathcal{W}}_t(g|w) - \bar{\mathcal{W}}_t(g|w) \right] \notag
 \eea
 follows. Next, we impose
 \begin{assumption} [Bounded outcomes]\label{s_bounded_y}
		There exists $M<\infty$ such that the support of outcome variable
	$Y_t$ is contained in $[-M/2,M/2]$.\footnote{The boundedness condition can be replaced by weaker
moment or tail assumptions.
For example,  sub-Gaussian or sub-exponential tails
for the error term, or 
finite higher-order moment conditions.}
\end{assumption}
 \begin{prop} \label{prop:simple_mds}
 	Under Assumptions \ref{ass:toy_example}-\ref{unconf} and \ref{s_bounded_y}, the sequence $\{\widehat{\mathcal{W}}_t(g|w)-\bar{\mathcal{W}}_t(g|w)\}_{t=1}^{T-1}$ is an MDS.
 	\end{prop} 
  The proof of this proposition can be found in Appendix \ref{app:proof_simple_mds}. With Proposition \ref{prop:simple_mds}, we can apply a concentration inequality for sums of MDS to obtain a high-probability bound for $\widehat{\mathcal{W}}(g|w) - \bar{\mathcal{W}}(g|w)$ that is uniform in $g$.  
 \begin{theorem}\label{thm:discrete_bound}
Under Assumptions \ref{ass:toy_example} to   \ref{s_bounded_y}, it holds that
 \begin{equation} \label{simp_rate}
        \sup_{g: \{0,1\}\to \{0,1\}}\E|\widehat{\mathcal{W}}(g|w) - \bar{\mathcal{W}}(g|w)| \leq \frac{C}{\sqrt{T-1}},
 \end{equation}
where $C$ is a constant defined in Appendix \ref{simple_mds_proof}.\end{theorem}
The proof of Theorem \ref{thm:discrete_bound} can be found in Appendix \ref{simple_mds_proof}.  Combining (\ref{ineq_s}) (an implication of Assumption \ref{equiv_W}) and (\ref{simp_rate}), we can conclude that the convergence rate of expected regret $\E[\mathcal{W}_{T}(g^*|w)-\mathcal{W}_{T}(\hat{g}|w)]$ is upper-bounded by $2c\cdot\frac{C}{\sqrt{T-1}}$ uniformly in $w\in \{0,1\}$.
  \begin{rem} {[}Higher Markov orders{]} \label{rem:higher/inf Markov}
	Theorem \ref{thm:discrete_bound}
	presents our main results of welfare regret upper bounds under the simple
	first-order Markovian structure outlined in Assumption \ref{ass:toy_example}.
	We can extend our analysis to a higher or infinite order Markovian structure as summarized below. We defer the detailed discussion  to Appendix \ref{app_multi_markov}.
	
	If current observations can depend causally or statistically on the realized treatment over the preceding $q$ periods (for some $1<q<T$), we can define $\mathcal{W}_{T}(g|w_{T-q:T-1})$, $\widehat{\mathcal{W}}(\hat{g}|w_{T-q:T-1})$, $\bar{\mathcal{W}}(g|w_{T-q:T-1})$, $g^*$, and $\hat{g}$,  similar to the definitions in \eqref{obj_simple}, \eqref{sample_simple}, \eqref{E_sum_c}, \eqref{best_g}, and \eqref{best_g_hat}, respectively. Let $w_{T-q:T-1}\in \{0,1\}^q$ be a realization of the treatment path spanning from time $T-q$ to $T-1$. Given the modified assumptions detailed in Appendix \ref{app_finite_markov}, we can apply similar reasoning as in Theorem \ref{thm:discrete_bound} to establish a convergence rate of $\frac{1}{\sqrt{T-q}}$.
	
	\end{rem}
 \begin{rem}{[}Comparison with \cite{bojinov2019time}{]}\label{Com_B&S}
	The major distinction
	between our work and \cite{bojinov2019time} is that we focus on policy decisions and future welfare, while \cite{bojinov2019time} study estimation and inference on the
	retrospective causal effects. An estimand of interest in \cite{bojinov2019time} is the temporal (zero-lag)
	average treatment effect (ATE) defined as
	\begin{equation}
		\bar{\tau}_{0} :=\frac{1}{T-1}\sum_{t=1}^{T-1}\left[Y_{t}\left(W_{0:t-1},1\right)-Y_{t}\left(W_{0:t-1},0\right)\right].\label{eq:BS_ATE_maintext}
	\end{equation}
	In contrast,  our paper focuses on maximizing 
	\begin{equation}
		\mathcal{W}_{T}(g|\mathcal{F}_{T-1}) =\mathsf{E}\left[ \tau_{T}(\mathcal{F}_{T-1}) g(X_{0:T-1})+Y_{T}(W_{0:T-1},0) |\mathcal{F}_{T-1}\right]
	\end{equation}
	where $\tau_{T}(\mathcal{F}_{T-1})$ is the conditional ATE (CATE) at time $T$,
	\begin{equation}
		\tau_{T}(\mathcal{F}_{T-1})=\mathsf{E}\left[Y_{T}(W_{0:T-1},1)-Y_{T}(W_{0:T-1},0)|\mathcal{F}_{T-1}\right]. \label{eq:T-EWM-ATE_maintext}
	\end{equation}
	
	$\tau_{T}(\mathcal{F}_{T-1})$ is the conditional ATE at an upcoming time
	period of $T$, while $\bar{\tau}_0$ sums up the causal
	effects of the past realized time periods from $0$ to $T-1$. This distinction necessitates different sets of assumptions between our work and \cite{bojinov2019time}. Specifically, the sequential
	unconfoundedness assumption (Assumption \ref{unconf}) is sufficient for unbiased estimation for $\bar{\tau}_0$, whereas it
	falls short for $\tau_{T}(\mathcal{F}_{T-1})$. Consequently,  Assumptions \ref{ass:toy_example}
	(Markov properties) and \ref{equiv_W} (Invariance of welfare ordering) are not imposed in \cite{bojinov2019time}. 
	
	Analogous to the construction of our empirical welfare criterion, we can consider the following estimator for $\tau_{T}(\mathcal{F}_{T-1})$:
	\begin{equation}
		\hat{\tau}_{T}(w)=T(w)^{-1}\sum_{1\leq t\leq T-1:W_{t-1}=w}\left[\frac{Y_{t}W_{t}}{e_{t}(W_{t-1})}-\frac{Y_{t}(1-W_{t})}{1-e_{t}(W_{t-1})}\right].\label{eq:e_ATE_maintext}
	\end{equation}
	To validate $\hat{\tau}_T(w)$ as an estimator for $\tau_{T}(\mathcal{F}_{T-1})$, the crucial
	step is to link the past CATEs (or the welfare in our context) from past periods to the future ones. This
	motivates our introduction of Assumption \ref{equiv_W}, the invariance of welfare ordering. 
	
	
	The construction of $\hat{\tau}_T(w)$ is model-free and 
	selects a subset of past periods that
	share the same conditioning with time $T$. To guarantee the availability of observations sharing the conditioning states, we impose Assumption \ref{ass:toy_example}, limiting the persistence of carryover effects of treatment. 
\end{rem}

  \begin{rem} We construct the empirical welfare $\widehat{\mathcal{W}}(g|W_{T-1} = w)$ in (\ref{sample_simple}) by the average with uniform weights. To generalize, we can specify $\widehat{\mathcal{W}}(g|W_{T-1} = w)$ as a weighted average:
\bea \label{sample_weighted}
\widehat{\mathcal{W}}(g|W_{T-1}=w) 
= \sum_{1 \leq t \leq T-1: W_{t-1} = w} a_t \left [\frac{Y_{t}W_{t}g(W_{t-1})}{e_t(W_{t-1})}+ \frac{Y_{t}(1-W_{t})\{1- g(W_{t-1})\}}{1- e_t(W_{t-1})} \right],
\eea  
where $a_t \geq 0$ is a prespecified weight assigned to period $1 \leq t \leq T-1$. Modifying intermediate welfare $\bar{\mathcal{W}}(g|w)$ accordingly by
\bea \label{E_weighted_sum_c}
\bar{\mathcal{W}}(g|w)
=\sum_{1 \leq t \leq T-1 : W_{t-1} = w}a_t\E\left[ Y_{t}(W_{t-1},1)g(W_{t-1})+Y_{t}(W_{t-1},0)\left[1-g(W_{t-1})\right]|W_{t-1}\right] 
\eea
and imposing Assumption \ref{equiv_W} with the modified intermediate welfare, we can study a condition for the weights that generalizes the welfare regret convergence of Theorem \ref{thm:discrete_bound}. The specification of nonuniform weights can reflect the planner's belief or knowledge on how the period $T$ potential outcome distribution differs from those of the past periods or which past period observations are more informative for the current period decision making. See, e.g., \cite{ishihara2024} for an optimal weighting of pieces of evidence for policy choice when the population in which the policy is implemented  differs from the populations in which pieces of evidence were collected. We leave a formal investigation for the current time-series setting for future research.
\end{rem}

\subsection{Infinite Markov order}
\label{sec:infi_markov}
Time-series models commonly used in macroeconomic policy analysis imply an infinite order Markovian structure. For example, we consider the following modification to  Example \ref{example_q1_Markov}.
\begin{exmp} \label{Example_inf_Markov}
Replace \eqref{eq:Yt} in  Example \ref{example_q1_Markov} with a structural MA($\infty$) model: 
\begin{equation}
Y_{t} =\alpha+\sum_{i=0}^{\infty}\beta_{i}W_{t-i}+\sum_{i=0}^{\infty}\gamma_{i}\varepsilon_{t-i},\label{eq:Yt_inf}
\end{equation}
while keeping the conditions \eqref{eq:Wt} to \eqref{eq:Vt_2} unchanged. 
This type of $\text{MA}(\infty)$ process underlies the causal impulse response analysis of structural vector autoregressions; see, e.g., \cite{kilian2017structural}.  For the effect of shocks to diminish, both $\gamma_i$ and $\beta_i$ must decay in absolute value with $i\to \infty$. In particular, if $Y_t$ is an AR(1) process, then $\beta_i$ and $\gamma_i$ are powers of the AR coefficient.
\end{exmp}

Without requiring stationarity or functional form restrictions, we can extend our framework to infinite Markovian order and obtain convergence of welfare regret conditional on a treatment path of infinite length, $\mathcal{W}_{T}(g^{*}|w_{-\infty:T-1})-\mathcal{W}_{T}(\hat{g}|w_{-\infty:T-1})$.
Let $1< m < T$ and $\hat{g}$ be a policy that maximizes the empirical welfare $\widehat{\mathcal{W}}(g|w_{T-m:T-1})$ with conditioning policy path truncated to $w_{T-m:T-1}$. Appendix \ref{rem: Inf_order} shows the following upper bound for the welfare regret of infinite order Markov models:
\begin{align*}
	\mathcal{W}_{T}(g^{*}|w_{-\infty:T-1})-\mathcal{W}_{T}(\hat{g}|w_{-\infty:T-1})& \leq 2c\sup_{g:\{w_{T-m:T-1}\}\to\{0,1\}}|\bar{\mathcal{W}}(g|\mathcal{F}_{t-1})-\widehat{\mathcal{W}}(g|w_{T-m:T-1})|\\
	& +2\cdot \widetilde{\text{w-bias}}_{\infty}\left(m\right),
\end{align*}
where $\bar{\mathcal{W}}(g|\mathcal{F}_{t-1})$ is the intermediate welfare defined by \eqref{eq:wel_bar_F_m} in Appendix \ref{rem: Inf_order}, around which the empirical welfare is expected to concentrate, and $\widetilde{\text{w-bias}}_{\infty}\left(m\right)$ is the welfare bias due to truncation of the empirical welfare, which is defined by \eqref{eq:inf_markov_bias} in Appendix \ref{rem: Inf_order}. 
The first term on the right-hand side represents the average of an MDS. Under the regularity conditions presented in Appendix \ref{rem: Inf_order}, we can show that this term converges at a rate of $\frac{1}{\sqrt{T-m}}$. For the second term on the right-hand side, we shall have $\text{plim}_{m\to\infty}\widetilde{\text{w-bias}}_{\infty}\left(m\right)=0$ under additional conditions that ensure the decay of temporal dependence. See Appendix \ref{rem: Inf_order} for more details. Furthermore, in Appendix \ref{App:exmple_2}, we show that the infinite order T-EWM proposed in Appendix \ref{rem: Inf_order} can be applied to  the policy choice problem specified in Example \ref{Example_inf_Markov}.
\section{Continuous covariates} \label{continuous}
This section extends the illustrative example of Section \ref{model_example} by allowing $X_t$ to contain continuous variables. We first introduce the continuous setting in Section \ref{sec:conti_setting}. Then,  Section \ref{sec:conti_kernel} presents a natural approach for handling the continuous policy variable, which employs a kernel function to construct an analogue of the  conditional empirical welfare function \eqref{sample_simple} presented in Section \ref{model_example}. 
Finally, having introduced the unconditional welfare in the time-series setting in Section \ref{sec:notation_timing_welfare}, we formally present the \emph{unconditional} T-EWM method in Section \ref{sec:uncon}, developed in parallel with the kernel-based approach. 




\subsection{Setting} \label{sec:conti_setting}

In addition to $(Y_t,W_t)$, we incorporate general covariates $Z_t$ and $M_t$ into $X_t \in \mathcal{X}$, which can be continuous. Now, $X_t=(W_t,M_t^{\prime}, Y_t, Z_t^{\prime})^{\prime}$.  For simplicity of exposition, we maintain the first-order Markovian structure similarly to the illustrative example, while modifying Assumptions \ref{ass:toy_example}, \ref{bound1}, and \ref{unconf} as follows. It is straightforward to incorporate a higher-order Markovian structure. 

\begin{assumption}  [Markov properties] \label{ass:continuous_Markov}The time series of potential outcomes and observable variables satisfy the following conditions:

(i) \emph{Markovian exclusion}: the same as Assumption \ref{ass:toy_example}(i).

(ii) \emph{Markovian exogeneity}: for $t=1, \dots, T$ and any treatment path $w_{0:t}$,
\bea
Y_t(w_{0:t}) \perp X_{0:t-1} | X_{t-1},
\eea
and for $t=1, \dots, T-1$,
\bea
W_t \perp X_{0:t-1} | X_{t-1}.
\eea
\end{assumption}
Similarly to (\ref{obj_simple}) in the illustrative example, Assumption \ref{ass:continuous_Markov} implies that we can reduce the conditioning information of $\mathcal{F}_{t-1}$ to only $X_{t-1}$ and reduce the policy to a binary map of $X_{t-1}$ without any loss of conditional welfare. 
Following these reductions and considering the planner's focus on the policy choice in period $T$, we can formulate the planner's objective function as follows:
\bea \label{conti_obj_kernel}
\mathcal{W}_{T}(g|X_{T-1}) 
=\E\left[ Y_{T}(W_{T-1},1)g(X_{T-1})+Y_{T}(W_{T-1},0)[1-g(X_{T-1})]|X_{T-1}\right].
\eea
Note that under Assumption \ref{ass:continuous_Markov},  the policy function is reduced to $g$: $\mathcal{X} \to \{0,1\}$.
We assume the strict overlap and unconfoundedness restrictions under the general covariates as follows.
\begin{assumption}[Strict overlap]\label{bound_c}
         Let $e_t(x)=\Pr(W_{t}=1|X_{t-1}=x)$ be the propensity score at time $t$. There exists $\kappa \in (0,1/2)$, such that
        \begin{equation*}
                \kappa \leq e_t(x) \leq 1-\kappa
        \end{equation*}
        holds for every $t = 1, \dots, T-1$ and each $x\in\mathcal{X}$.
        \end{assumption}
{ 
\begin{assumption}[Unconfoundedness] \label{unconf_c} 
        For any $t=1,2,\dots,T-1$ and $w\in \{0,1\}$,
\[
Y_{t}(W_{t-1},w)\perp W_t|X_{0:t-1}.
\]
\end{assumption}

Under Assumptions \ref{ass:continuous_Markov} and \ref{unconf_c}, we can generalize (\ref{simple_cond_exp}) by including the set of covariates in the conditioning variables: for any measurable function $f$, 
\bea \label{equiv_markov}
\E(f(Y_t(W_{0:t}),W_t)| \mathcal{F}_{t-1})=\E(f(Y_t(W_{t-1},W_t),W_t)|X_{t-1})= \E (f(Y_t,W_t)|X_{t-1}). 
\eea 

Similar to our approach with the simple model described in Section \ref{model_example},  we assume that $Y_t$, which is reintroduced in Section \ref{sec:conti_setting}, has a bounded support.

\begin{assumption}[Bounded outcome] \label{ass bounded y}
	
	There exists $M<\infty$ such that the support of outcome variable
	$Y_t$ is contained in $[-M/2,M/2]$.
	
\end{assumption}

\subsection{Kernel approach for conditional welfare} \label{sec:conti_kernel}
For continuous conditioning covariates $X_{T-1}$, a simple sample analogue of the objective function is not available due to the lack of multiple observations at any single conditioning value of $X_{T-1}$. One approach is to use nonparametric smoothing to construct an estimate for conditional welfare. 

For simplicity, we let $X_{T-1}\in \mathbb{R}$. Then, with a kernel function $K(\cdot)$ and a bandwidth $h$, the empirical conditional welfare can be rewritten as\footnote{If the set of conditioning variables $X_{T-1}$ contains both continuous and discrete components, we can adopt a hybrid method to construct a valid sample analogue combining kernel-smoothing (for continuous variables) and subsamples (for discrete variables). In this section, we focus on the case where the target welfare function is conditional on a univariate continuous variable.}
{\small \bea \label{eq:sample_kernel}
        \widehat{\mathcal{W}}(g|x)
        =\frac{\sum_{t=1}^{T-1}K_h(X_{t-1},x)\widehat{\mathcal{W}}_t(g)}{\sum_{t=1}^{T-1}K_h(X_{t-1},x)},
        \eea}
where $\widehat{\mathcal{W}}_t(g) := \frac{Y_{t}W_{t}}{e_t(X_{t-1})} g(X_{t-1})+ \frac{Y_{t}(1-W_{t})}{1- e_t(X_{t-1})}[1-g(X_{t-1})]$, and $K_h(a,b):=\frac{1}{h}K(\frac{a-b}{h})$.

Using the same notation, we can define $\mathcal{W}_{t}(g|x) 
=\E[ Y_{t}(W_{t-1},1)g(X_{t-1})+Y_{t}(W_{t-1},0)[1-g(X_{t-1})]|X_{t-1}=x]$. 
Then, similar to \eqref{best_g} and \eqref{best_g_hat}, we define
\beaa
g_x^* &\in& \mbox{argmax}_{g}\mathcal{W}_T(g|x),\\
\hat{g}_x &\in& \mbox{argmax}_{g}\widehat{\mathcal{W}}(g|x),
\eeaa
to be the maximizers of $\mathcal{W}_T(g|x)$  and $\widehat{\mathcal{W}}(g|x)$. Moreover, 
define an intermediate welfare function:
\bea \label{eq:kernel_bar}
\bar{\mathcal{W}}_{h}(g|x)
&=&\frac{\sum_{t=1}^{T-1}K_h(X_{t-1},x)\mathcal{W}_t(g|x)}{\sum_{t=1}^{T-1}K_h(X_{t-1},x)}.
\eea

The invariance of welfare ordering assumption is modified to:
\begin{assumption}[Invariance of ordering for the conditional welfare] \label{equiv_W_c}
        For any  $g:\mathcal{X}\to \{0,1\}$ and $x\in \mathcal{X}$, there exists some positive constant $c$, such that
        \beq \label{uncon_con_kernel}
        \mathcal{W}_{T}(g_x^*|x)- \mathcal{W}_{T}(g|x) \leq c\big[\bar{\mathcal{W}}_h(g_x^*|x)-\bar{\mathcal{W}}_h(g|x)\big].
        \eeq
\end{assumption}
Similar to Assumption  \ref{equiv_W}, Assumption  \ref{equiv_W_c} holds if the stochastic process $S_t(x):=Y_{t}(W_{t-1},1)$ $g(X_{t-1})+Y_{t}(W_{t-1},0)[1-g(X_{t-1})]|_{X_{t-1}=x}$ is weakly stationary.

Let $P_T$ be a joint probability distribution of a sample path of length $(T-1)$, $\mathcal{P}_T(M,\kappa)$ be the class of $P_T$, which satisfies Assumptions \ref{ass:continuous_Markov} to \ref{equiv_W_c} and \ref{kernel}  specified in Appendix \ref{max_kernel}; let ${\E}_{P_{T}}$ be the expectation taken over different realizations of random samples. The following theorem shows an upper bound for conditional regret in the one-dimensional covariate case (i.e., $X_t\in \mathbb{R}^1$). This result can be readily extended to the multiple-covariate case.  
\begin{theorem}  \label{thm: kernel}
Under Assumptions \ref{ass:continuous_Markov} to \ref{equiv_W_c} and \ref{kernel}  specified in Appendix \ref{max_kernel}, 
        \begin{equation*}
        \underset{P_T\in\mathcal{P}_T(M,\kappa)}{\sup}\E{_{P_T}}[\mathcal{W}_{T}(g_x^*|x)-\mathcal{W}_{T}(\hat {g}_x|x)]  \leq c_1\left[\left(\sqrt{(T-1)h}\right)^{-1} +h^{2}\right].
\end{equation*}
\end{theorem}
Setting $h = O(T^{-1/5})$, the right-hand side bound becomes $O(T^{-2/5})$. A proof is presented in Appendix \ref{max_kernel}.

\subsection{Unconditional T-EWM}\label{sec:uncon}
In this section, we shift the focus to maximizing unconditional welfare. For ease of exposition with the unconditional T-EWM, given a policy function $g:\mathcal{X}^{T} \to \{0,1\}$, we define the corresponding region in the space of the covariate vector for which the decision rule chooses $W_T =1$ to be 
 \beq \label{gG}
G = \{ X_{0:T-1} : g(X_{0:T-1}) = 1 \} \subset \mathcal{X}^{T}.
\eeq
We refer to $G$ as a \emph{decision set}. Under Assumption \ref{ass:continuous_Markov}, the decision set is reduced to $G = \{ X_{T-1} : g(X_{T-1}) = 1 \} \subset \mathcal{X}$, and  the unconditional welfare under policy $G$ is
\begin{equation}\label{wel:uncondi}
        \mathcal{W}_{T}(G) =\mathsf{E}\left[ Y_{T}(W_{T-1},1)1(X_{T-1}\in G)+Y_{T}(W_{T-1},0)1(X_{T-1}\notin G)
        \right].
\end{equation}
The sample analogue of $\mathcal{W}_T(G)$ can be expressed as
\bea \label{eq:unconditional_sample}
\widehat{\mathcal{W}}(G)
=\frac{1}{T-1}\sum_{t=1}^{T-1} \left[\frac{Y_{t}W_{t}}{e_t(X_{t-1})}1(X_{t-1}\in G)+ \frac{Y_{t}(1-W_{t})}{1- e_t(X_{t-1})}1(X_{t-1}\notin G) \right].
\eea
For $\mathcal{G}$ denoting a class of decision sets, we define
\bea 
G_*&\in& \mbox{argmax}_{G\in\mathcal{G}}\mathcal{W}_{T}(G),\label{maxG}\\ 
\hat{G}  &\in&  \mbox{argmax}_{G\in\mathcal{G}}\widehat{\mathcal{W}}(G).\label{sample:maxG}
\eea
 In addition, define two intermediate welfare functions, 
\begin{align} \label{2intermediate_welfare}
        \bar{\mathcal{W}}(G) &=\frac{1}{T-1}\sum_{t=1}^{T-1}\mathsf{E}\left[ Y_{t}(W_{t-1},1)1(X_{t-1}\in G)+Y_{t}(W_{t-1},0)1(X_{t-1}\notin G)
        |\mathcal{F}_{t-1}\right], \nonumber \\
        \widetilde{\mathcal{W}}(G) &=\frac{1}{T-1}\sum_{t=1}^{T-1}\mathsf{E}\left[ Y_{t}(W_{t-1},1)1(X_{t-1}\in G)+Y_{t}(W_{t-1},0)1(X_{t-1}\notin G)\right].
\end{align}
The need for an additional intermediate welfare function, $\widetilde{\mathcal{W}}(G)$, arises because we aim to bound the unconditional welfare. After centering the empirical welfare around its conditional mean, $\bar{\mathcal{W}}(G) $, we are left with another difference, $\bar{\mathcal{W}}(G)-\widetilde{\mathcal{W}}(G)$, which we will control below.

To obtain a regret bound for unconditional welfare, Assumption \ref{equiv_W} is modified to:
\begin{assumption} [Invariance of ordering for the unconditional welfare]\label{equiv_W_uc}
For any $G\in \mathcal{G}$, there exists some constant $c$, such that
\beq \label{uncon_con2}
\mathcal{W}_{T}(G_*)- \mathcal{W}_{T}(G) \leq c\big[\widetilde{\mathcal{W}}(G_*)-\widetilde{\mathcal{W}}(G)\big]
\eeq
holds with probability one, i.e., $P_T \left( \text{inequality (\ref{uncon_con2}) holds} \right) = 1$, 
where $P_T$ is the probability distribution for $X_{0:T-1}$.
\end{assumption}
Note that by Assumption \ref{equiv_W_uc}, we have
\begin{equation} \label{welfare_inequality}
        \mathcal{W}_{T}(G_{*})-\mathcal{W}_{T}(\hat{G})
         \leq c\left[\mathcal{\widetilde{W}}(G_{*})-\mathcal{\widetilde{W}}(\hat{G})\right]
         \leq 2c\sup_{G\in \mathcal{G}}\left|\widehat{\mathcal{W}}(G)-\mathcal{\widetilde{W}}(G)\right|,
\end{equation}
where the first inequality follows from \eqref{uncon_con2}, and the second inequality follows from an argument similar to the one for \eqref{ineq_s}.

Note that $\widehat{\mathcal{W}}(G)-\widetilde{\mathcal{W}}(G)$
is \emph{not} a sum of MDS. Instead, it can be decomposed as
\bea \label{decompose_uncon}
\widehat{\mathcal{W}}(G)-\widetilde{\mathcal{W}}(G) & =& \bar{\mathcal{W}}(G)-\widetilde{\mathcal{W}}(G) +( \widehat{\mathcal{W}}(G)-\bar{\mathcal{W}}(G))
= I+II,
\eea
where $I:=\bar{\mathcal{W}}(G)-\widetilde{\mathcal{W}}(G)$ and $II:=\widehat{\mathcal{W}}(G)-\bar{\mathcal{W}}(G)$. Subject to assumptions specified later, Theorem \ref{thm:mds_bound} below shows that $II$, which is a sum of MDS,  converges at $\frac{1}{\sqrt{T-1}}$-rate, and Theorem \ref{mean_bound} below shows that $I$ converges at the same rate.

\eqref{decompose_uncon} reveals that our proof strategy is considerably more complicated than the proof for the EWM model with i.i.d. observations of \cite{kitagawa2018should}, although the rates are similar. Specifically, we need to derive a  bound for the tail probability of the sum of martingale difference
sequences. In addition, we need to handle complex functional classes induced by nonstationary processes.  
For the EWM model, the main task is to show the convergence rate of a sample analogue of $II$, which can be achieved with standard empirical process theory for i.i.d.\  samples. In comparison, we not only have to treat our $II$ more carefully due to time-series dependence, but we also have to deal with $I$.

\subsubsection{Bounding II}

Define empirical welfare at time $t$ and its population conditional expectation as follows, 
\begin{align*}
\widehat{\mathcal{W}}_t(G)&=\frac{Y_{t}W_{t}}{e_t(X_{t-1})}1(X_{t-1}\in G)+ \frac{Y_{t}(1-W_{t})}{1- e_t(X_{t-1})}1(X_{t-1}\notin G),\\
\bar{\mathcal{W}}_t(G)&=\mathsf{E}\left[ Y_{t}(W_{t-1},1)1(X_{t-1}\in G)+Y_{t}(W_{t-1},0)1(X_{t-1}\notin G)|\mathcal{F}_{t-1}\right].
\end{align*}
Then, we examine two summations:
\beaa
	\widehat{\mathcal{W}}(G)=\frac{1}{T-1}\sum_{t=1}^{T-1}\widehat{\mathcal{W}}_t(G), \quad
	\bar{\mathcal{W}}(G) =\frac{1}{T-1}\sum_{t=1}^{T-1} \bar{\mathcal{W}}_t(G).
	\eeaa

For each $t=1, \dots, T-1$, define a function class indexed by $G \in \mathcal{G}$,
\begin{equation} \label{H_class}
	\mathcal{H}_t=\{h_t(\cdot;G)=\widehat{\mathcal{W}}_t(G)-\bar{\mathcal{W}}_t(G):G\in\mathcal{G}\},
\end{equation}
where the arguments of the function $h_t(\cdot;G)$ are $Y_{t}$, $W_t$, and $X_{t-1}$. In the following, we use $n$ to represent the number of summands since the endpoints of summation may vary across different settings in this section, subsequent sections, and appendices. For example, in the case of multi-period welfare functions, the endpoint of a sample is no longer fixed at $T-1$. Given the class of functions $\mathcal{H}_t$, we consider a martingale difference array $\{h_t(Y_{t},W_t,X_{t-1};G)\}_{t=1}^{n}$,
and denote its average by
$${\E}_{n} h\defeq \frac{1}{n}\sum_{t=1}^{n} h_t(Y_{t},W_t,X_{t-1}; G),$$
where $h \defeq \left\{ h_{1}(\cdot; G),h_{2}(\cdot;G),\dots,h_{n}(\cdot;G)\right\} $, and we suppress $n$ and $G$ if there is no confusion in the context. 

 Since we do not restrict $X_t$ to be stationary, we shall handle a vector of function classes that possibly vary over $t$. To this end, we define the following set of notations. Let $H_t$ denote the envelope for the function class $\mathcal{H}_t$, and $\overline{H}_{n}=(H_{1},H_{2},\cdots,H_{n})^{\prime}$, and $\mathbf{H_{n}}=\mathcal{\mathcal{H}}_{1}\times \mathcal{\mathcal{H}}_{2}\times \dots \times \mathcal{H}_{n}.$ 
 For a function $f$ supported on $\mathcal{X}$, define $\|f\|_{Q,r} \defeq (\int_{x\in\mathcal{X}} |f(x)|^r dQ(x))^{1/r} $, and for an $n$-dimensional vector $v=\left\{ v_{1},\dots,v_{n}\right\} $,
its $l_{2}$ norm is denoted by $|v|_{2}\overset{\text{def}}{=}\left(\sum_{i=1}^{n}v_{i}^{2}\right)^{1/2}$.
The covering number of a function class $\mathcal{H}$ w.r.t. a metric $\rho$ is denoted by $\mathcal{N}(\varepsilon, \mathcal{H}, \rho(.))$. For two series of functions $f=\{f_t
\}_{i=1}^n$ and $g=\{g_t
\}_{i=1}^n$, define the metrics $\rho_{2,n}(f,g) = (n^{-1}\sum_t|f_t-g_t|^2)^{1/2}$ and   $\sigma_n(f,g) = \left( n^{-1} \sum_t \E[(f_t-g_t)^2| \mathcal{F}_{t-1}]\right)^{1/2}$. Let $\alpha_{n}$ denote an $n$-dimensional vector in $\mathbb{\mathbb{R}}^{n}$ and $\circ$ denote the element-wise product. 
In the next assumption, we want to bound the covering number of $\mathcal{N}(\delta|\overline{H}{}_{n}|_{2}, \mathbf{H_{n}}, \rho_{2,n})$ by the covering number of all its one-dimensional projection. 
 Now, let $A\lesssim B$ denote that there exists some constant $c_0$ such that $A \leq c_0\cdot B$.
\begin{assumption}[Function classes]\label{cover1}
Let $n=T-1$. For  any discrete measures $Q$, any  $\alpha_{n}\in\mathbb{R}_{+}^{n}$, and all $\delta>0$, we have 
\begin{align}\label{cover}
 	\mathcal{N}(\delta|\tilde{\alpha}_{n}\circ\overline{H}{}_{n}|_{2},\mathbf{\tilde{\alpha}_{n}\circ H_{n}},|.|_{2})
\leq \max_t\sup_Q \mathcal{N} ( \delta \|H_t\|_{Q,2}, \mathcal{H}_t, \|.\|_{Q,2}) \lesssim K (v+1) (4e)^{v +1}(\frac{2}{\delta})^{crv},
\end{align}
where $K$, $v$, $c$, and $e$ are positive constants; $r$ is a positive integer and $\tilde{\alpha}_{n,t}=\frac{\sqrt{\alpha_{n,t}}}{\sqrt{\sum_{t}\alpha_{n,t}}}$.
\end{assumption}

Assumption \ref{cover1} restricts the complexity of the function class to be of polynomial discrimination type, and the complexity index $v$ appears in the derived regret bounds. See Appendix \ref{justify_entropy} for a justification for this assumption.

\begin{assumption}[Empirical sum]\label{norm}
{
There exists a constant $L>0$ such that  $\Pr (\sigma_n(f,g)/$ $ \rho_{2,n}(f,g) >L) \to 0$ as $n \to \infty$.}
Also, $\Pr (\left( n^{-1} \sum_t \E[(f_t-g_t)^2| \mathcal{F}_{t-2}]\right)^{1/2}/ \rho_{2,n}(f,g) >L) \to 0$ as $n \to \infty$.
\end{assumption}

$\rho_{2,n}(f,g)^2$ is the quadratic variation difference and $\sigma_n(f,g)^2$ is its conditional equivalent. It is evident that $\rho_{2,n}(f,g)^2- \sigma_n(f,g)^2$ involves  martingale difference sequences. In the special case of i.i.d. observations, $\sigma_n(f,g)^2$ is equivalent to the sample average of unconditional expectations. Assumption \ref{norm} can thus be viewed as specifying that $\rho_{2,n}(f,g)^2$ and $\sigma_n(f,g)^2$ are asymptotically equivalent in a probability sense. A similar condition can be seen, for example, in Theorem 2.23 of \cite{hall2014martingale}.

Now, let $A\lesssim_p B$ denote $A=O_p(B)$. Then, we have for $II$:

\begin{theorem}\label{thm:mds_bound}
Under Assumptions \ref{ass:continuous_Markov} to \ref{ass bounded y}, and \ref{cover1} to \ref{norm}, 
\begin{equation*}
\sup_{G\in \mathcal{G}}|\widehat{\mathcal{W}}(G)-\bar{\mathcal{W}}(G)|
\lesssim_p C\sqrt{\frac{v}{T-1}},
\end{equation*}
where $C$ is a constant that depends only on $M$ and $\kappa$.
\end{theorem}
The proof of Theorem \ref{thm:mds_bound} is presented in Appendix \ref{proof_MDS_t}.

\subsubsection{Bounding I}

Here, we complete the process of bounding unconditional regret. Let us define
 \beqq
S_{t}(G) = \begin{array}{c}
	Y_{t}(W_{t-1},1)1(X_{t-1}\in G)
		+Y_{t}(W_{t-1},0)1(X_{t-1}\notin G),
	\end{array}
\eeqq
and
\beq \label{eq:bar_St}
\overline{S}_{t}(G)= \E(S_{t}(G)| \mathcal{F}_{t-1}) - \E(S_{t}(G)| \mathcal{F}_{t-2}),
\eeq
\beq \label{eq:tilde_St}
\tilde{S}_{t}(G) =  \E(S_{t}(G)| \mathcal{F}_{t-2}) -  \E(S_{t}(G)).
\eeq
We can apply a similar technique in Theorem \ref{thm:mds_bound} (See Lemma \ref{bound} in Appendix \ref{proof_MDS_t}) to bound the sum of $\overline{S}_{t}(G)$. The second term $\tilde{S}_{t}(G)$ is handled below.
Recalling \eqref{2intermediate_welfare} and \eqref{decompose_uncon}, we have $I=\frac{1}{T-1}\sum_{t=1}^{T-1}\tilde{S}_{t}(G)+ \frac{1}{T-1}\sum_{t=1}^{T-1}\overline{S}_{t}(G)$. Define the function classes,
\beaa
\overline{\mathcal{S}}_{t} &=&\{ f_{t} = \E(S_{t}(G)| \mathcal{F}_{t-1}) - \E(S_{t}(G)|\mathcal{F}_{t-2}): G \in \mathcal{G}  \},\\
\tilde{\mathcal{S}}_{t} &=&\{ f_{t} = \E(S_{t}(G)| \mathcal{F}_{t-2}) - \E(S_{t}(G)): G \in \mathcal{G}  \}.
\eeaa
Note that by Assumption \ref{ass:continuous_Markov}, we have $\E(S_{t}(G)| \mathcal{F}_{t-1}) = \E(S_{t}(G)| X_{t-1})$, and $\E(S_{t}(G)| \mathcal{F}_{t-2})=\E(S_{t}(G)| X_{t-2})$.
\begin{definition}  Let $\{\vps_t\}_{t=-\infty}^{\infty}$ be a sequence of i.i.d.\ random variables, and $\{g_t\}_{t=-\infty}^{\infty}$ is a sequence of measurable functions of $\vps$'s, which might vary with time $t$. For a process $\xi_{.}\defeq\{\xi_t\}_{t=-\infty}^{\infty}$ with $\xi_t \defeq g_t(\vps_t, \vps_{t-1},\cdots)$ and integers $l,q\geq 0$,  
we define the dependence adjusted norm for an arbitrary process $\xi_t $ as
\beq
\theta_{\xi,q} =  \sum_{l = 0}^{\infty} \max_t  \|\xi_t - \xi_{t,l}^*\|_{q},
\eeq
where $\|\cdot\|_q$ denotes $(\E|\cdot|^q)^{1/q}$, and $\xi_{t,l}^*=g_t(\vps_t,\cdots, \vps^{\prime}_{t-l},\cdots)$ is the random variable $\xi_{t}$ with its $l$-th lag replaced by $\varepsilon^{\prime}_{t-l}$,  an independent copy of $\varepsilon_{t-l}$. The sub-exponential/Gaussian dependence adjusted norm is given by:
\beq
    \Phi_{\phi_{\tilde{v}}} (\xi_{.})= \sup_{q\geq 2}(\theta_{\xi,q}/q^{\tilde{v}}),\label{eq:DA_norm}
\eeq
where  $\tilde{v} = 1/2$ (resp. 1) corresponds to the case that the process $\xi_{i}$ is sub-Gaussian (resp. sub-exponential). \end{definition}
\begin{assumption}
[Data generating process] \label{D.1} $X_t={g}{}_{t}(\vps_t, \vps_{t-1}, \cdots),$ where $\vps_t$ is a sequence of i.i.d. random variables, and $g_t$  is a measurable function of $\vps$'s. 
\end{assumption}

{It shall be noted that Assumption \ref{D.1} implies  $\tilde{S}_{t}(G) = \tilde{g}_{t}(\vps_t,\vps_{t-1}, \cdots )$, where $\tilde{g}$ is another measurable function of $\vps$'s.}

\begin{assumption}
\label{D.2}
(i) (Markov exogeneity for $Z$) $Z_t \perp X_{0:t-1} | X_{t-1}$.
(ii) (Envelope functions)\footnote{
Under Assumption~\ref{ass bounded y}, Assumption~\ref{D.2}(ii) is
satisfied with $\tilde{F}_{t}\equiv M$. We state it as a separate
assumption to make the envelope's role in the proofs of
Theorem~\ref{thm:mds_bound} and Theorem~\ref{mean_bound} explicit.}
$ \tilde{S}_{t}(G) $ has an envelope $\tilde{F}_{t}(\cdot)$, 
i.e., $\sup_{G\in \mathcal{G}}  |\E(S_{t}(G)| X_{t-2}=x) -  \E(S_{t}(G))|\leq |\tilde{F}_t(x)|$ for every $t$  and $x\in\mathcal{X}$; for every $t$ and any discrete measures $Q$, it holds that $\mbox{max}_t\mbox{sup}_Q     \|\tilde{F}_{t}\|_{Q,2}< \infty$.
\end{assumption}

The next assumption pertains to the tail of the underlying innovations of time series.
   \begin{assumption}[Tail assumption] \label{D.3}For $\tilde{v} =1/2$ or $1$, it holds that $\sup_{G\in \mathcal{G}}\Phi_{\phi_{\tilde{v}}}(\tilde{S}_{.}(G))   < \infty.$  
   \end{assumption}
   

}
        \begin{assumption}[Further function classes]
        \label{D.4}  Suppose that $F_{t}$ (resp. $\tilde{F}_{t}$) is the envelope of the function class $ \overline{\mathcal{S}_t}$ (resp. $\tilde{\mathcal{S}_t}$).
        Define $\overline{F}_n = (F_{1}, F_2, \cdots,F_{n})$ (resp. $\overline{\tilde{F}}_n = (\tilde{F}_{1}, \tilde{F}_2, \cdots,\tilde{F}_{n})$), and $\mathbf{F}_{n} = \{\overline{\mathcal{S}}_{1},\overline{\mathcal{S}}_{2}, \cdots, \overline{\mathcal{S}}_{n}\}$ (resp. $\tilde{\mathbf{F}}_{n} = \{\tilde{\mathcal{S}}_{1},\tilde{\mathcal{S}}_{2}, \cdots, \tilde{\mathcal{S}}_{n}\}$).
 Let $Q$ denote a discrete measure over a finite number of $n$ points. For $n=T-1$  and all $\delta>0$,
{
there exist positive constants $v$, $V$, and $c$, such that, }
\beaa
\mathcal{N}(\delta|\tilde{\alpha}_{n}\circ\overline{F}_{n}|_{2},\mathbf{\tilde{\alpha}_{n}\circ F_{n},|.|_{2}})&\leq&\max_t \sup_Q \mathcal{N}(\delta \|F_{t}\|_{Q,2}, \overline{\mathcal{S}}_t, \|.\|_{Q,2}) \lesssim (1/\delta)^{cV},\\
\mathcal{N}(\delta|\tilde{\alpha}_{n}\circ\overline{\tilde{F}}{}_{n}|_{2},\mathbf{\tilde{\alpha}_{n}\circ \tilde{F}_{n},|.|_{2}})&\leq&\max_t \sup_Q \mathcal{N}(\delta \|\tilde{F}_{t}\|_{Q,2}, \tilde{\mathcal{S}}_t, \|.\|_{Q,2}) \lesssim (1/\delta)^{cv}.
\eeaa
\end{assumption}
Assumption \ref{D.1} imposes that the time series $\tilde{S}_{t}(G)$ and $X_t$ can be expressed as measurable functions of i.i.d.\ innovations $\vps_t$. Assumption \ref{D.2}(i) completes the Markov exogeneity (Assumption \ref{ass:continuous_Markov}) in the presence of the covariate $Z_t$, which is introduced at the beginning of Section \ref{continuous}.  Assumption \ref{D.2}(ii) is a standard envelope assumption on the function class, and it states that the function  of interest is enveloped by a function of $X_{t-2}$. 
Assumption \ref{D.3} implies that $\Phi_{\phi_{\tilde{v}}}(\tilde{S}_.(G))< \infty$  for $\tilde{v} = 1/2$ or 1. Assumption \ref{D.4} restricts the complexity of the function classes. Based on these assumptions, we have the following rate:


\begin{theorem} \label{mean_bound}

Under Assumptions \ref{ass:continuous_Markov} to \ref{ass bounded y}, and  \ref{cover1} to \ref{D.4}, 
\begin{equation*} 
\underset{G\in \mathcal{G}}{\sup} \left|\bar{\mathcal{W}}(G)-\widetilde{\mathcal{W}}(G)\right|\lesssim_p \frac{c_T [2V(\log T)e\gamma ]^{1/\gamma}  \sup_{G\in \mathcal{G}}\Phi_{\phi_{\tilde{v}}}(\tilde{S}_{.}(G))}{\sqrt{T-1}}+ C\sqrt{\frac{v}{T-1}},
\end{equation*}
where $c_T$ is a large enough constant; $\tilde{v} = 1/2$ or 1, and $\gamma = 1/(1+2\tilde{v})$; $V$ and $v$ are the constants defined in Assumption \ref{D.4}; and $C$ is a constant analogous to that in Theorem \ref{thm:mds_bound}, which depends only on $M$ and $\kappa$. 
\end{theorem}
A proof is presented in Appendix \ref{Proof_of_mean_bound}.
The bound depends on the complexity of the function class, $V$ and $v$, and the time-series dependency, $\mbox{sup}_{G}\Phi_{\phi_{\tilde{v}}}(\tilde{S}_{.}(G))$. As we consider the sample analogue of unconditional welfare, all the observations are utilized, resulting in a $\frac{1}{\sqrt{T-1}}-$rate of convergence.

\subsubsection{Regret bound}

Now, we can obtain the overall bound for unconditional welfare, using \eqref{welfare_inequality}. Let $P_T$ be a joint probability distribution of a sample path of length $(T-1)$, $\mathcal{P}_T(M,\kappa)$ be the class of $P_T$, which satisfies Assumptions \ref{ass:continuous_Markov} to \ref{unconf_c}, and \ref{equiv_W_uc} to  \ref{D.4}; let ${\E}_{P_{T}}$ be the expectation taken over different realizations of random samples. 
\begin{theorem} \label{thm:mds_mean_bound_unconditional}
Under Assumptions \ref{ass:continuous_Markov} to \ref{ass bounded y}, and \ref{equiv_W_uc} to  \ref{D.4},
{\small 
\begin{equation} \label{eq:main_Theorem}
\underset{P_{T}\in\mathcal{P}_T(M,\kappa)}{\sup} {\E}_{P_T}[\mathcal{W}_{T}(G_*)-\mathcal{W}_{T}(\hat{G})]\lesssim C \sqrt{\frac{v}{T-1}} + \frac{c_T [2V(\log T)e\gamma]^{1/\gamma} \sup_{G\in \mathcal{G}} \Phi_{\phi_{\tilde{v}}}(\tilde{S}_{.}(G))}{\sqrt{T-1}},
\end{equation}}
where $G_*$ is the optimal policy defined in \eqref{maxG}.
\end{theorem}
This theorem follows from \eqref{welfare_inequality}, Theorems \ref{thm:mds_bound} and \ref{mean_bound}, and analogous reasoning presented in Appendix \ref{simple_mds_proof}.


\subsubsection{Using unconditional T-EWM to bound conditional regret}\label{sec:uncon_con} 

The kernel method for the conditional regret discussed in Section \ref{sec:conti_kernel} is a direct way to estimate an optimal policy with the conditional welfare criterion. However, the localization by bandwidth slows down the speed of learning; the regret of conditional welfare can only achieve a $\frac{1}{\sqrt{(T-1)h}}$-rate of convergence rather than a $\frac{1}{\sqrt{T-1}}$-rate.  
This slow convergence rate may limit its practicality when the covariates $X_{t-1}$ are multi-dimensional or when we extend the order of Markovian dependence to multiple periods.  Therefore, in what follows, we instead pursue a novel approach that estimates a conditional optimal policy rule by maximizing an empirical analogue of unconditional welfare over a specified class of decision sets, $\mathcal{G}$. We show that under additional assumptions, this approach can lead to the convergence rate of the conditional welfare that is free from the curse of dimensionality. To this end, we rewrite the conditional welfare \eqref{conti_obj_kernel} with the notation for the decision set $G$:
\bea \label{conti_obj}
\mathcal{W}_{T}(G|x) 
=\E\left[ Y_{T}(W_{T-1},1)1(X_{T-1}\in G)+Y_{T}(W_{T-1},0)1(X_{T-1}\notin G)|X_{T-1}=x\right].
\eea
The relationship between $g$ and $G$ is shown in \eqref{gG}. We first clarify how a maximizer of conditional welfare  \eqref{conti_obj} can be linked to a maximizer of unconditional welfare. With this result in hand, we can focus on estimating unconditional welfare and choosing a policy by maximizing it.  In our setup, the complexity of the functional class needs to be specified by the user.
 
 We will show that this approach can attain a $\frac{1}{\sqrt{T-1}}$ rate of convergence. Faster convergence relative to the kernel approach comes at the cost of imposing an additional restriction on the data-generating process, as we spell out in the next assumption.  


\begin{assumption}[Correct specification] \label{ass:correct_specify}

Consider the conditional welfare  defined in \eqref{conti_obj} and
the unconditional welfare under policy $G$ defined in \eqref{wel:uncondi}. At every $x\in \mathcal{X}$, it holds that
\[ \mbox{argmax}_{G\in\mathcal{G}} \mathcal{W}_{T}(G)\subset \mbox{argmax}_{G\in\mathcal{G}} \mathcal{W}_{T}(G|x). \]
\end{assumption}
This assumption ensures that maximizing unconditional welfare corresponds to maximizing conditional welfare over $\mathcal{G}$. A sufficient condition for the equivalence of maximizing conditional and unconditional welfare is that the specified class of policy rules, $\mathcal{G}$, includes the first best policy $G^*_{FB}$ for the unconditional problem, where
\begin{equation} \label{eq:first_best}
G^*_{FB}:=\left\{x\in \mathcal{X}:\mathsf{E}[Y_{T}(W_{T-1},1)-Y_{T}(W_{T-1},0)|X_{T-1}=x]\geq 0\right\}.
\end{equation}
This sufficient condition states that the class of policy rules over which unconditional empirical welfare is maximized contains the set of points in $\mathcal{X}$ where the conditional average treatment effect $\mathsf{E}[Y_{T}(W_{T-1},1)-Y_{T}(W_{T-1},0)|X_{T-1}=x]$ is positive. This assumption thus restricts the distribution of potential outcomes at $T$ and its dependence on $X_{T-1}$. We refer to Assumption \ref{ass:correct_specify} as `correct specification'.\footnote{\cite{kitagawa2018should}, \cite{KST21}, and \cite{Sakaguchi21} consider correct specification assumptions exclusively for unconditional welfare criteria. These assumptions correspond to $G_{FB}^{\ast} \in \mathcal{G}$.} The planner can be confident about Assumption \ref{ass:correct_specify} if, for instance, $\mathsf{E}[Y_{T}(W_{T-1},1)-Y_{T}(W_{T-1},0)|X_{T-1}=x]$ is believed to be monotonic in $x$ (element-wise) and the class $\mathcal{G}$ consists of decision sets with monotonic boundaries \citep{MT17,KST21}.


With Assumption \ref{ass:correct_specify}, we can shift the focus to maximizing unconditional welfare, even when the planner's ultimate objective function is conditional welfare. 
We also note several side-benefits of considering optimal policy in terms of unconditional welfare.
The following proposition directly results from Assumptions \ref{ass:continuous_Markov} (i), Assumption \ref{ass:correct_specify}, and the definition of $G^*_{FB}$. 

\begin{prop} \label{first_b_equal}
        Under Assumptions \ref{ass:continuous_Markov} (i) and \ref{ass:correct_specify}, an optimal policy rule $G_* \in \mathcal{G}$ defined by \eqref{maxG} maximizes the conditional welfare function,
    $G_*  \in \mbox{argmax}_{G\in\mathcal{G}}\mathcal{W}_{T}(G|X_{T-1})$.
        Furthermore, if the first best solution belongs to the class of feasible policy rules,  $G^*_{FB}\in\mathcal{G}$, then we have
\beaa
G^*_{FB} \in \mbox{argmax}_{G\in\mathcal{G}}\mathcal{W}_{T}(G|X_{T-1}).
\eeaa
\end{prop}

Having assumed the relationship of optimal policies between the two welfare criteria, we now show how the unconditional welfare function can bound the conditional function. For $G \in \mathcal{G}$ and $X_{T-1}=x$, define conditional regret as
\[
R_{T}(G|x)=\mathcal{W}_{T}(G_{*}|x)-\mathcal{W}_{T}(G|x).
\]
 Note that the unconditional regret can be expressed as an integral of the conditional regret,
\begin{align*}
\mathcal{W}_{T}(G)&=\int\mathcal{W}_{T}(G|x)dF_{X_{T-1}}(x),\\
R_{T}(G)&=\mathcal{W}_{T}(G_*)-\mathcal{W}_{T}(G)=\int R_{T}(G|x)dF_{X_{T-1}}(x).
\end{align*}
For $x' \in \mathcal{X}$, define
\begin{align}
A(x^{\prime},G)&=\{x \in\mathcal{X} : R_{T}(G|x)\geq R_{T}(G|x^{\prime})\}, \notag \\
p_{T-1}(x^{\prime},G)&=\Pr\left(X_{T-1}\in A(x^{\prime},G)\right)=\int_{x\in A(x^{\prime},G)}d F_{X_{T-1}}(x), \notag
\end{align}
and let $x^{obs}$ denote  the observed value of $X_{T-1}$. We assume the following:
\begin{assumption} \label{lower_p_dis_con}
        For  $x^{obs}\in \mathcal{X} $ and any $G\in \mathcal{G}$, there exists a positive constant $\underline{p}$ such that
        \begin{equation}
                p_{T-1}(x^{obs},G)\geq\underline{p}>0.
        \end{equation}
\end{assumption}
\begin{rem}
This assumption is satisfied if $X_{T-1}$ is a discrete random variable taking a finite number of different values. In this case, $p_{T-1}(x^{obs},G)\geq \min_{x\in \mathcal{X}}\Pr\left(X_{T-1}= x\right)>0$, so we can set $\underline{p}=\min_{x\in \mathcal{X}}\Pr\left(X_{T-1}= x\right)$.  If $X_t$ is continuous, then we need to exclude a set of points around the maximum of the function $R_{T}(G|x)$ for the assumption to hold.
Namely, we can assume that we focus on $x$ belonging to a compact subset $\tilde{\mathcal{X}}\subset \mathcal{X}$ such that $\argmax_{x\in\mathcal{X} }R_{T}(G|x)\notin \tilde{\mathcal{X}}$. 
If we would like to include the whole support of $X_t$, we can modify the proof by imposing an additional uniform continuity condition on $R_{T}(G|\cdot)$.  
\end{rem}
The following lemma provides a bound for conditional regret $ R_{T}(G|x^{obs})$ using  unconditional regret $ R_{T}(G)$.

\begin{lemma} \label{lem:lower_RT(G)} Under Assumptions \ref{ass:correct_specify} and \ref{lower_p_dis_con},
\begin{align}
    R_{T}(G|x^{obs})  \leq  \frac{1}{\underline{p}} R_{T}(G). 
\end{align}
\end{lemma}
The proof of this lemma can be found in Appendix \ref{app:proof_uncon_con_bound}. Using Assumption \ref{lower_p_dis_con} and Lemma \ref{lem:lower_RT(G)}, we proceed to bound the regret for conditional welfare $ \mathcal{W}_{T}(G_{*}|x^{obs})-\mathcal{W}_{T}(\hat{G}|x^{obs})$ by regret for unconditional welfare $\left[\mathcal{W}_{T}(G_{*})-\mathcal{W}_{T}(\hat{G})\right]$, and further by $\left[\mathcal{\widetilde{W}}(G)-\widehat{\mathcal{W}}{}(G)\right]$ (up to constant factors):
\begin{align} \label{welfare_inequality2}
        \mathcal{W}_{T}(G_{*}|x^{obs})-\mathcal{W}_{T}(\hat{G}|x^{obs})
        & \leq\frac{1}{\underline{p}}\left[\mathcal{W}_{T}(G_{*})-\mathcal{W}_{T}(\hat{G})\right]\nonumber\\
        & \leq\frac{c}{\underline{p}}\left[\mathcal{\widetilde{W}}(G_{*})-\mathcal{\widetilde{W}}(\hat{G})\right]
         \leq\frac{2c}{\underline{p}}\sup_{G\in \mathcal{G}}\left|\widehat{\mathcal{W}}(G)-\mathcal{\widetilde{W}}(G)\right|.
\end{align}
The first inequality
follows from Lemma \ref{lem:lower_RT(G)} and Assumption \ref{lower_p_dis_con}.
The second inequality follows from Assumption \ref{equiv_W_uc}.
The last inequality follows from an argument similar to the one below \eqref{ineq_s}. As a result, the bound of the conditional regret can be further established using Theorems \ref{thm:mds_bound} and \ref{mean_bound}; recall that $x_{obs}$ is defined to be the observed value of $X_{T-1}$.







\begin{theorem} \label{thm:mds_mean_bound}
Under Assumptions \ref{ass:continuous_Markov} to \ref{ass bounded y}, and \ref{equiv_W_uc} to \ref{lower_p_dis_con},
{\small 
\begin{equation} \label{eq:main_Theorem}
\underset{P_{T}\in\mathcal{P}_T(M,\kappa)}{\sup} {\E}_{P_T}[\mathcal{W}_{T}(G_*|x_{obs})-\mathcal{W}_{T}(\hat{G}|x_{obs})]\lesssim \frac{1}{\underline{p}}\left(C \sqrt{\frac{v}{T-1}} + \frac{c_T [2V(\log T)e\gamma]^{1/\gamma} \sup_{G\in \mathcal{G}} \Phi_{\phi_{\tilde{v}}}(\tilde{S}_{.}(G))}{\sqrt{T-1}}\right).
\end{equation}}
\end{theorem}

\begin{rem}
 [on Assumptions \ref{ass:correct_specify} and \ref{lower_p_dis_con}] 
 In this theorem, the conditional regret is compared with the conditional welfare achieved under $G_*$, which is the optimal policy within the class of feasible unconditional decision sets, $\mathcal{G}$. Under Assumption \ref{ass:correct_specify}, $G_*$ can be substituted with $G^*_{FB}$ without altering the regret bound. However, if Assumption \ref{ass:correct_specify} is violated, the regret bound relative to $G^*_{FB}$ becomes
 {\small 
 \beaa&&\underset{P_{T}\in\mathcal{P}_T(M,\kappa)}{\sup} {\E}_{P_T}[\mathcal{W}_{T}(G^*_{FB}|x_{obs})-\mathcal{W}_{T}(\hat{G}|x_{obs})]
\\&&\lesssim \frac{1}{\underline{p}}\left(C \sqrt{\frac{v}{T-1}} + \frac{c_T [2V(\log T)e\gamma]^{1/\gamma} \sup_{G\in \mathcal{G}} \Phi_{\phi_{\tilde{v}}}(\tilde{S}_{.}(G))}{\sqrt{T-1}}\right)
+\left[\mathcal{W}_{T}(G^*_{FB}|x_{obs})-\mathcal{W}_{T}(G_*|x_{obs})\right],
\eeaa}
 where the last term represents the cost to social welfare incurred by adhering to policy restrictions (such as ethical, moral, or legislative considerations) that are implied by $\mathcal{G}$.
 
 The factor $\frac{1}{\underline{p}}$ arises from Assumption \ref{lower_p_dis_con}.  A violation of this assumption, i.e., setting $\underline{p}=0$, could result in unbounded conditional regret.  Under Assumption \ref{lower_p_dis_con}, Theorem \ref{thm:mds_mean_bound}  shows that the regret bound of our proposed policy choice converges at a rate of $\frac{1}{\sqrt{T-1}}$.
\end{rem}

\section{Extensions and Discussion}\label{ext}
In this section, we discuss a few possible extensions. 
Section \ref{sec_multi} introduces a multi-period policy-making framework. 
Section \ref{sec:lucas critique} concerns the ability of the current methods to handle Lucas critique. Section \ref{sec:connection} discusses connections of T-EWM to the literature of reinforcement learning with Markov decision processes and structural impulse response analysis. More extensions are discussed in Appendix \ref{App_2}. 


\subsection{Multi-period welfare} \label{sec_multi}
So far, we have considered only the case of a one-period welfare function.
This subsection discusses how to extend  the current setting to a multiple-period policy framework. In the interest of space, we focus on cases with discrete covariates
and extend the simple model in Section \ref{one-p} to a two-period welfare function. Extending this model beyond two periods is straightforward using similar reasoning.
Recall Assumption  \ref{ass:toy_example}, which imposes Markov properties on the data-generating processes,
\beaa
Y_{t}(w_{0:t}) &=& Y_t(w_{t-1},w_t),\nonumber\\
Y_{t}(w_{0:t}) &\perp& X_{0:t-1}|W_{t-1},\nonumber\\
W_{t} &\perp& X_{0:t-1}|W_{t-1}.
\eeaa

The planner chooses policy rules for two periods, $g_1(\cdot)$ and $g_2(\cdot): \{0,1\}\to  \{0,1\}$, to maximize aggregate welfare over periods $T$ and $T+1$. 
By Assumption \ref{ass:toy_example}, the period $T+1$ welfare is determined only by the treatment choice in the period $T$, i.e., $W_T=g_1(W_{T-1})$.
The  two-period welfare function can be written as
\begin{align} \label{two-p-two-g}
	\mathcal{W}_{T:T+1}(g_1(\cdot), g_2(\cdot)|\mathcal{F}_{T-1}) & =\mathcal{W}_{T:T+1}(g_1(\cdot),g_2(\cdot)|W_{T-1}) \nonumber \\
	& =\mathcal{W}_{T}(g_1(\cdot)|W_{T-1})+\mathcal{W}_{T+1}\left(g_2(\cdot)|W_{T}=g_1\left(W_{T-1}\right)\right),
\end{align}
where the last equality follows from Assumption \ref{ass:toy_example} and \eqref{simple_cond_exp}.
To simplify notation, we suppress the $(\cdot)$ in $g_i(\cdot)$, when the meaning is clear from the context.
\bea \label{eq:multi_welfare}
&& \mathcal{W}_{T:T+1}(g_1, g_2|W_{T-1}=w)\nonumber\\
& &=\E\left[ Y_{T}(W_{T-1},1)g_{1}(W_{T-1}) +Y_{T}(W_{T-1},0)\left( 1-g_{1}(W_{T-1}) \right)|W_{T-1}=w\right] \nonumber \\
& &+\E\left[ Y_{T+1}(W_T,1)g_{2}(W_{T}) +Y_{T+1}(W_T,0)(1-g_{2}(W_{T}) )|W_T=g_{1}(w) \right].
\eea
The second term of \eqref{eq:multi_welfare} follows from
\beaa
&&\E\left[ Y_{T+1}(g_{1}(W_{T-1}) ,1)g_{2}(W_{T}) +Y_{T+1}(g_{1}(W_{T-1}) ,0)(1-g_{2}(W_{T}) )|W_{T-1}=w\right] \\
&&=\E\left[ Y_{T+1}(W_T,1)g_{2}(W_{T}) +Y_{T+1}(W_T,0)(1-g_{2}(W_{T}) )|W_T=g_{1}(W_{T-1}) , W_{T-1}=w\right]\\
&&=\E\left[ Y_{T+1}(W_T,1)g_{2}(W_{T}) +Y_{T+1}(W_T,0)(1-g_{2}(W_{T}) )|W_T=g_{1}(w) \right],
\eeaa
where the last equality follows from Assumption \ref{ass:toy_example}. To estimate the above welfare function, we recall the definition of $T(w)=\#\{1\leq t\leq T-1:W_{t-1}=w\}$, and we define $T(g_{1}(w) )$ similarly. Then
the empirical analogue of \eqref{two-p-two-g} can be written as,
\begin{align} \label{sample_multi}
	\widehat{\mathcal{W}}_{T:T+1}(g_{1},g_{2}|w)
	= & \frac{1}{T(w)}\sum_{t:W_{t-1}=w}\left\{ \frac{Y_{t}W_{t}g_{1}(W_{t-1}) }{e_{t}(W_{t-1})}+\frac{Y_{t}\left(1-W_{t}\right)\left(1-g_{1}(W_{t-1}) \right)}{1-e_{t}(W_{t-1})}\right\} \nonumber \\
	+ & \frac{1}{T(g_{1}(w))}\sum_{t:W_{t-1}=g_{1}(w)}\left\{ \frac{Y_{t}W_{t}g_{2}(W_{t-1}) }{e_{t}(W_{t-1})}+\frac{Y_{t}\left(1-W_{t}\right)\left(1-g_{2}(W_{t-1}) \right)}{1-e_{t}(W_{t-1})}\right\}.
\end{align}

The maximizer of \eqref{sample_multi}, $\hat{g}_1$, $\hat{g}_2$,  can be obtained by backward induction, a technique widely applied in the Markov decision process (MDP) and dynamic treatment regime literature. See Section \ref{MDP} and Appendix \ref{MDP_A} for more discussion on the relationship between T-EWM and MDP.
To derive the theoretical property of the estimator, we also define
\begin{align*}
	\bar{\mathcal{W}}_{T:T+1}(g_1,g_2|w)
	& =\frac{1}{T(w)}\sum_{t:W_{t-1}=w}\E\left[ Y_{t}(1)g_{1}(W_{t-1}) +Y_{t}(0)\left[1-g_{1}(W_{t-1}) \right]|W_{t-1}=w\right] \\
	& +\frac{1}{T(g_{1}(w))}\sum_{t:W_{t-1}=g_{1}(w)}\E\left[ Y_{t}(1)g_{2}(W_{t-1}) +Y_{t}(0)\left[1-g_{2}(W_{t-1}) \right]|W_{t-1}=g_{1}(w) \right].
\end{align*}
Similarly to the derivation in the previous sections,   $\widehat{\mathcal{W}}_{T:T+1}(g_1,g_2|w)- \bar{\mathcal{W}}_{T:T+1}(g_1,g_2|w)$ is a (weighted) sum of MDS. Its upper bound can be shown by the method of Section \ref{discrete_bound}.
We show in Appendix \ref{appendmulti} the extension to multi-period welfare with continuous conditioning covariates. 

\subsection{Accounting for  Lucas critique} \label{sec:lucas critique}

How can the framework of T-EWM handle the Lucas critique? In this section, we clarify a link between T-EWM and the  SVAR approach based on a three-equation New Keynesian model where a choice of a policy regime can take into account the economy's policy response.


We start with a three-equation New Keynesian model. (See, e.g., Chapter 8 of \cite{walshmonetary}.) At time $t$, let $\pi_{t}$ denote inflation, $x_{t}$ the output gap, and $i_{t}$ the interest rate.
\bea \label{eq:three_equation}
&\text{Phillips curve: }&\pi_{t}=\beta {\E}_{t}\pi_{t+1}+\kappa x_{t}+\varepsilon_{t},\nonumber \\
&\text{IS curve: }&x_{t}={\E}_{t}x_{t+1}-\sigma^{-1}(i_{t}-{\E}_{t}\pi_{t+1}), \nonumber\\
&\text{Taylor rule: }&i_{t}=\delta\pi_{t}+v_{t},
\eea
Here $v_{t}$ is the monetary policy shock (baseline target rate), which can be viewed as a policy variable that the planner can manipulate. In the process of generating the sample, it is often assumed to follow an AR(1) process $ v_{t}=\rho v_{t-1}+e_{t}$. The AR coefficient of $v_t$, $\rho$, represents a monetary policy regime.
We also assume that the shocks in the Phillips curve follow an AR(1), $\varepsilon_t=\gamma \varepsilon_{t-1}+\delta_t$. Define $
d_{t}=\left( \begin{array}{c}
	v_{t}\\
	\varepsilon_{t}
\end{array}\right) \text{, } F=\left( \begin{array}{cc}
	\rho & 0\\
	0 & \gamma
\end{array}\right) $, and a vector of noises $\eta_t=\left( \begin{array}{c}
e_{t}\\
\delta_{t}
\end{array}\right) $. Then, the process for $d_t$ can be written as $d_t=Fd_{t-1}+\eta_t$.

Define the outcome variables of the system \eqref{eq:three_equation} to be $\tilde{Y}_{t}=\left( \begin{array}{c}
	x_t\\
	\pi_t
\end{array}\right)$. At the end of time $T-1$, the goal of the planner is to minimize (or maximize) the expectation of some function of $\tilde{Y}_{T}$. For example,  an objective function is the welfare cost that penalizes the time $T$ output gap and inflation, $Y_{T}= |x_{T}|^2+|\pi_{T}-\pi_0|^2$, where $\pi_0$ is the inflation target.

Appendix \ref{MDP_B} shows that the
VAR-reduced form of the system \eqref{eq:three_equation} can be expressed as:
\begin{equation}
	\tilde{Y}_{t}=M(\rho)d_{t},\label{eq:Yt_dt2}
\end{equation}
where  $M(\rho)$ is a non-random matrix defined in Appendix \ref{MDP_B}. If the model \eqref{eq:three_equation} is correctly specified, the solution to \eqref{eq:Yt_dt2} takes the Lucas critique into account since it solves for a \emph{deep} parameter $\rho$. A change in $\rho$ incorporates both the direct effect of the policy regime (through $\rho$ in $d_t$ equation) and private agents' anticipation of the policy change ($\rho$ in $M(\rho)$).

Now we show how this is related to the T-EWM framework. The treatment $v_t$ in \eqref{eq:three_equation} corresponds to $W_t$ in the previous sections.
We can write $M(\rho)=\left( \begin{array}{cc}
	m_{11}(\rho) & m_{12}(\rho)\\
	m_{21}(\rho) & m_{22}(\rho)
\end{array}\right) $. When $v_t$ is a binary variable (e.g., high target rate and low target rate), we can define the potential outcomes by setting $v_t = 1$ or $0$ in (\ref{eq:Yt_dt2}): $\tilde{Y}_t(1)=\left( \begin{array}{c} m_{11}(\rho)+m_{12}(\rho) \varepsilon_{t}\\m_{21}(\rho)+m_{22}(\rho) \varepsilon_{t}  \end{array}\right) $ and $\tilde{Y}_t(0)=\left( \begin{array}{c}m_{12}(\rho) \varepsilon_{t}\\m_{22}(\rho) \varepsilon_{t}  \end{array}\right) $. Both $\tilde{Y}_t(1)$ and $\tilde{Y}_t(0)$ depend on $\rho$, so we can write them as $\tilde{Y}_t(1;\rho)$ and $\tilde{Y}_t(0;\rho)$.  
Transforming $\tilde{Y}_t=(x_t,\pi_t)^{\prime}$ into $Y_t=|x_t|^2 + |\pi_t - \pi_0|^2$, we can define the potential outcomes for the welfare $(Y_t(1; \rho), Y_t(0; \rho))$, which also depend on $\rho$. The policy choice problem for $v_T \in \{0, 1\}$ can be set to minimize the expected welfare cost:
\begin{equation}
	 \E \left[ Y_{T}(1;\rho) g(v_{T-1}) + Y_T(0 ; \rho) (1-g(v_{T-1})) | \mathcal{F}_{T-1} \right]\label{eq:vT_lucas}
\end{equation}
in $g(v_{T-1}) \in \{0, 1\}$, assuming that the one-time policy choice of $v_T$ does not change the value of $\rho$ governing the outcome-generating process (\ref{eq:Yt_dt2}). Thus, one can view the framework of T-EWM in the previous sections concerns a choice of binary policy shock $v_T$, assuming the fixed deep parameter of $\rho$. 

In contrast, consider the case where the policy choice of interest concerns regime $\rho$ instead of $v_T$. 
We can modify the T-EWM framework as follows in order to handle this case. 
Assume that the monetary policy shock $v_t$ is continuously distributed and let $f_{v_{t}|\mathcal{F}_{t-1}}( v;\rho)$ be the conditional density of $v_t$ given the filtration $\mathcal{F}_{t-1}$ evaluated at $v$ and $\rho$. In the outcome process of $t=T-1$ and earlier, the distribution $ f_{v_{t}|\mathcal{F}_{t-1}}( v;\rho)$ at the value of $\rho$ generates the outcomes up to $T-1$. In the counterfactual policy scenario of period $t=T$ and later, the planner manipulates the distribution of $v_T$ by changing the value of $\rho$ instead of directly setting a particular value of $v_T$. 
For such an intervention, we specify the planner's problem as a choice of $\rho$ to minimize
\begin{equation}
	\E \left[Y_{T}(v_{T};\rho)|\mathcal{F}_{T-1}\right]=\int \E\left[Y_{T}(v;\rho)|\mathcal{F}_{T-1}\right]f_{v_{T}|\mathcal{F}_{T-1}}(v;\rho)dv\label{eq:continuous_lucas}.
\end{equation}
This objective function differs from the welfare objective function of  \eqref{eq:vT_lucas} in the following two aspects. 
First, in (\ref{eq:vT_lucas}), the policy rule for $v_T$ is deterministic, whereas in (\ref{eq:continuous_lucas}), the planner chooses $\rho$ to change the conditional probability density function of the treatment $v_T$, i.e., a randomized policy for $v_T$. 
Second, we see that in (\ref{eq:vT_lucas}), the policy $v_T$ affects the outcome $Y_T$ only through $v_{T}$ with $\rho$ fixed, whereas in (\ref{eq:continuous_lucas}),
the policy $\rho$ affects the outcome through both the distribution of $v_T$
and the deep policy parameter $\rho$ underlying the potential outcomes.

Despite these differences, we can pursue a T-EWM approach for a statistical policy choice of $\rho$ as long as one can construct a sample analogue of the welfare criterion of (\ref{eq:continuous_lucas}). 
For instance,  let $\rho_t$ be the policy regime in the sampling period $t=1, \dots, T-1$, and assume $\rho_t$ is observable or estimable.\footnote{$\rho_t$ can be estimated by the method proposed in \cite{schorfheide2005learning}, assuming that monetary policy follows a nominal interest rate rule that is subject to regime shifts.} 
For the sake of illustration, we will use a simplified model by  assuming that $\rho_t$ takes values in a finite set $\Omega$ and $v_t\in\{0,1\}$. 
 For a given $\rho\in\Omega$, we assume that $\rho_t$ is independent of $Y_{t}(v_t;\rho)$  conditional on $\mathcal{F}_{t-1}$, and we let
 Assumption \ref{ass:toy_example} hold with $W_{t} = v_{t}$, such that $\E(\cdot|\mathcal{F}_{t-1})=\E(\cdot|v_{t-1})$. 
Then, for a policy function $g_{\rho}:\{0,1\}\to \Omega$ and $v\in\{0,1\}$, we can define the empirical welfare estimating \eqref{eq:continuous_lucas} as
 \begin{equation}
\widehat{\mathcal{W}}\left(g_{\rho}|v\right)
		=  T(v)^{-1}\sum_{t:v_{t-1}=v}\left[\sum_{\rho\in\Omega}\frac{Y_{t}\IF(\rho_t=\rho)}{\Pr[\rho_t=\rho|v_{t-1}]}\IF[g_{\rho}(v_{t-1})=\rho]\right], \label{eq:emp_w_lucas}
 \end{equation}
where $T(v):=\# \{ v_{t-1}=v\}$. We then maximize this empirical criterion with respect to $v=v_{T-1}$ to estimate an optimal policy regime. 


It is worth noting that 
as long as $\rho_t$ is observable (or estimable), 
the empirical analogue of the welfare criterion is free from functional form restrictions of $Y_t(v; \rho)$ or a distributional restriction of $f_{v_t|\mathcal{F}_{t-1}}(v; \rho)$. 
On the other hand, whether the empirical analogue can estimate the welfare criterion (\ref{eq:continuous_lucas}) well relies on whether the agents in the economy have correct knowledge of $\rho_t$ and behave in response to the shift of $\rho_t$.  
We leave for future research a rigorous characterization of the learnability of an optimal policy regime along this T-EWM approach and its implementability in practice.

\subsection{Connections to other policy choice models in the literature}\label{sec:connection}
In this section, we discuss T-EWM's relation to the literature on treatment and optimal policy analysis.

\subsubsection{Connection to MDP and Reinforcement learning} \label{MDP}

Markov Decision Process (MDP) is the standard framework for reinforcement learning algorithms commonly applied to decision-making problems in dynamic environments. See, e.g.,  \cite{kallenberg2016markov} for a comprehensive introduction. 
We can relate the current T-EWM model to a special case of MDP with a finite horizon. 
The conditioning variables $X_{t-1}$ correspond to the Markov state at time $t$, and the welfare outcome $Y_{t}$ corresponds to the reward at $t$. 
Before the planner intervenes at time $T$, the Markov state transitions follow the data-generating process in which the transitions of the policies are prescribed by the propensity score. 
After the planner intervenes, the transition of policies is governed by a deterministic rule described by the (estimated) optimal policy function \eqref{eq:policy_f}. 
As we show in Appendix \ref{MDP_A}, the population conditional welfare of T-EWM with a finite-horizon welfare target can be viewed as the value function of a finite-horizon MDP with a nonstationary solution. See Chapter 2 of  \cite{kallenberg2016markov} for more details about finite-horizon MDPs.

The reinforcement learning (RL) literature is vast and provides a rich toolbox to solve MDPs given the knowledge of state transition and reward generating processes. Treating the state transition and reward generating processes as unknown, the recent works of off-policy evaluation for MDPs including \cite{Hu_Wager2023_AnnStat}, \cite{Kallus_Uehara2020}, \cite{ Liao_etal2021_JASA}, \cite{Shi_etal2023_JASA}, among others, study estimation and inference for the welfare values under (counterfactual) Markov policies using training data. Our proposal and analysis of T-EWM contrast with these works since our focus is on estimation of an optimal policy and its short-run welfare regret performance rather than statistical inference for the welfare values.

As a closely-related work to ours, \cite{Liao_etal2022_AnnStat} studies estimation of an optimal policy in MDP within a parametrically specified class of policies. 
However, our T-EWM and their MDP approaches differ in several aspects. First, T-EWM allows for a nonstationary environment in both causal effects and data-generating processes without requiring modeling them explicitly, while optimal policies in MDP assume time homogeneous transitions of the underlying states.
Second, T-EWM obtains a policy by optimizing an empirical analogue of a short-run welfare criterion, whereas the MDP approach estimates an optimal policy by maximizing an estimate of the value function in the Bellman equation.

\subsubsection{Comparison with impulse response functions}

In empirical macroeconomics, a common practice is to measure the causal effects of policies using the impulse response functions (IRF) to the structural policy shocks; see \cite{sims1980macroeconomics}, \cite{ramey2016macroeconomic}, \cite{plagborg2016essays}, \cite{stock2017twenty}, among others. This approach is based on the representation that the endogenous outcome variables are expressed as a weighted sum of the structural shocks, in which the coefficients of the structural shocks correspond to IRFs. With our notation, it can be expressed as
\begin{equation} \label{Structural MA}
Y_t = \sum_{h=0}^{\infty} ( \theta_{h} W_{t-h} + \phi_h \epsilon_{t-h} ),
\end{equation}
where $W_t$ is the policy shock and $\epsilon_t$ is a non-policy structural shock that the planner cannot control. Setting the value of structural shocks exogenously to $w_{-\infty:t}$ without intervening in the joint distribution of the other shocks $\epsilon_{-\infty:t}$ defines the potential outcome $Y_t(w_{-\infty:t})$. The coefficients $(\theta_h, \phi_h)$, $h = 0, 1, 2, \dots,$ are IRFs of the $h$-period ahead outcome variable to the unit changes in policy and non-policy shocks, respectively.

Our approach of T-EWM based on the potential outcome time series differs from the approach of causal IRF analysis in several aspects. First, we assume that the policy shocks of interest $\{W_t: t=0, \dots, T-1 \}$ are directly observable in the available time-series data, while the common framework of IRF analysis such as SVAR considers settings where the policy shocks are not observable and need to be identified through identification of a linear simultaneous equation system. Second, our framework imposes little restrictions on the functional form of how the current and past shocks affect the outcome, while the IRF analysis based on (\ref{Structural MA}) assumes that the structural equation for $Y_t$ is linear in the structural shocks and additive over the horizons. Third, our framework allows unrestricted heterogeneity (i.e., nonstationarity) of the causal effects, whereas the shock-based causal models (\ref{Structural MA}) assume either stationary IRFs that depend only on horizon $h$ or nonstationary IRFs with explicit modeling of how they evolve over time.\footnote{See \cite{bojinov2019time}
and \cite{rambachan2021when} for comparative discussions between the IRFs and the average treatment effects with potential outcome time series.}  Fourth, the existing IRF analysis focuses mainly on estimating and inferring the IRFs, whereas the literature has not formulated how to perform statistical policy choice based on the estimated IRFs. Our T-EWM approach, in contrast, explicitly formulates the policy choice problem using potential outcome time series and analyzes the welfare performance of the statistical policy choice.

\subsection{\label{sec: unknown propensity-score}Estimation with unknown propensity score}
To this point, we have treated the propensity score function as known, but this is infeasible in many applications. Here we consider the case where the propensity score at each time $t$,
$e_{t}(\cdot)$, is unknown. Estimation can be
either parametric or non-parametric. Let $\hat{e}_{t}(\cdot)$ denote the estimator of the propensity score function, and $\hat{G}_{\hat{e}}$ denote the optimal policy obtained using $\hat{e}_t$.

\subsubsection{The convergence rate with estimated propensity scores}

In this subsection, we adapt Theorem 2.5 of \citet{kitagawa2018should}
to our setting and obtain a new regret bound with estimated propensity scores. We show that, with estimated propensity scores,
the convergence rate is determined by the slower one of the rate of convergence of
$\hat{e}_{t}(\cdot)$ and the rate of convergence of the estimated welfare loss (given
in Theorem \ref{thm:mds_mean_bound_unconditional}).
\begin{theorem} \label{thm:e_propensity}
	 Let $\hat{e}_{t}(\cdot)$ be an estimated
	propensity score of $e_{t}(\cdot)$, and
		$\hat{\tau}_{t}=\frac{Y_{t}W_{t}}{\hat{e}_t(W_{t-1})}- \frac{Y_{t}(1-W_{t})}{1- \hat{e}_t(W_{t-1})}$
	be a feasible estimator for
$\tau_{t}=\frac{Y_{t}W_{t}}{e_t(W_{t-1})}- \frac{Y_{t}(1-W_{t})}{1- e_t(W_{t-1})}$.
	Given a class of data-generating processes $\mathcal{P}_{T}(M,\kappa)$ defined above equation \eqref{eq:main_Theorem}, we assume
	that there exists a sequence $\phi_{T}\to\infty$ such that the series
	of estimators $\hat{\tau}_{t}$ satisfy
		\begin{equation}
		\text{lim}\underset{T\to\infty}{\text{sup}}\underset{P_T\in\mathcal{P}_{T}(M,\kappa)}{\text{sup}}\phi_{T}{\E}_{P_{T}}\left[(T-1)^{-1}\sum_{t=1}^{T-1}|\hat{\tau}_{t}-\tau_{t}|\right]<\infty.
		\label{eq:ass phi}
	\end{equation}
	
	Then under Assumptions \ref{ass:continuous_Markov} to \ref{ass bounded y}, and \ref{equiv_W_uc} to  \ref{D.4}, we have
	\[
	\underset{P_T\in\mathcal{P}_{T}(M,\kappa)}{\text{sup}}{\E}_{\mathcal{P}_{T}}[\mathcal{W}(G_*)-\mathcal{W}(\hat{G}_{\hat{e}})]\lesssim
	(\phi_{T}^{-1}\vee \frac{1}{\sqrt{T-1}}).
	\]
		
\end{theorem}
To ensure that $\hat{e}_t(W_{t-1})$ is a valid estimator of ${e}_t(W_{t-1})$ in the sense of satisfying (\ref{eq:ass phi}), we need to restrict the dynamics of $W_t$. It is, however, important to note that the condition of (\ref{eq:ass phi}) does not require stationarity of the outcome variable $Y_t(.)$.

		A proof is presented in Appendix \ref{sec:proof_e_propensity}. This theorem shows that if the propensity score is estimated with sufficient accuracy (a rate of $\phi_{T}^{-1} \lesssim \sqrt{T}^{-1}$),  we obtain a similar regret bound to the previous sections. It is not surprising to see that the rate is affected by the
		estimation accuracy of the propensity score, and it is the maximum of $\frac{1}{\sqrt{T-1}}$ and $\phi_{T}^{-1}$.
\begin{rem}
In the cross-sectional setting,  \cite{athey2021policy} show an improved rate of welfare convergence when propensity scores are unknown and estimated. It is possible to extend their analysis to our time-series setting and assess whether or not the rate shown in Theorem \ref{thm:e_propensity} can be improved. However, this is not a trivial extension, and we leave it for further research.  \end{rem}

\subsubsection{Estimation of propensity scores} \label{sec:estimated_prop}
In this subsection, we briefly review various methods of estimating propensity score functions. The propensity score function $e_{t}(\cdot)$ can be estimated parametrically
or nonparametrically.
An example of a parametric estimator is the (ordered) probit model,
which is employed by \cite{hamilton2002model}, \cite{scotti2018bivariate},
and \cite{angrist2018semiparametric}. Under Assumption \ref{ass:continuous_Markov}, $e_{t}$ can be expressed as a function of
$X_{t-1}$. Then, the structure of the propensity score can be given by a probit model,
\begin{align*}
	e_{t}(X_{t-1}) & \equiv P(W_{t}=1|X_{t-1})=\Phi(\beta^{\prime}X_{t-1}),\\
	1-e_{t}(X_{t-1}) & \equiv P(W_{t}=0|X_{t-1})=1-\Phi(\beta^{\prime}X_{t-1}).
\end{align*}

A more complicated structure, such as the dynamic
probit model (\cite{eichengreen1985bank}; \cite{davutyan1995operations}) can also be employed.

We can also use a nonparametric estimator to estimate $e_{t}(\cdot)$.
For example, \cite{frolich2006non} and \cite{park2017nonparametric} extend the
local polynomial regression of \cite{fan1995data} to a dynamic setting. Their methods can be employed here. For simplicity, we assume that the propensity score function is invariant across times, i.e., $e_t(\cdot)=e(\cdot)$ for any $t$, and $e(\cdot)$ is continuous, and we set $\mathcal{X}\subset \mathbb{R}^1$.
For a local polynomial of order $p=1$ and any $x\in\mathcal{X}$, a local likelihood logit model can be specified as
$$\text{log}\left[\frac{e(x)}{1-e(x)}\right]=\alpha_x,$$
for a local parameter $\alpha_x$.
By the continuity of the propensity score function $e(\cdot)$, for
some $X_{t-1}$ close to $x$, we can find a local parameter $\beta_{x}$, such that
$\text{log}\left[\frac{e(X_{t-1})}{1-e(X_{t-1})}\right]\approx\alpha_{x}+\beta_{x}(X_{t-1}-x)$.
The estimated propensity score evaluated at $x$, $\hat{e}(x)$, can be obtained by solving
\begin{align*}
(\hat{\alpha}_{x},\hat{\beta}_{x}) & =\text{argmax}_{\alpha,\beta}\frac{1}{T-1}\sum_{t=1}^{T-1}\bigg\{\bigg[W_{t}\log\left(\frac{\exp\left(\alpha+\beta\left(X_{t-1}-x\right)\right)}{1+\exp\left(\alpha+\beta\left(X_{t-1}-x\right)\right)}\right)\\
 & +(1-W_{t})\log\left(\frac{1}{1+\exp\left(\alpha+\beta\left(X_{t-1}-x\right)\right)}\right)\bigg]K\left(\frac{X_{t-1}-x}{h}\right)\bigg\}.
\end{align*}


where $K(\cdot)$ is a kernel function, and $h$ is the bandwidth. Then, we have 
$\hat{e}(x)=\frac{\exp(\hat{\alpha}_{x})}{1+\exp(\hat{\alpha}_{x})}.$

\section{Application} \label{applica}

 During the Covid-19 pandemic, policymakers
 around the world faced the problem of making effective policies. In this section, we illustrate the usage of T-EWM with an application to choosing the stringency of restrictions imposed in the United States during the pandemic. Let  the treatment $W_{t}$ be a binary indicator for whether the government maintains restrictions at week $t$: $W_{t}=1$ indicates that the stringency of restrictions is maintained or increased from week $t-1$, while $W_{t}=0$ indicates that restrictions are relaxed. The stringency of restrictions is measured by the Oxford Stringency Index,
	which is described in Section \ref{subsec:Data-description}. We assume that a change in the stringency of restrictions at time $t$ will have a lagged effect on deaths, and set 
	the outcome variable $Y_{t}$ to be  \emph{$-1\times$two-week-ahead deaths}. The  information set at time $t$ is,
	\begin{align}
		\ensuremath{}\ensuremath{X_{t}=} & (\text{cases}_{t},\text{deaths}_{t},\text{change in cases}_{t},\text{change in deaths}_{t},\nonumber \\
		& \text{restriction stringency}_{t},\text{vaccine coverage}_{t},\text{economic conditions}_{t}).\label{eq:app_X}
	\end{align}
	The variables in $X_{t}$ are chosen to include the most important factors considered by  policymakers when deciding the stringency of restrictions.
	The inclusion of the economic conditions reflects policymakers'
    concerns over the economic effect of restrictions. Finally, the propensity score is estimated as a logit model with a linear
index,
\begin{equation}
	\text{log}\left(\frac{\Pr(\text{keeping or increasing restrictions at week }t)}{\Pr(\text{decreasing restriction at week }t)}\right)=\alpha+\beta X_{t-1}.\label{eq:p_score_app}
\end{equation}

\subsection{Data\label{subsec:Data-description}}

The dataset consists of weekly data for the United States. It runs from April 2020 to January 2022 and contains 92 observations. Data on cases, deaths, and vaccine coverage are downloaded from the website
of the Centers for Disease Control and Prevention (CDC). These series are
plotted in Figures \ref{fig:Cases-and-deaths} and \ref{fig:Vaccine} in Appendix \ref{app_app}. We use the Oxford Stringency Index as a measure
of the stringency of restrictions. This index is taken from the Oxford Covid-19 Government
Response Tracker (OxCGRT), which tracks policy interventions and constructs a suite of composite indices that measure
governments' responses.\footnote{The index is calculated based on a dataset
that is updated by a professional team of over one hundred students,
alumni, staff, and project partners. It is a composite index covering
different types of restriction, such as school and workplace
closures, restrictions on the size of gatherings, internal and international
movements, facial covering policies, etc. See \cite{hale2020variation} for
more details.} The left panel of Figure \ref{fig:Restriction-econ} in Appendix \ref{app_app} plots the time series of the stringency index.

The economic impact of restrictions is a crucial factor that every policymaker has
to consider. To measure economic conditions, we use the Lewis-Mertens-Stock weekly economic index (WEI). \footnote{The WEI is an index of ten indicators of real economic activity. It represents the common component of series covering consumer behavior, the labor market, and production. See \cite{lewis2021measuring} for more details. }The right panel of Figure \ref{fig:Restriction-econ} in Appendix \ref{app_app} plots the WEI.
In general, none
of these time series seems stationary over the sample period. They exhibit both seasonal fluctuations and changes with respect to different stages of the pandemic.

\subsection{On T-EWM assumptions}\label{checkas}
Before we proceed with the analysis, we discuss the validity of some of the T-EWM assumptions in this context. The Markov properties (Assumption \ref{ass:continuous_Markov}) require that (i) the potential
outcome at time $t$ depends only on the time $t$ and time $t-1$ treatments; (ii) conditional on $X_{t-1}$, the potential outcomes
and treatment at time $t$ are not affected by the path $X_{0:t-2}$.
(i) can be justified here as the treatment is the 
\emph{direction of change} in the stringency of restrictions. While the \emph{level}
before $t-1$ may affect the outcome at $t$, its directional
change may not have such a long-existing impact. (ii) requires that conditional on current cases and deaths, cases and deaths
one week ago (as $X_{t}$ includes the first difference of
cases and deaths, it includes the level of cases and deaths from one week ago), current economic
conditions, the stringency of restrictions, vaccine coverage, and deaths in two weeks
will be independent of lags of these variables if the lag is greater than one. For this assumption to hold, it must be that current infections are unrelated to
deaths in more than three weeks. In general, it is unclear
whether this is true, although there is some support
in the literature. For example, based on data collected during the
early stage of the pandemic in China, \cite{verity2020estimates} calculate
that the posterior mean time from infection to death is 17.8 days, with
a 95\%-credible interval of {[}16.9, 19.2{]}.



Strict overlap (Assumption \ref{bound_c}) requires that, for all $x\in\mathcal{X}$,
the propensity score is strictly larger than 0 and smaller than
1. Figure \ref{fig:PS_app} in Appendix \ref{app_app} shows the histogram of estimated propensity
scores. The strict overlap assumption can be verified by visual inspection: the smallest (estimated)
propensity score is larger than 0.3, and the largest is smaller than
0.9. Sequential unconfoundedness (Assumption  \ref{unconf_c}) requires that conditional
on $X_{t-1}$, treatment assignment at time $t$ is quasi-random. This assumption cannot be tested in general. However, in the first two years of pandemic, policymakers have had very limited knowledge of Covid-19 beyond the observable data. After controlling for the observables
in (\ref{eq:app_X}), it seems reasonable that the remaining
random factors in two-week-ahead deaths and changes in
the stringency of restrictions are conditionally independent.



\subsection{Estimation results and policy recommendations}
In this subsection, we summarize the estimation results and discuss the policy recommendations. We show that T-EWM leads to sensible and robust policy decisions.
After observing $X_{T-1}$, we aim to maximize expected welfare,
$\E(-1\cdot\text{deaths}_{T+1})$, over a set of quadrant policies. A class of quadrant treatment rules with $k$ policy variables, $x=(x_{1}\dots x_k)^{\prime}$, is defined as
\begin{equation*}
	\mathcal{G}\equiv\left\{ \begin{array}{c}
		\left\{x:s_{1}(x_{1}-b_{1})>0\;\&\;\dots\;\&\;  s_{k}(x_{k}-b_{k})>0\right\};\\
		s_{1},\dots,s_{k}\in\{-1,1\},b_{1},\dots,b_{k}\in \mathbb{R}
	\end{array}\right\}. 
\end{equation*}
Let $X_{T-1}^{P}$
denote a vector of variables for policy choice. This can be any
subvector of $X_{T-1}$.
For our first set of results, we use $X_{T-1}^{P}=\left(\text{change in deaths}_{T-1},\text{restriction stringency}_{T-1}\right).$
Figure \ref{fig:Policy_k=00003D2} presents the estimated treatment region.
The estimated optimal decision rule states that restrictions should not be relaxed (i.e., $W_{T}=1$), if the weekly fall in deaths is below $745$, and the current level of restrictions is lower
than 62.6.

\begin{figure}[h]
	\caption{\label{fig:Policy_k=00003D2} Estimated optimal policy based on $X_{T-1}^{P}=(\text{change in deaths}_{T-1},\text{restriction stringency}_{T-1})$}
	
	\centering
	
	\includegraphics[scale=0.35]{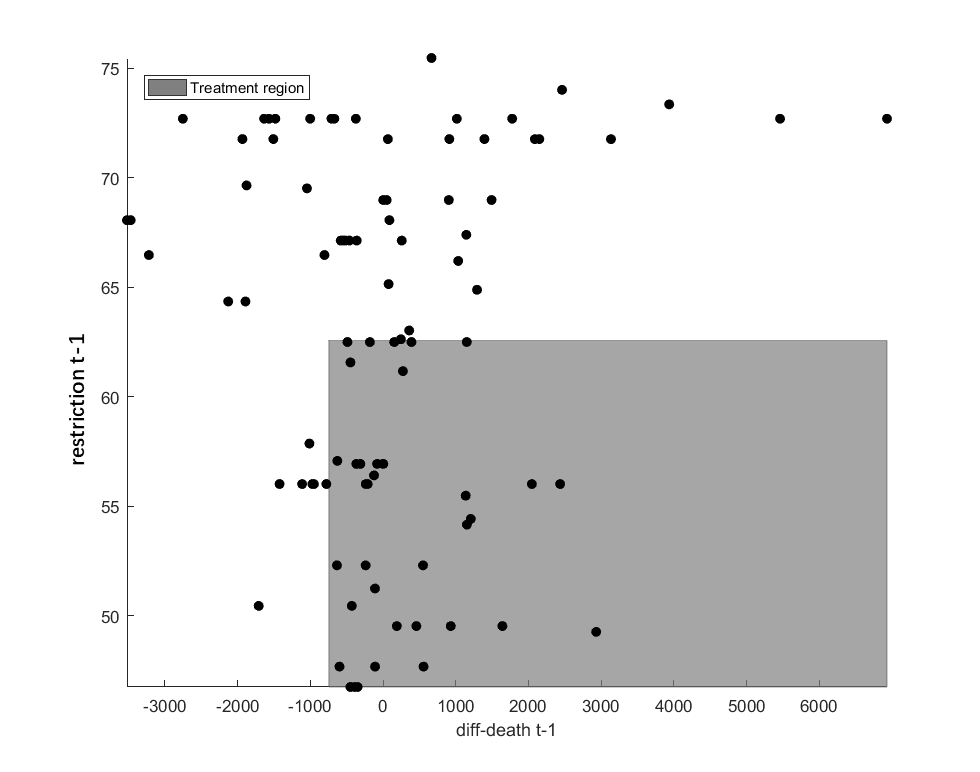}
		\noindent\begin{minipage}[t]{1\columnwidth}%
		\centering{\footnotesize{}The $x$-axis is the change in deaths at
			$T-1$, and the $y$-axis is the stringency of restrictions  at $T-1$.}%
	\end{minipage}
\end{figure}

We can examine the robustness of this result by expanding the set of variables for policy choice. We first add vaccine coverage: $X_{T-1}^{P}=(\text{change in deaths}_{T-1},$ $\text{restriction stringency}_{T-1},$ $\text{vaccine coverage}_{T-1})$.
The optimal estimated quadrant policy is then a 3d-quadrant. Figure
\ref{fig:Policy_k=00003D3}  in Appendix \ref{app_app} shows the projection of this 3d-quadrant
onto the 2d-planes of $(\text{change in deaths}_{T-1},$ $\text{restriction stringency}_{T-1})$
and $(\text{restriction stringency}_{T-1}$,\\  $\text{vaccine coverage}_{T-1}$).
We further expand the set of variables for policy choice by adding the change in cases, so that $X_{T-1}^{P}=(\text{change in deaths}_{T-1},$
$\text{restriction stringency}_{T-1},$ \\$\text{vaccine coverage}_{T-1},\text{ change in cases}_{T-1})$.
The optimal estimated quadrant policy in this case is a 4d-quadrant. Figure
\ref{fig:Policy_k=00003D4} in Appendix \ref{app_app} shows projections of this 4d-quadrant
onto the 2d-planes of $(\text{change in deaths}_{T-1}$, $\text{restriction stringency}_{T-1})$
and $(\text{vaccine coverage}_{T-1},$ $\text{change in cases}_{T-1}$).

Table \ref{tab: policy choice} summarizes the estimation results of Figures \ref{fig:Policy_k=00003D2}, \ref{fig:Policy_k=00003D3},
and \ref{fig:Policy_k=00003D4}. It shows that the thresholds of the T-EWM policies are insensitive to the additions of the variables that the quadrant policies depend on. As the number of variables increases, the threshold of the change in weekly deaths 
decreases slightly from $-745$ to $-1007.5$, while the other thresholds
remain stable. During the sample period, the mean of weekly cases is 594344.3. Therefore, the threshold in
the last column and the last row corresponds to an approximate $10\%$ fall relative to the
average number of cases during the sample period. In sum, T-EWM suggests that the policymaker should not relax restrictions
($W_{T}=1$) if there are no significant drops in deaths and cases, current stringency is comparatively low, and vaccine coverage is comparatively low.

\begin{table}[H]
	\caption{\label{tab: policy choice}Estimated optimal quadrant rules}
	
	\centering{}{\small{}}%
	\begin{tabular}{cccccc}
		\toprule 
		&  \#vars  & change in deaths & restriction & vaccine ($\%$) & change in cases\tabularnewline
		\midrule
		Figure \ref{fig:Policy_k=00003D2} & 2  & $>-745$ & $<62.6$ & -- & --\tabularnewline
		Figure \ref{fig:Policy_k=00003D3} & 3 & $>-1007.5$ & $<62.6$ & $<115.2$ & --\tabularnewline
		Figure \ref{fig:Policy_k=00003D4} & 4 & $>-1007.5$ & $<62.6$ & $<115.2$ & $>-50842.5$\tabularnewline
		\bottomrule
	\end{tabular}{\small\par}
\end{table}

Including more than four policy variables presents a significant computational challenge for the grid-search method we have used to optimize quadrant treatment rules. \footnote{When five variables are included in $X_{T-1}^{P}$, the estimated running time is 12 hours; with six variables, the estimated running time skyrockets to 4459 hours. (These estimates were made on a computer with a 12th Gen Intel(R) Core(TM) i7-1265U 2.7 GHz processor and 32GB of RAM.)} To overcome this issue, we employ decision trees to search for an optimal policy based on time-series empirical welfare.

The decision tree approach (\cite{breiman1984classification}) sequentially searches for policy variables and their threshold-based splits to maximize welfare. \cite{athey2021policy} and \cite{zhou2023offline} study the properties and implementation of this approach in maximizing the doubly robust empirical welfare criterion with brute force searches for tree partitions. \cite{ida2022choosing} implement a decision tree with a heuristic two-step optimization in the context of estimating a rebate assignment policy in an electricity market. We implement a classification tree algorithm based on \cite{hastie2009elements}. See Appendix \ref{app:emp_tree_algorithm} for the details of our implementation.


Figure  \ref{T-EWM tree} illustrates a policy rule obtained by a T-EWM decision tree, which considers all the seven policy variables within $X_{T-1}$, i.e., $ X_{T-1}^{P}=X_{T-1}=(\text{cases}_{T-1},\,\text{deaths}_{T-1},$  $\text{change in cases}_{T-1},\,\text{change in deaths}_{T-1},\,\text{restriction stringency}_{T-1},\,\text{vaccine coverage}_{T-1},$\\  $\text{economic conditions}_{T-1}).$ Due to the relatively small sample size of 92, we only generate trees with a depth of three.  
\begin{figure}[H]
	\caption{\label{T-EWM tree} T-EWM decision tree with seven policy variables}
	
	\centering
	
	\includegraphics[width=85mm]{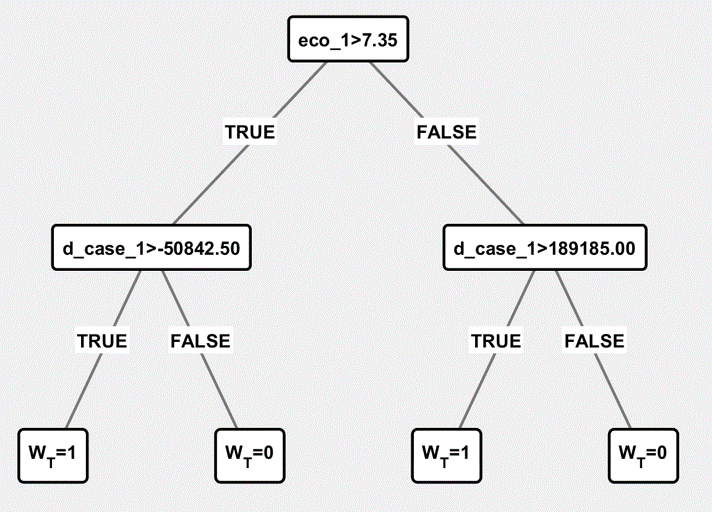}
	\noindent\begin{minipage}[t]{1\columnwidth}%
	\end{minipage}
\end{figure}

The result from the  decision tree indicates that if the economy performs relatively well, policymakers should adopt a more aggressive approach in maintaining or tightening Covid restrictions. 
(Note that the decision given by the left node of the second level aligns with the last one in the last row of Table \ref{tab: policy choice}, which considers four policy variables.) Conversely, if the economy performs relatively poorly, policymakers should refrain from further tightening Covid restrictions unless the increase in cases is significant. 

 To wrap up, we illustrate the use and implementation of the proposed T-EWM approach in making policy decisions. Furthermore, utilizing decision trees with the T-EWM approach allows for the inclusion of more policy variables while keeping computation time manageable.  More results from the T-EWM decision trees including a specification with a higher order of the Markovian structure of $q=2$ can be found in Appendix \ref{app_app}. 

\section{Conclusion}
This article proposes T-EWM, a framework and method for choosing optimal policies based on time-series data. We characterize assumptions under which this method can learn an optimal policy. We evaluate its statistical properties by deriving nonasymptotic upper bounds of the unconditional and conditional regret. We discuss its connections to the existing literature, including the Markov decision process and impulse response analysis. We present simulation results and empirical applications to illustrate the computational feasibility and applicability of T-EWM. As a benchmark formulation, this paper mainly focuses on a one-period social welfare function as a planner's objective. Extensions of the analysis to policy choices for the middle-run and long-run cumulative social welfare are left for future research.

\begin{spacing}{1.0} 
\bibliographystyle{abbrvnat}
\bibliography{biball}
\end{spacing}
\newpage
\appendix
\numberwithin{equation}{section}
\section{Supplemental  Appendix A}

\subsection{Proof of Proposition \ref{prop:simple_mds}} \label{app:proof_simple_mds}
\begin{proof}
We have 
\begin{align}
 & \mathsf{E}\left[\widehat{\mathcal{W}}_{t}(g|w)-\bar{\mathcal{W}}_{t}(g|w)|\mathcal{F}_{t-1}\right]\nonumber \\
 & =\mathsf{E}\left[1(W_{t-1}=w)\left[\frac{Y_{t}W_{t}g(W_{t-1})}{e_{t}(W_{t-1})}+\frac{Y_{t}(1-W_{t})\{1-g(W_{t-1})\}}{1-e_{t}(W_{t-1})}\right]|\mathcal{F}_{t-1}\right]\nonumber \\
 & -\mathsf{E}\left[1(W_{t-1}=w)\mathsf{E}\{Y_{t}(W_{0:t-1},1)g(W_{t-1})+Y_{t}(W_{0:t-1},0)(1-g(W_{t-1}))|W_{t-1}=w\}|\mathcal{F}_{t-1}\right]\nonumber \\
 & =:I_{t}-II_{t},\label{eq:MDS_check}
\end{align}

Under Assumption \ref{s_bounded_y}, it remains to show $I_{t}-II_{t}=0$.
By the Law of iterated expectation,
\begin{align}
I_{t} & =\mathsf{E}\left[1(W_{t-1}=w)\left[\frac{Y_{t}W_{t}g(W_{t-1})}{e_{t}(W_{t-1})}+\frac{Y_{t}(1-W_{t})\{1-g(W_{t-1})\}}{1-e_{t}(W_{t-1})}\right]|\mathcal{F}_{t-1}\right]\nonumber \\
 & =1(W_{t-1}=w)\left\{ \frac{g(W_{t-1})}{e_{t}(W_{t-1})}\mathsf{E}\left[Y_{t}W_{t}|\mathcal{F}_{t-1}\right]+\frac{1-g(W_{t-1})}{1-e_{t}(W_{t-1})}\mathsf{E}\left[Y_{t}(1-W_{t})|\mathcal{F}_{t-1}\right]\right\} .\label{eq:MDS_hat}
\end{align}
Note that 
\begin{align}
\mathsf{E}\left[Y_{t}W_{t}|\mathcal{F}_{t-1}\right] & =\mathsf{E}\left[Y_{t}(W_{0:t-1},W_{t})W_{t}|\mathcal{F}_{t-1}\right]\nonumber  =\mathsf{E}\left[Y_{t}(W_{t-1},W_{t})W_{t}|\mathcal{F}_{t-1}\right]\nonumber \\
 & =\mathsf{E}\left[Y_{t}(W_{t-1},1)W_{t}|\mathcal{F}_{t-1}\right]\nonumber  =\mathsf{E}\left[Y_{t}(W_{t-1},1)|\mathcal{F}_{t-1}\right]\mathsf{E}\left[W_{t}|\mathcal{F}_{t-1}\right]\nonumber \\
 & =\mathsf{E}\left[Y_{t}(W_{t-1},1)|\mathcal{F}_{t-1}\right]e_{t}(W_{t-1}),\label{eq:MDS_hat_y1}
\end{align}
where the second equality follows from Assumption \ref{ass:toy_example}(i),
the fourth equality follows from Assumption \ref{unconf}, and the
last equality follows from $\mathsf{E}\left[W_{t}|\mathcal{F}_{t-1}\right]=\mathsf{E}\left[W_{t}|W_{t-1}\right]=e_{t}(W_{t-1})$
by the second statement of Assumption \ref{ass:toy_example}(ii). Applying the same arguments,
we have 
\begin{equation}
\mathsf{E}\left[Y_{t}(1-W_{t})|\mathcal{F}_{t-1}\right]=\mathsf{E}\left[Y_{t}(W_{t-1},0)|\mathcal{F}_{t-1}\right]\left[1-e_{t}(W_{t-1})\right].\label{eq:MDS_hat_y2}
\end{equation}
Combining (\ref{eq:MDS_hat}), (\ref{eq:MDS_hat_y1}),  (\ref{eq:MDS_hat_y2}), and the Law of iterated expectation, we obtain
\[
I_{t}=1(W_{t-1}=w)\mathsf{E}\left[Y_{t}(W_{t-1},1)g(W_{t-1})+Y_{t}(W_{t-1},0)(1-g(W_{t-1}))|\mathcal{F}_{t-1}\right].
\]
On the other hand, 
\begin{align*}
II_{t} & =\mathsf{E}\left[1(W_{t-1}=w)\mathsf{E}\{Y_{t}(W_{0:t-1},1)g(W_{t-1})+Y_{t}(W_{0:t-1},0)(1-g(W_{t-1}))|W_{t-1}=w\}|\mathcal{F}_{t-1}\right]\\
 & =1(W_{t-1}=w)\mathsf{E}\left[Y_{t}(W_{t-1},1)g(W_{t-1})+Y_{t}(W_{t-1},0)(1-g(W_{t-1}))|W_{t-1}=w\right]\\
 & =1(W_{t-1}=w)\mathsf{E}\left[Y_{t}(W_{t-1},1)g(W_{t-1})+Y_{t}(W_{t-1},0)(1-g(W_{t-1}))|W_{t-1}\right]\\
 & =1(W_{t-1}=w)\mathsf{E}\left[Y_{t}(W_{t-1},1)g(W_{t-1})+Y_{t}(W_{t-1},0)(1-g(W_{t-1}))|\mathcal{F}_{t-1}\right],
\end{align*}
where the second equality follows from $\sigma\left(W_{t-1}\right)\subseteq\mathcal{F}_{t-1}$,
the third equality follows from the definition of the indicator function,
and the last equality follows from (\ref{simple_cond_exp}), which
is an implication of Assumption \ref{ass:toy_example}. Consequently,
we have $I_{t}-II_{t}=0$.\end{proof}

\subsection{Proof of Theorem \ref{thm:discrete_bound} } \label{simple_mds_proof}
\subsubsection{Preliminary Lemmas}
The following two lemmas will be used in the proofs of the main results.
\begin{lemma}[Freedman's inequality]\label{free} Let $\xi_{a,i}$
	be a martingale difference sequence indexed by $a \in \mathcal{A}$, $i=1, \dots, n$, $\mathcal{F}_{i}$ be the filtration,
 $V_{a}=\sum_{i=1}^{n}\E(\xi_{a,i}^{2}|\mathcal{F}_{i-1})$, and
	$M_{a}=\sum_{i=1}^{n}\xi_{a,i}$. {For positive numbers $A$ and $B$}, we have
	\begin{equation} \label{freedman}
		\Pr(\max_{a\in\mathcal{A}}|M_{a}|\geq z)\leq\sum_{i=1}^{n} \Pr(\max_{a\in\mathcal{A}}\xi_{a,i}\geq
		{A})+2\Pr(\max_{a\in\mathcal{A}}V_{a}\geq {B})+2|\mathcal{A}|e^{-z^{2}/{(2zA+2B)}}.
	\end{equation}
\end{lemma}
The proof can be found in \cite{freedman1975tail}. Next, let $a \lesssim b$ indicate that there is a positive constant $C$, such that $a\leq C\cdot b$.

\begin{lemma}[Maximal inequality based on Freedman's inequality]\label{max}
Let $\{\xi_{a,i}\}_{i=1}^{n}$ be a martingale difference sequence with
respect to a filtration $\{\mathcal{F}_{i}\}_{i=0}^{n}$ for each
$a\in\mathcal{A}$, where $\mathcal{A}$ is a finite index set. Set
$M_{a}=\sum_{i=1}^{n}\xi_{a,i}$ and $V_{a}=\sum_{i=1}^{n}\E(\xi_{a,i}^{2}\mid\mathcal{F}_{i-1})$.
Suppose that, for some positive constants $A$ and $B$,
$\max_{a\in\mathcal{A}}|\xi_{a,i}|\leq A$ for all $i$ and
$\max_{a\in\mathcal{A}}V_{a}\leq B$ almost surely. Then
\begin{equation} \label{Maximal_Freedman}
\E\!\left[\max_{a\in\mathcal{A}}|M_{a}|\right]
\lesssim A\log(1+|\mathcal{A}|)+\sqrt{B}\sqrt{\log(1+|\mathcal{A}|)}.
\end{equation}
\end{lemma}
\begin{proof}
This follows from Lemma~19.33 of \citet{van2000asymptotic} and
Lemma~\ref{free}. From Freedman's inequality, we have for each
$a\in\mathcal{A}$ and any $z>0$,
\begin{equation}\label{eq:two-regime}
\Pr\!\left(|M_a| \geq z\right)\leq
\begin{cases}
2\exp(-z/(4A)) & \text{if } z\geq B/A,\\
2\exp(-z^{2}/(4B)) & \text{if } z<B/A.
\end{cases}
\end{equation}
We decompose $M_a = C_a + D_a$ where
\[
C_a = M_a\mathbf{1}\{|M_a|>B/A\},\qquad
D_a = M_a\mathbf{1}\{|M_a|\leq B/A\}.
\]

By the same argument in \citet{van2000asymptotic}, 
\[
\mathrm{P}\left(\left|C_{a}\right|>z\right)\leq2\exp\left(\frac{-z}{4A}\right),\quad\mathrm{P}\left(\left|D_{a}\right|>z\right)\leq2\exp\left(\frac{-z^{2}}{4B}\right).
\]

Setting $\phi_{p}(x)=\exp\left(x^{p}\right)-1\text{ for }p=1,2$,
we have
\[
\mathsf{E}\phi_{1}\left(\left|C_{a}/12A\right|\right) \leq 1,\qquad\mathsf{E}\,\psi_{2}\!\left(\left|D_{a}/\sqrt{12B}\right|\right)\leq 1.
\]

Moreover, using $\phi^{-1}_{1}\left(u\right)=\log\left(1+u\right)$
and $\phi^{-1}_{2}\left(u\right)=\sqrt{\log\left(1+u\right)}$ and
rearranging yields
\[
\mathsf{E}\max_{a\in\mathcal{A}}\left|C_{a}\right|\leq12A\cdot\log\left(|\mathcal{A}|+1\right),\qquad\mathsf{E}\max_{a\in\mathcal{A}}|D_{a}|\le\sqrt{12B}\cdot\sqrt{\log(|\mathcal{A}|+1)}.
\]
Finally, with the triangle inequality $|M_{a}|\le|C_{a}|+|D_{a}|$,
we have
\begin{align*}
\mathsf{E}\max_{a\in\mathcal{A}}|M_{a}| & \le\mathsf{E}\max_{a\in\mathcal{A}}|C_{a}|+\mathsf{E}\max_{a\in\mathcal{A}}|D_{a}|\lesssim A\log(1+|\mathcal{A}|)+\sqrt{B}\sqrt{\log(1+|\mathcal{A}|)}.
\end{align*}
\end{proof}

\subsubsection{Proof of Theorem \ref{thm:discrete_bound}}
\begin{proof}
Define $\tilde{p}_w \defeq \frac{1}{T-1}\sum_{t=1}^{T-1}\Pr(W_{t-1}=w|\mathcal{F}_{t-2})$.
 Writing
$\frac{T(w)}{T-1}=\left(\frac{T(w)}{T-1}-\tilde{p}_w\right)+\tilde{p}_w$,
\begin{align}
\sup_{g:\{0,1\}\to\{0,1\}}\left|\widehat{\mathcal{W}}(g|w)-\bar{\mathcal{W}}(g|w)\right|
&=\sup_{g}\left(\frac{T(w)}{T-1}\right)^{-1}(T-1)^{-1}\left|\sum_{t=1}^{T-1}\left[\widehat{\mathcal{W}}_t(g|w)-\bar{\mathcal{W}}_t(g|w)\right]\right|.
\label{mds_decomp}
\end{align}
Throughout this section, we set $\widehat W(g|w)=0$ on the event $T(w)=0$. Since this event has exponentially small probability under strict overlap, this convention does not affect the rate. Define
\[
\frac{T(w)}{T-1}-\tilde{p}_w
=\frac{1}{T-1}\sum_{t=1}^{T-1}\xi_t,\qquad
\xi_t\defeq\mathbf{1}(W_{t-1}=w)-\Pr(W_{t-1}=w|\mathcal{F}_{t-2}),
\]
which is an average of a martingale difference sequence with respect
to $\{\mathcal{F}_{t-2}\}$. Since $|\xi_t|\leq1$, applying Freedman's
inequality (Lemma~\ref{free}) with $|\mathcal{A}|=1$, $A=1$, and
$B=\sum_{t}e(w|\mathcal{F}_{t-2})[1-e(w|\mathcal{F}_{t-2})]<T-1$,
\begin{align}\label{conv_weight}
\Pr\!\left(\left|\frac{1}{T-1}\sum_{t=1}^{T-1}\xi_t\right|\geq z\right)
&\leq 2\exp\!\left[\frac{-z^2(T-1)^2}{2(T-1)z+2(T-1)}\right],
\end{align}
which holds for every fixed $z>0$ and tends to zero as $T\to\infty$.
(The first term of \eqref{freedman} vanishes since
$|\xi_t|\leq1-\kappa<1=A$ by Assumption~\ref{bound1}.)
Define the event
\[
\mathcal{E}_T=\left\{\left|\frac{T(w)}{T-1}-\tilde{p}_w\right|\leq\kappa/2\right\}.
\]
By \eqref{conv_weight} with $z=\kappa/2$, $\Pr(\mathcal{E}_T^c)\leq 2\exp(-c(T-1))$
for some $c>0$. On $\mathcal{E}_T$, since $\tilde{p}_w\geq\kappa$ by
Assumption~\ref{bound1},
\begin{equation}\label{kappa_bound}
\left(\frac{T(w)}{T-1}\right)^{-1}
\leq\left(\tilde{p}_w-\left|\frac{T(w)}{T-1}-\tilde{p}_w\right|\right)^{-1}
\leq(\kappa-\kappa/2)^{-1}=\frac{2}{\kappa}.
\end{equation}

 Let
$\xi_{g,t}\defeq\widehat{\mathcal{W}}_t(g|w)-\bar{\mathcal{W}}_t(g|w)$.
By Proposition~\ref{prop:simple_mds}, $\{\xi_{g,t}\}$ is a martingale
difference sequence with respect to $\{\mathcal{F}_{t-1}\}$, i.e.,
$\E(\xi_{g,t}|\mathcal{F}_{t-1})=0$. With
$\mathcal{G}\equiv\{g:\{0,1\}\to\{0,1\}\}$ (a finite set), Assumptions
\ref{bound1} and \ref{s_bounded_y} give a constant $C_A$ with
$\sup_{g\in\mathcal{G}}|\xi_{g,t}|\leq M+M/\kappa<C_A$ for all $t$,
and a constant $C_B$ (depending only on $M,\kappa$) with
$V_g=\sum_{t=1}^{T-1}\E(\xi_{g,t}^2|\mathcal{F}_{t-1})<C_B(T-1)$.
Applying Lemma~\ref{max} with $A=C_A$ and $B=C_B(T-1)$,
\begin{align}\label{e_bound}
\E\!\left(\sup_{g\in\mathcal{G}}\left|\sum_{t=1}^{T-1}\xi_{g,t}\right|\right)
&\lesssim C_A\log(1+|\mathcal{G}|)+\sqrt{C_B(T-1)}\sqrt{\log(1+|\mathcal{G}|)}\nonumber\\
&\lesssim\sqrt{(T-1)\,C_B\log(1+|\mathcal{G}|)}.
\end{align}

 On $\mathcal{E}_T$, combining
\eqref{mds_decomp}, \eqref{kappa_bound}, and \eqref{e_bound},
\[
\E\!\left[\sup_{g\in\mathcal{G}}\left|\widehat{\mathcal{W}}(g|w)-\bar{\mathcal{W}}(g|w)\right|\mathbf{1}_{\mathcal{E}_T}\right]
\leq\frac{2}{\kappa}\frac{1}{T-1}\,\E\!\left(\sup_{g}\left|\sum_{t=1}^{T-1}\xi_{g,t}\right|\right)
\leq\frac{C}{\sqrt{T-1}},
\]
where $C$ depends only on $\kappa$, $|\mathcal{G}|$, and $M$. On the
complement $\mathcal{E}_T^c$, the boundedness
$|\widehat{\mathcal{W}}_t(g|w)|\leq M/\kappa$ and $|\bar{\mathcal{W}}_t(g|w)|\leq M$
give the deterministic bound
$\sup_g|\widehat{\mathcal{W}}(g|w)-\bar{\mathcal{W}}(g|w)|\leq 2M/\kappa$, so
\[
\E\!\left[\sup_{g}\left|\widehat{\mathcal{W}}(g|w)-\bar{\mathcal{W}}(g|w)\right|\mathbf{1}_{\mathcal{E}_T^c}\right]
\leq\frac{2M}{\kappa}\Pr(\mathcal{E}_T^c)
\leq\frac{4M}{\kappa}\exp(-c(T-1)),
\]
which is $o((T-1)^{-1/2})$. Adding the two contributions yields
\[
\E\!\left[
\sup_{g:\{0,1\}\to\{0,1\}}
\left|\widehat{\mathcal{W}}(g|w)-\bar{\mathcal{W}}(g|w)\right|
\right]
\leq\frac{C}{\sqrt{T-1}}.
\]
Consequently, since $\sup_g\E|X_g|\leq\E\sup_g|X_g|$, we have
\[
\sup_{g:\{0,1\}\to\{0,1\}}
\E\!\left|\widehat{\mathcal{W}}(g|w)-\bar{\mathcal{W}}(g|w)\right|
\leq\frac{C}{\sqrt{T-1}},
\]
which is the claimed bound \eqref{simp_rate}.
\end{proof}

\subsection{Higher/infinite Markov order} \label{app_multi_markov}
\subsubsection{The $q$-th ($q > 1$) Markov order}\label{app_finite_markov}
We modify Assumption \ref{ass:toy_example} to:\\
\textbf{Assumption \ref{ass:toy_example}{*}} ($q$-th order Markov
properties). For an integer $q>1$, the time series of potential outcomes
and observable variables satisfy the following conditions:

(i) $q$-th order Markovian exclusion: for $t=q+1,q+2,\dots,T$ and
for arbitrary treatment paths $(w_{0:t-q-1},w_{t-q:t})$ and $(w_{0:t-q-1}^{\prime},w_{t-q:t})$,
where $w_{0:t-q-1}\neq w_{0:t-q-1}^{\prime}$, 
\[
Y_{t}(w_{0:t-q-1},w_{t-q:t})=Y_{t}(w_{0:t-q-1}^{\prime},w_{t-q:t}):=Y_{t}(w_{t-q:t})
\]
 holds with probability one.

(ii) $q$-th order Markovian exogeneity: for $t=q,q+1,\dots,T$ and
any treatment path $w_{0:t}$, 
\[
Y_{t}(w_{0:t})\perp X_{0:t-1}|W_{t-q:t-1},
\]
 and for $t=q,q+1,\dots,T-1$, 
\[
W_{t}\perp X_{0:t-1}|W_{t-q:t-1}.
\]

Under these modified assumptions, the policy function $g$ is now
defined on $\{0,1\}^{q}$, mapping to $\{0,1\}$. For a vector $w\in\{0,1\}^{q}$,
the propensity score can be defined as $e_{t}(w)=\Pr(W_{t}=1|W_{t-q:t-1}=w)$.
Furthermore, we can redefine:{\footnotesize
\begin{align}
\mathcal{W}_{T}(g|w) & =\mathsf{E}\{Y_{T}(W_{T-q:T-1},1)g(W_{T-q:T-1})+Y_{T}(W_{T-q:T-1},0)(1-g(W_{T-q:T-1}))|W_{T-q:T-1}=w\},\nonumber \\
\widehat{\mathcal{W}}(g|w) & =\frac{1}{T(w)}\sum_{q\leq t\leq T-1:W_{t-q:t-1}=w}\left[\frac{Y_{t}W_{t}g(W_{t-q:t-1})}{e_{t}(W_{t-q:t-1})}+\frac{Y_{t}(1-W_{t})\{1-g(W_{t-q:t-1})\}}{1-e_{t}(W_{t-q:t-1})}\right],\nonumber \\
\bar{\mathcal{W}}(g|w) & =\frac{1}{T(w)}\sum_{q\leq t\leq T-1:W_{t-q:t-1}=w}\mathsf{E}\left[ Y_{t}(W_{t-q:t-1},1)g(W_{t-q:t-1})+Y_{t}(W_{t-q:t-1},0)\left[1-g(W_{t-q:t-1})\right]|W_{t-q:t-1}\right] .\label{eq:redef.welfare}
\end{align}
}These welfare functions will also be used in the subsequent Subsection \ref{rem: Inf_order}.

Moreover,
Assumption \ref{unconf} (sequential unconfoundedness) can be modified
to: for any $t=1,2,\dots,T-1$ and $w\in\{0,1\}$, 
\[
Y_{t}(W_{t-q:t-1},w)\perp W_{t}|X_{0:t-1}.
\]
Then, a convergence rate of $\frac{1}{\sqrt{T-q}}$ can be established
by similar arguments used for the proof of Theorem \ref{thm:discrete_bound}
(cf. Appendix \ref{simple_mds_proof}).

\subsubsection{Infinite Markov order} \label{rem: Inf_order}

Under the infinite Markov order, the current variables might depend on historical values that precede the first observations of the sample time series. To accommodate this concept, we change the notation from $X_{0:t}$ to $X_{-\infty:t}$ in this subsection.

If we allow all the historical treatments to affect the current
outcome, we can drop Assumption \ref{ass:toy_example}(i) completely.
For a concise and consistent discussion of the policy function, we maintain the first statement of Assumption \ref{ass:toy_example}(ii) and adjust it to accommodate the case of infinite Markov order:
\begin{equation}
Y_{t}(w_{-\infty:t})\perp X_{-\infty:t-1}|W_{-\infty:t-1}.\label{eq:infinite_A2.1(ii)}
\end{equation}
It means that for simplicity, we continue to exclude the past outcomes,
$Y_{-\infty:t-1}$, from the variables influencing the current outcome
$Y_{t}$. As a result, we can omit $Y_{-\infty:t-1}$ from the arguments
of the policy function $g$, thereby defining $g$ as a function mapping
from $\{0,1\}^{\infty}$ to $\{0,1\}$.

If the Markovian process has an infinite order, a valid empirical
welfare function should take the form of $\widehat{\mathcal{W}}(g|w_{-\infty:T-1})$,
which is infeasible. In practice, the planner can only use a truncated
treatment path to construct the empirical welfare. Let the truncation
be implemented at the period $T-m\in\{0:T-1\}$, and the policy is
given by 
\begin{equation} \label{eq:g_hat_inf}
\hat{g}\in\mbox{argmax}_{g:\{w_{T-m:T-1}\}\to\{0,1\}}\widehat{\mathcal{W}}(g|w_{T-m:T-1}), \end{equation}
i.e., the optimizer is based on the truncated Markov order $m$.

Let us (re)define the optimal policy as $g^*\in\mbox{argmax}_{g:\{w_{-\infty:T-1}\}\to\{0,1\}}\mathcal{W}_T(g|w_{-\infty:T-1}).$ Then, the population-level regret can be decomposed as:
\begin{align}
 & \mathcal{W}_{T}(g^{*}|w_{-\infty:T-1})-\mathcal{W}_{T}(\hat{g}|w_{-\infty:T-1})\nonumber  =\mathcal{W}_{T}(g^{*}|w_{T-m:T-1})-\mathcal{W}_{T}(\hat{g}|w_{T-m:T-1})\nonumber \\
 & +\mathcal{W}_{T}(g^{*}|w_{-\infty:T-1})-\mathcal{W}_{T}(g^{*}|w_{T-m:T-1})+\mathcal{W}_{T}(\hat{g}|w_{T-m:T-1})-\mathcal{W}_{T}(\hat{g}|w_{-\infty:T-1})\nonumber \\
 & \leq\mathcal{W}_{T}(g^{*}|w_{T-m:T-1})-\mathcal{W}_{T}(\hat{g}|w_{T-m:T-1})+2\sup_{g:\{w_{-\infty:T-1}\}\to\{0,1\}}|\mathcal{W}_{T}(g|w_{-\infty:T-1})-\mathcal{W}_{T}(g|w_{T-m:T-1})|\nonumber \\
 & \leq2c\sup_{g:\{w_{T-m:T-1}\}\to\{0,1\}}|\bar{\mathcal{W}}(g|w_{T-m:T-1})-\widehat{\mathcal{W}}(g|w_{T-m:T-1})|+2\cdot\text{w-bias}_{\infty}\left(m\right),\label{eq:infinite_decomp1}
\end{align}
where $\text{w-bias}_{\infty}\left(m\right):=\sup_{g:\{w_{-\infty:T-1}\}\to\{0,1\}}|\mathcal{W}_{T}(g|w_{-\infty:T-1})-\mathcal{W}_{T}(g|w_{T-m:T-1})|$. The last inequality follows by imposing an assumption similar
to Assumption \ref{equiv_W}. For the second term of the last row, we should have $\lim_{m\to\infty}\text{w-bias}_{\infty}\left(m\right)=0$
under mild regularity conditions. Thus, we can focus on the conditions required
for the convergence of $|\bar{\mathcal{W}}(g|w_{T-m:T-1})-\widehat{\mathcal{W}}(g|w_{T-m:T-1})|$.

Note that in the case of infinite Markov order, we must drop Assumption
\ref{ass:toy_example}(i) completely. As a result, $\bar{\mathcal{W}}(g|w_{T-m:T-1})-\widehat{\mathcal{W}}(g|w_{T-m:T-1})$
is no longer an average of an MDS since conditioning on $W_{t-m:t-1}$
is no longer equivalent to conditioning on $\mathcal{F}_{t-1}$ for
any $m<\infty$.

To address this issue, we define, for a treatment path $w_{T-m:T-1}\in\{0,1\}^{m}$,
{\scriptsize
\begin{equation}
\bar{\mathcal{W}}(g|\mathcal{F}_{t-1}) :=\frac{1}{T(w_{T-m:T-1})}\sum_{\substack{m\leq t\leq T-1:\\W_{t-m:t-1}=w_{T-m:T-1}}}\mathsf{E}\left[ Y_{t}(W_{-\infty:t-1},1)g(W_{t-m:t-1})+Y_{t}(W_{-\infty:t-1},0)\left[1-g(W_{t-m:t-1})\right]|\mathcal{F}_{t-1}\right].\label{eq:wel_bar_F_m}
\end{equation}
}Now, we can further decompose the first term of \eqref{eq:infinite_decomp1} into:{\footnotesize
\begin{align}
 & \sup_{g:\{w_{T-m:T-1}\}\to\{0,1\}}|\bar{\mathcal{W}}(g|w_{T-m:T-1})-\widehat{\mathcal{W}}(g|w_{T-m:T-1})|\nonumber \\
 & \leq\sup_{g:\{w_{T-m:T-1}\}\to\{0,1\}}|\bar{\mathcal{W}}(g|\mathcal{F}_{t-1})-\widehat{\mathcal{W}}(g|w_{T-m:T-1})|+\sup_{g:\{w_{T-m:T-1}\}\to\{0,1\}}|\bar{\mathcal{W}}(g|\mathcal{F}_{t-1})-\bar{\mathcal{W}}(g|w_{T-m:T-1})|\nonumber \\
 & =\sup_{g:\{w_{T-m:T-1}\}\to\{0,1\}}|\bar{\mathcal{W}}(g|\mathcal{F}_{t-1})-\widehat{\mathcal{W}}(g|w_{T-m:T-1})|+\overline{\text{w-bias}}_{\infty}\left(m\right),\label{eq:infinite_decomp2}
\end{align}
} where $\overline{\text{w-bias}}_{\infty}\left(m\right):=\sup_{g:\{w_{T-m:T-1}\}\to\{0,1\}}|\bar{\mathcal{W}}(g|\mathcal{F}_{t-1})-\bar{\mathcal{W}}(g|w_{T-m:T-1})|$.
Summarizing (\ref{eq:infinite_decomp1}) and (\ref{eq:infinite_decomp2}),
we obtain
\begin{align*}
 \mathcal{W}_{T}(g^{*}|w_{-\infty:T-1})-\mathcal{W}_{T}(\hat{g}|w_{-\infty:T-1})& \leq 2c\sup_{g:\{w_{m:T-1}\}\to\{0,1\}}|\bar{\mathcal{W}}(g|\mathcal{F}_{t-1})-\widehat{\mathcal{W}}(g|w_{T-m:T-1})|\\
 & +2\cdot \widetilde{\text{w-bias}}_{\infty}\left(m\right),
\end{align*}
where \begin{equation} \label{eq:inf_markov_bias}
\widetilde{\text{w-bias}}_{\infty}\left(m\right):=
 c\cdot\overline{\text{w-bias}}_{\infty}\left(m\right)+\text{w-bias}_{\infty}\left(m\right),
\end{equation}
In the following, we discuss the scenarios under which this bias term converges to zero. Let us strengthen Assumption \ref{unconf} to: for $w\in\{0,1\}$,
\begin{equation}
Y_{t}(W_{-\infty:t-1},w)\perp W_{t}|X_{-\infty:t-1}.\label{eq:infinite_unconf}
\end{equation}
As mentioned above, the first statement of Assumption \ref{ass:toy_example}(ii)
is maintained and modified for simplicity. The second statement of
Assumption \ref{ass:toy_example}(ii) is modified to
\begin{equation}
W_{t}\perp X_{-\infty:t-1}|W_{t-m:t-1}.\label{eq:ps_order_m}
\end{equation}
This requires that conditioning on $W_{t-m:t-1}$, the truncated history,
$X_{-\infty:t-m-1}$, does not have direct influences on $W_{t}$.
In other words, the propensity score has limited overlaps: $e_{t}(\mathcal{F}_{t-1})=\Pr(W_{t}=1|\mathcal{F}_{t-1})=\Pr(W_{t}=1|W_{t-m:t-1})$.


To summarize. First, under conditions (\ref{eq:infinite_unconf}) 
and (\ref{eq:ps_order_m}), the difference $\bar{\mathcal{W}}(g|\mathcal{F}_{t-1})-\widehat{\mathcal{W}}(g|w_{T-m:T-1})$
is an average of an MDS. For a given $m$, it converges at $\frac{1}{\sqrt{T-m}}$
rate. Second, under additional regularity conditions of decay temporal dependence, we have $\lim_{m\to\infty}$ $\text{w-bias}_{\infty}\left(m\right)=0$. 
Third, under the condition (\ref{eq:infinite_A2.1(ii)}), we have
 {\small\begin{align*}
		&\bar{\mathcal{W}}(g|\mathcal{F}_{t-1})=\frac{1}{T(w_{T-m:T-1})} \\
		 &\sum_{\substack{m\leq t\leq T-1:\\W_{t-m:t-1}=w_{T-m:T-1}}}\mathsf{E}\left[ Y_{t}(W_{-\infty:t-1},1)g(W_{t-m:t-1})+Y_{t}(W_{-\infty:t-1},0)\left[1-g(W_{t-m:t-1})\right]|W_{-\infty:t-1}\right].
	\end{align*}}
This equality leads to $\text{plim}_{m\to\infty}\overline{\text{w-bias}}_{\infty}\left(m\right)=0$, given that  $T$, $T(w_{T-m:T-1})$, and $m$ diverge at appropriate rates, as well as some other additional regularity conditions of decay temporal dependence are met. Consequently, $\text{plim}_{m\to\infty}\widetilde{  \text{w-bias}}_{\infty}\left(m\right)=0$.


\subsubsection{Example \ref{Example_inf_Markov} continued} \label{App:exmple_2}

Now we verify whether the model specified in Example 2 satisfies the
adjusted assumptions discussed in Appendix \ref{rem: Inf_order}. The assumptions adjusted for the case of infinite Markov
order are summarized below:
\begin{enumerate}
	
\item The condition \eqref{eq:infinite_unconf} of sequence unconfoundedness, $Y_{t}(W_{-\infty:t-1},w)\perp W_{t}|X_{-\infty:t-1}$, which represents the adjusted Assumption \ref{unconf}.

\item The condition \eqref{eq:ps_order_m}, $W_{t}\perp X_{-\infty:t-1}|W_{t-m:t-1}$, which
represents the second statement of the adjusted Assumption \ref{ass:toy_example}(ii).

\item The modified Assumption \ref{equiv_W} of the invariance of the welfare ordering, $\mathcal{W}_T(g^*|w_{T-m:T-1})- \mathcal{W}_T(g|w_{T-m:T-1}) \leq c\big[\bar{\mathcal{W}}(g^*|w_{T-m:T-1})-\bar{\mathcal{W}}(g|w_{T-m:T-1})\big]$.
\item The condition \eqref{eq:infinite_A2.1(ii)}, $Y_{t}(w_{-\infty:t})\perp X_{-\infty:t-1}|W_{-\infty:t-1}$,
which represents the first statement of the adjusted Assumption \ref{ass:toy_example}(ii).
\end{enumerate}
 Firstly, the condition \eqref{eq:infinite_unconf} is guaranteed by the independence between
$\varepsilon_{t}$ and $W_{t}$, conditional on $X_{-\infty:t-1}$. Secondly, the condition \eqref{eq:ps_order_m} is satisfied for any $m\geq 1$, as $W_{t}$ is an AR(1) process
and its distribution depends solely on $W_{t-1}$, as presented by \eqref{eq:Wt}.

Thirdly, regarding the invariance of welfare ordering, it holds that 
{\footnotesize\begin{align*}
&\mathcal{W}_{T}(g^{*}|w_{T-m:T-1})-\mathcal{W}_{T}(g|w_{T-m:T-1}) \\ &=\mathrm{E}\left[\alpha+\beta_{0}g^{*}+\sum_{i=1}^{\infty}\beta_{i}W_{T-i}+\sum_{i=0}^{\infty}\gamma_{i}\varepsilon_{T-i}|W_{T-m:T-1}=w_{T-m:T-1}\right]\\
 & -\mathrm{E}\left[\alpha+\beta_{0}g+\sum_{i=1}^{\infty}\beta_{i}W_{T-i}+\sum_{i=0}^{\infty}\gamma_{i}\varepsilon_{T-i}|W_{T-m:T-1}=w_{T-m:T-1}\right] =\beta_{0}\left(g^{*}-g\right),
\end{align*}}
 and {\footnotesize
\begin{align*}
 & \bar{\mathcal{W}}(g^{*}|w_{T-m:T-1})-\bar{\mathcal{W}}(g|w_{T-m:T-1})\\
 & =\frac{1}{T(w_{T-m:T-1})}\sum_{m\leq t\leq T-1:W_{t-m:t-1}=w_{T-m:T-1}}\mathrm{E}\left[\alpha+\beta_{0}g^{*}+\sum_{i=1}^{\infty}\beta_{i}W_{t-i}+\sum_{i=0}^{\infty}\gamma_{i}\varepsilon_{t-i}|W_{t-m:t-1}=w_{T-m:T-1}\right]\\
 & -\frac{1}{T(w_{T-m:T-1})}\sum_{m\leq t\leq T-1:W_{t-m:t-1}=w_{T-m:T-1}}\mathsf{\mathrm{E}}\left[\alpha+\beta_{0}g+\sum_{i=1}^{\infty}\beta_{i}W_{t-i}+\sum_{i=0}^{\infty}\gamma_{i}\varepsilon_{t-i}|W_{t-m:t-1}=w_{T-m:T-1}\right]\\
 & =\frac{1}{T(w_{T-m:T-1})}\sum_{m\leq t\leq T-1:W_{t-m:t-1}=w_{T-m:T-1}}\beta_{0}\left(g^{*}-g\right) =\beta_{0}\left(g^{*}-g\right),
\end{align*}}
where the second equality follows from the fact that the sequence $\{Y_{t},W_{t}\}$, which is generated from \eqref{eq:Yt_inf} and \eqref{eq:Wt} to \eqref{eq:Vt_2}, remains stationary.
Consequently, Assumption \ref{equiv_W} holds with $c=1$.

Finally, as mentioned at the beginning of Appendix \ref{rem: Inf_order}, the purpose of the condition \eqref{eq:infinite_A2.1(ii)}, which  excludes the effects of past values of $Y_t$, is to ensure a simple and consistent discussion throughout Section \ref{model_example}. Although it may not be satisfied by the outcome equation \eqref{eq:Yt_inf}, the linear structure of \eqref{eq:Yt_inf} ensures that under additional assumptions, our  T-EWM approach for the case of infinite Markov order, as presented in Appendix \ref{rem: Inf_order}, remains valid. 

 To see this point, note that in Appendix \ref{rem: Inf_order}, the condition \eqref{eq:infinite_A2.1(ii)} is only used for showing
$\text{plim}_{m\to\infty}\overline{\text{w-bias}}_{\infty}\left(m\right)=0$.
Recall that $\overline{\text{w-bias}}_{\infty}\left(m\right):=\sup_{g:\{w_{T-m:T-1}\}\to\{0,1\}}|\bar{\mathcal{W}}(g;m|\mathcal{F}_{t-1})-\bar{\mathcal{W}}(g|w_{T-m:T-1})|$.
In Example 2, this difference becomes 
{\footnotesize\begin{align*}
	& \bar{\mathcal{W}}(g;m|\mathcal{F}_{t-1})-\bar{\mathcal{W}}(g|w_{T-m:T-1})\\
	& =\frac{1}{T(w_{T-m:T-1})}\sum_{m\leq t\leq T-1:W_{t-m:t-1}=w_{T-m:T-1}}\mathsf{E}\left[\alpha+\beta_{0}g+\sum_{i=1}^{\infty}\beta_{i}W_{t-i}+\sum_{i=0}^{\infty}\gamma_{i}\varepsilon_{t-i}|\mathcal{F}_{t-1}\right]\\
	& -\frac{1}{T(w_{T-m:T-1})}\sum_{m\leq t\leq T-1:W_{t-m:t-1}=w_{T-m:T-1}}\mathsf{E}\left[\alpha+\beta_{0}g+\sum_{i=1}^{\infty}\beta_{i}W_{t-i}+\sum_{i=0}^{\infty}\gamma_{i}\varepsilon_{t-i}|W_{t-m:t-1}=w_{T-m:T-1}\right]\\
	& =\frac{1}{T(w_{T-m:T-1})}\sum_{m\leq t\leq T-1:W_{t-m:t-1}=w_{T-m:T-1}}\left[\sum_{i=m+1}^{\infty}\beta_{i}W_{t-i}+\sum_{i=1}^{\infty}\gamma_{i}\varepsilon_{t-i}\right]\\
	& -\frac{1}{T(w_{T-m:T-1})}\sum_{m\leq t\leq T-1:W_{t-m:t-1}=w_{T-m:T-1}}\mathsf{E}\left(\sum_{i=m+1}^{\infty}\beta_{i}W_{t-i}+\sum_{i=1}^{\infty}\gamma_{i}\varepsilon_{t-i}|W_{t-m:t-1}=w_{T-m:T-1}\right)\\
	& =\frac{1}{T(w_{T-m:T-1})}\sum_{m\leq t\leq T-1:W_{t-m:t-1}=w_{T-m:T-1}}\left[\sum_{i=m+1}^{\infty}\beta_{i}W_{t-i}-\mathsf{E}\left(\sum_{i=m+1}^{\infty}\beta_{i}W_{t-i}|W_{t-m:t-1}=w_{T-m:T-1}\right)\right]\\
	& +\frac{1}{T(w_{T-m:T-1})}\sum_{m\leq t\leq T-1:W_{t-m:t-1}=w_{T-m:T-1}}\left[\sum_{i=1}^{\infty}\gamma_{i}\varepsilon_{t-i}-\mathsf{E}\left(\sum_{i=1}^{\infty}\gamma_{i}\varepsilon_{t-i}|W_{t-m:t-1}=w_{T-m:T-1}\right)\right]\\
	& =:I_{m}+II_{m},
\end{align*}}
where the first term of the second equality follows from that the sum is over only those observations that satisfy $W_{t-m:t-1}=w_{T-m:T-1}$, and thus the path $W_{t-m:t-1}$ is fixed on $w_{T-m:T-1}$ within both $\bar{\mathcal{W}}(g;m|\mathcal{F}_{t-1})$ and $\bar{\mathcal{W}}(g|w_{T-m:T-1})$. 

Given that $\sum_{i=1}^{\infty}\left|\beta_{i}\right|<\infty$
holds and  $T$, $T(w_{T-m:T-1})$, and $m$ diverge at appropriate rates, we shall
have $\text{plim}_{m\to\infty}I_{m}=0$. For $II_{m}$, note that
it is an average of an MA($\infty$) process, centered around
its conditional mean given $W_{t-m:t-1}=w_{T-m:T-1}$. This should
converge to zero if 
$\max_{l\geq 1} \sum_{i\geq l }|\gamma_i|\leq k\cdot l^{-\alpha}$ holds for some positive $k$ and $\alpha$,
and the noise $\varepsilon_{t}$ is i.i.d.\ with zero mean and has a finite second
moment. 

\subsection{Proof of Theorem \ref{thm: kernel}} \label{max_kernel}
For simplicity, we maintain Assumption \ref{ass:continuous_Markov}, and one of its implications is equation \eqref{equiv_markov}.

In addition to \eqref{eq:sample_kernel} and \eqref{eq:kernel_bar}, we define
{\small \beaa
	\widetilde{\mathcal{W}}_{h}(g, x) 
	& =& \frac{\sum_{t=1}^{T-1}K_h(X_{t-1},x)\E\big[\frac{Y_{t}W_{t}}{e_{t}(X_{t-1})}g(X_{t-1})+\frac{Y_{t}(1-W_{t})}{1-e_{t}(X_{t-1})}[1-g(X_{t-1})]|\mathcal{F}_{t-1}\big]}{\sum_{t=1}^{T-1}K_h(X_{t-1},x)}\\
	&=&\frac{\sum_{t=1}^{T-1}K_h(X_{t-1},x)\mathcal{W}_t(g|\mathcal{F}_{t-1})}{\sum_{t=1}^{T-1}K_h(X_{t-1},x)},
\eeaa}
where $K(\cdot)$ is a bounded kernel with a bounded support, $K_h(a,b):=\frac{1}{h}K(\frac{a-b}{h})$. The second equality follows from Assumption \ref{unconf_c}.

Our strategy is to show for any $x\in \mathcal{X}$ and any $g:\mathcal{X}\to \{0,1\}$, such that
\begin{align*}
	&\mathcal{W}_{T}(g|x)-\mathcal{W}_{T}(\hat{g}_x|x)  \\
	&\leq c\big[\bar{\mathcal{W}}_{h}(g|x)-\bar{\mathcal{W}}_{h}(\hat{g}_x|x)\big]\\
	& \leq c\big[	\widetilde{\mathcal{W}}_{h}(g,x)-\widetilde{\mathcal{W}}_{h}(\hat{g}_x,x)\big]+O_{p}(h^{2})+O_{p}(c_{w}^{-1}(\sqrt{(T-1)h})^{-1})\\
	& \leq \sup_{g:\mathcal{X} \to \{0,1\}} 2c|	\widetilde{\mathcal{W}}_{h}(g,x)-\widehat{\mathcal{W}}(g|x) |+O_{p}(h^{2})+O_{p}(c_{w}^{-1}(\sqrt{(T-1)h})^{-1})\\
	&=O_{p}(h^{2})+O_{p}(c_{w}^{-1}(\sqrt{(T-1)h})^{-1}),  
\end{align*}
where the first inequality follows from Assumption \ref{equiv_W_c}. The second inequality follows from Lemma \ref{lem:kernel_bias} below. The third inequality follows from similar arguments to \eqref{ineq_s}. The last equality follows from Lemma \ref{lem:kernel_variance} stated below.

\begin{assumption}[Kernel]\label{kernel}
$K(\cdot)$ is a symmetric bounded kernel function with bounded support and
$\int K(v)\,dv = 1$. Let $K_h(u,x):=h^{-1}K((u-x)/h)$. There exist
constants $c_w, c_m, c'_m > 0$ such that for all $x\in\mathcal{X}$,
all $t$, and almost surely:
\begin{itemize}
\item[(i)] $\E[K_h(X_{t-1},x)\mid\mathcal{F}_{t-2}] \geq c_w$;
\item[(ii)] $\E[K_h(X_{t-1},x)^2\mid\mathcal{F}_{t-2}] \leq c_m\, h^{-1}$;
\item[(iii)] $\E[K_h(X_{t-1},x)^4\mid\mathcal{F}_{t-2}] \leq c'_m\, h^{-3}$.
\end{itemize}
The bandwidth satisfies $h\to 0$ and $(T-1)h\to\infty$ as $T\to\infty$.
\end{assumption}

\begin{lemma} \label{lem:kernel_bias}
Under Assumption \ref{kernel}, for any  $g:\mathcal{X}\to \{0,1\}$ and $x \in \mathcal{X}$, 
\[
\widetilde{\mathcal{W}}_{h}(g, x) -\bar{\mathcal{W}}_{h}(g|x)=O_{p}(h^{2})+O_{p}\left(c_{w}^{-1}(\sqrt{(T-1)h})^{-1}\right).
\]
\end{lemma}
\begin{proof}
Under Assumption \ref{kernel}, 
\begin{align*}
	&	\bigg[\widetilde{\mathcal{W}}_{h}(g, x) -\bar{\mathcal{W}}_{h}(g|x)\bigg]\bigg[\frac{1}{T-1}\sum_{t=1}^{T-1}K_h(X_{t-1},x)\bigg]  \\ &=\frac{1}{T-1}\sum_{t=1}^{T-1}K_h(X_{t-1},x)(\mathcal{W}_t(g|X_{t-1})-\mathcal{W}_t(g|x))\\
	& =\frac{1}{T-1}\sum_{t=1}^{T-1}\bigg\{K_h(X_{t-1},x)(\mathcal{W}_t(g|X_{t-1})-\mathcal{W}_t(g|x))\\
	&-\E[K_h(X_{t-1},x)(\mathcal{W}_t(g|X_{t-1})-\mathcal{W}_t(g|x))|\mathcal{F}_{t-2}]\bigg\}\\
	& +\frac{1}{T-1}\sum_{t=1}^{T-1}\bigg\{\E[K_h(X_{t-1},x)(\mathcal{W}_t(g|X_{t-1})-\mathcal{W}_t(g|x))|\mathcal{F}_{t-2}]\bigg\}.
	\end{align*}
Rearranging the equation, we have
\beaa
\widetilde{\mathcal{W}}_{h}(g, x) -\bar{\mathcal{W}}_{h}(g|x)=O_{p}(c_{w}^{-1}(\sqrt{(T-1)h})^{-1})+O_{p}(h^{2}),
\eeaa
where the first term on the right-hand side follows from similar arguments in  the proof for Theorem \ref{thm:mds_bound} and the second term follows from the standard result concerning the bias of the kernel estimator.
\end{proof}

\begin{lemma} \label{lem:kernel_variance}
Suppose $(T-1)h\to\infty$. Under Assumption \ref{kernel} and the
assumption that $\mathcal{G}$ has finite VC dimension $v_{\mathcal{G}}$,
for each fixed $x\in\mathcal{X}$,
\begin{equation}
\sup_{g\in\mathcal{G}}
\left|\widehat{\mathcal{W}}(g|x) - \widetilde{\mathcal{W}}_h(g, x)\right|
= O_p\!\left(c_w^{-1}\sqrt{\frac{v_{\mathcal{G}}}{(T-1)h}}\right).
\end{equation}
\end{lemma}
\begin{proof}
Fix $x\in\mathcal{X}$. We have
\[
\widehat{\mathcal{W}}(g|x) - \widetilde{\mathcal{W}}_h(g,x)
=\frac{\sum_{t=1}^{T-1} K_h(X_{t-1},x)\,\xi_t(g)}{\sum_{t=1}^{T-1} K_h(X_{t-1},x)},
\qquad
\xi_t(g):=\widehat{\mathcal{W}}_t(g)-\mathcal{W}_t(g|\mathcal{F}_{t-1}),
\]
and bound the numerator and denominator separately.

Define $Q_T(x):=\sum_{t=1}^{T-1}K_h(X_{t-1},x)^2$. The increments
$K_h(X_{t-1},x)^2-\E[K_h(X_{t-1},x)^2|\mathcal{F}_{t-2}]$ form a
martingale difference sequence bounded by $h^{-2}\|K\|_\infty^2$ with
conditional second moment $\lesssim h^{-3}$ (since $K$ is bounded with
bounded support). Since $x$ is fixed, Lemma~\ref{max} gives
\[
\left|Q_T(x)-\sum_{t=1}^{T-1}\E[K_h(X_{t-1},x)^2|\mathcal{F}_{t-2}]\right|
=O_p\!\left(h^{-2}+\sqrt{(T-1)h^{-3}}\right)=o_p\!\left((T-1)/h\right),
\]
using $(T-1)h\to\infty$, which makes both $h^{-2}$ and
$\sqrt{(T-1)h^{-3}}$ negligible relative to $(T-1)/h$. Since
$\sum_t\E[K_h(X_{t-1},x)^2|\mathcal{F}_{t-2}]\leq (T-1)c_m h^{-1}$
(Assumption~\ref{kernel}), the event
$\mathcal{A}_T:=\{Q_T(x)\leq 2c_m(T-1)h^{-1}\}$ satisfies
$\Pr(\mathcal{A}_T)\to1$.

For fixed $g$,
$\zeta_t(g):=K_h(X_{t-1},x)\,\xi_t(g)$ is a martingale difference
sequence with respect to $\{\mathcal{F}_{t-1}\}$, with increments
$|\zeta_t(g)|\leq h^{-1}\|K\|_\infty(M/\kappa)\lesssim h^{-1}$ and
predictable quadratic variation
$\sum_t\E[\zeta_t(g)^2|\mathcal{F}_{t-1}]\leq C_\xi\, Q_T(x)$,
where $C_\xi$ bounds $\E[\xi_t(g)^2|\mathcal{F}_{t-1}]$ under
Assumptions \ref{ass bounded y} and \ref{bound_c}. On
$\mathcal{A}_T$ the predictable quadratic variation is at most
$B:=2C_\xi c_m(T-1)h^{-1}$. Applying Freedman's inequality
(Lemma~\ref{free}) on $\mathcal{A}_T$ and adding $\Pr(\mathcal{A}_T^c)=o(1)$
gives the fixed-$g$ bound with increment bound $A\asymp h^{-1}$ and
$B\asymp (T-1)h^{-1}$. Since $\mathcal{G}$ has finite VC dimension
$v_{\mathcal{G}}$, a standard VC-type maximal inequality for martingale
arrays then yields the uniform bound over $g$,
\[
\sup_{g\in\mathcal{G}}
\left|\sum_{t=1}^{T-1} K_h(X_{t-1},x)\,\xi_t(g)\right|
=O_p\!\left(h^{-1}v_{\mathcal{G}}+\sqrt{\frac{(T-1)\,v_{\mathcal{G}}}{h}}\right)
=O_p\!\left(\sqrt{\frac{(T-1)\,v_{\mathcal{G}}}{h}}\right),
\]
where the second equality uses $(T-1)h\to\infty$, which makes the linear
term $h^{-1}v_{\mathcal{G}}$ negligible relative to
$\sqrt{(T-1)v_{\mathcal{G}}/h}$.

The summands
$\eta_t(x):=K_h(X_{t-1},x)-\E(K_h(X_{t-1},x)|\mathcal{F}_{t-2})$
form a martingale difference sequence with $|\eta_t(x)|\leq 2h^{-1}\|K\|_\infty$
and $\E[\eta_t(x)^2|\mathcal{F}_{t-2}]\leq h^{-1}c_m$. Since $x$ is fixed,
Lemma~\ref{max} with $|\mathcal{A}|=1$ gives
\[
\left|(T-1)^{-1}\sum_{t=1}^{T-1}\eta_t(x)\right|
=O_p\!\left(\frac{1}{(T-1)h}+\frac{1}{\sqrt{(T-1)h}}\right)=o_p(1),
\]
under $(T-1)h\to\infty$. By Assumption~\ref{kernel},
$\E(K_h(X_{t-1},x)|\mathcal{F}_{t-2})\geq c_w$ uniformly in $t$, so
\[
(T-1)^{-1}\sum_{t=1}^{T-1} K_h(X_{t-1},x)\geq c_w/2
\]
with probability approaching one.

On the intersection of $\mathcal{A}_T$ and
the denominator event (which has probability approaching one), dividing
the numerator bound by the denominator's lower bound,
\[
\sup_{g\in\mathcal{G}}
\left|\widehat{\mathcal{W}}(g|x) - \widetilde{\mathcal{W}}_h(g,x)\right|
=O_p\!\left(\frac{(T-1)^{-1}\sqrt{(T-1)h^{-1}v_{\mathcal{G}}}}{c_w/2}\right)
=O_p\!\left(c_w^{-1}\sqrt{\frac{v_{\mathcal{G}}}{(T-1)h}}\right),
\]
which is the claimed bound.
\end{proof}


\subsection{Justifying Assumption \ref{cover1}} \label{justify_entropy}
Here we interpret the entropy condition in Assumption \ref{cover1} in a more general way. We follow the argument of Chapter 11 in \cite{kosorok2008introduction} for the functional class related to non-i.i.d. data.

First, we illustrate the case that the stochastic process, $$\{D_{t}\}_{t=-\infty}^{\infty}\overset{\text{def}}{=}\{Y_{t},W_{t},X_{t-1}\}_{t=-\infty}^{\infty}, $$
is assumed to be stationary. In this case, we have for all $t\in\{1,2,\dots,n\}$,
\begin{align*}
	h_{t}(\cdot;G)  =h(\cdot;G)
	\mathcal{\text{ and }H}  =\{h(\cdot;G)=\widehat{\mathcal{W}}_{t}(G)-\bar{\mathcal{W}}_{t}(G):G\in\mathcal{G}\}.
\end{align*}
For an $n$-dimensional non-negative vector $\alpha_{n}=\{\alpha_{n,1},\alpha_{n,2},\dots,\alpha_{n,n}\}$,
define $Q_{\alpha_{n}}$ as a discrete measure with
probability mass $\frac{\alpha_{n,t}}{\sum_{t=1}^{n}\alpha_{n,t}}$ on the value
$D_{t}$. Recall for a function $f$, its $L_{r}(Q)-$norm is denoted
by $||f||_{Q,r}\overset{\text{def}}{=}\left[\int_{v}|f(v)|^{r}dQ(v)\right]^{1/r}$.
Thus, given a sample $\{D_{t}\}_{t=1}^{n}$ and $h(\cdot;G)$, we have $||h(\cdot;G)||_{Q_{\alpha_{n}},2}=\left[\frac{1}{\sum_{t=1}^{n}\alpha_{n,t}}\sum_{t}\alpha_{n,t}h(D_{t};G)^{2}\right]^{1/2}.$

In the stationary case, we define the \emph{restricted} function class $\mathbf{H_{n}}=\{h_{1:n}:h_{1:n}\in\mathcal{\mathcal{H}}\times\mathcal{\mathcal{H}}\times\dots\times\mathcal{\mathcal{H}},h_{1}=h_{2}=\dots=h_{n}\}$
and the envelope function $\overline{H}{}_{n}=(H_{1},H_{2},\cdots,H_{n})^{\prime}$.
Furthermore, let $\mathcal{Q}$ denote the class of all discrete probability measures
on the domain of the random vector $D_{t}=\{Y_{t},W_{t},X_{t-1}\}$.
Recall for an $n$-dimensional vector $v=\left\{ v_{1},\dots,v_{n}\right\} $,
its $l_{2}-$norm is denoted by $|v|_{2}\overset{\text{def}}{=}\left(\sum_{i=1}^{n}v_{i}^{2}\right)^{1/2}$.
Then, for any $\alpha_{n}\in\mathbb{R}_{+}^{n}$ and $\tilde{\alpha}_{n,t}=\frac{\sqrt{\alpha_{n,t}}}{\sqrt{\sum_{t}\alpha_{n,t}}}$, we have
\[
\mathcal{N}(\delta|\tilde{\alpha}_{n}\circ\overline{H}{}_{n}|_{2},\tilde{\alpha}_{n}\circ\mathbf{H}_{n},|.|_{2})=\mathcal{N}(\delta||H||_{Q_{\alpha_{n}},2},\mathcal{H},||.||_{Q_{\alpha_{n}},2})\leq\sup_{Q\in\mathcal{Q}}\mathcal{N}(\delta||H||_{Q,2},\mathcal{H},||.||_{Q,2}).
\]

In light of this relationship in the stationary case, we generalize the setup. Let $\mathcal{K}$ be a subset of $\{1,\cdots,n\}$, and its dimension is $K=|\mathcal{K}|$. Let $\alpha_{n,K}$ denote the $K$-dimensional sub-vector of $\alpha_n$ corresponding
to the index set $\mathcal{K}.$ Recall that $H_{t}$ denotes the envelope function of $\mathcal{\mathcal{H}}_{t}$, 
the functional class corresponding to $\{h_{t}(.,G),G\in\mathcal{G}\}$, and $\mathbf{H_{n}}=\mathcal{\mathcal{H}}_{1}\times \mathcal{\mathcal{H}}_{2}\times \dots \times \mathcal{H}_{n}.$ 
Then, for the subset $\mathcal{K}$, we similarly define $\mathbf{H_{n,K}}$ as the corresponding functional
class, and the vector of its envelope functions as $\overline{H}{}_{n,K}$.

Assumption~\ref{cover1} should be viewed as a high-level restriction
that rules out an accumulation of complexity across time. To
illustrate when such a restriction is plausible, consider the
following structural condition: for any fixed $K$, the covering number
of the product class can effectively be reduced to its $K$-dimensional
components, i.e.,
\begin{align*}
	\mathcal{N}(\delta|\tilde{\alpha}_{n}\circ\overline{H}{}_{n}|_{2},\mathbf{\tilde{\alpha}_{n}\circ H_{n},}|.|_{2}) & \leq\max_{\mathcal{K}\subseteq\{1:n\}}\mathcal{N}(\delta|\tilde{\alpha}_{n,K}\circ\overline{H}{}_{n,K}|_{2},\tilde{\alpha}_{n,K}\mathbf{\circ H_{n,K},}|.|_{2}),\\
	& \leq \max_{\mathcal{K}\subseteq \{1:n\}}\prod_{t\in\mathcal{K}}\mathcal{N}(\delta|\tilde{\alpha}_{n,t}\cdot {H}_{t}|_{2},\tilde{\alpha}_{n,t}\cdot \mathcal{\mathcal{H}}_{t},|.|_{2}),\\
	& \leq\sup_{Q\in\mathcal{Q}}\max_{t\in\mathcal{K}}\max_{\mathcal{K}\subseteq\{1:n\}}\mathcal{N}(\delta\|H_{t}\|_{Q,2},\mathcal{\mathcal{H}}_{t},\|.\|_{Q,2})^{K}.
\end{align*}

This structural condition is itself a restriction rather than a
universal property: in general, the covering number of a
product/triangular array class grows with $n$ or with $|\mathcal{K}|$,
and reducing it to a one-dimensional bound is a strong assumption that
requires additional structure on the dependence.
The special case $|\mathcal{K}|=1$ in particular is illustrative
rather than automatic: under this case, the first inequality of
\eqref{cover} in Assumption \ref{cover1} corresponds to verifying the
one-dimensional covering number $\mathcal{N}(\delta\|H_{t}\|_{Q,2},\mathcal{\mathcal{H}}_{t},\|.\|_{Q,2})$,
which we examine below for a specific class.

To provide an example for the second inequality of \eqref{cover}, we let $\IF(X_{t-1} \in G)=\IF( X_{t-1}^{\top}\theta \leq 0)$ for some $\theta \in \Theta$, where $\Theta$ is a compact set in $\mathbb{R}^d$. Without loss of generality, we assume that $\Theta$ is the unit ball in $\mathbb{R}^d$, i.e., for all $\theta\in \Theta$: $|\theta |_2\leq 1$.
For $w=0$ and 1, define $S_{t,w} \defeq Y_t(w)\IF(W_t = w)/\Pr(W_t=w|X_{t-1})$, where $Y_t(w)$ is the abbreviation of $Y_t(W_{t-1},w)$. Under Assumption \ref{ass bounded y}, $|Y_t(1)|, |Y_t(0)|\leq M/2$, so we have $\|H_t\|_{Q,r} \leq M+ M/\kappa \overset{\text{def.}}{=} M'$, and
\begin{eqnarray*}
	h^{\theta}_{ t}& = &\E(Y_t(1)\IF(X_{t-1}^{\top}\theta \leq 0)|X_{t-1})+ \E(Y_t(0)\IF(X_{t-1}^{\top}\theta >0)|X_{t-1})\\&&- S_{t,1}\IF(X_{t-1}^{\top}\theta \leq 0) -  S_{t,0}\IF(X_{t-1}^{\top}\theta > 0).
\end{eqnarray*}
The corresponding functional class can be written as
\begin{eqnarray*}
	\mathcal{H}_t = \{h: (y_t,w_t, x_{t-1}) \to f_{1,1}^{\theta} + f_{1,0}^{\theta} + f_{0,1}^{\theta}+f_{0,0}^{\theta}, \theta \in \Theta\},
\end{eqnarray*}
where $f_{1,1}^{\theta}$ (resp. $f_{1,0}^{\theta}, f_{0,1}^{\theta}, $ and $f_{0,0}^{\theta} $) corresponds to $\E(Y_t(1)\IF(X_{t-1}^{\top}\theta \leq 0)|X_{t-1})$ (resp.\\ $\E(Y_t(0)\IF(X_{t-1}^{\top}\theta>0)|X_{t-1}),-S_{t,1}\IF(X_{t-1}^{\top}\theta \leq 0)$, and $- S_{t,0}\IF(X_{t-1}^{\top}\theta > 0)$). Let the corresponding functional class be denoted by $\mathcal{F}_{1,1}$ (resp. $\mathcal{F}_{1,0},\mathcal{F}_{0,1},$ and $\mathcal{F}_{0,0})$. For all finitely discrete norms $Q$ and any positive $\vps$,
we know that
\begin{equation}
	\small
	\sup_Q\mathcal{N}(\vps, \mathcal{H}_t, \|.\|_{Q,r}) \leq \sup_Q\mathcal{N}(\vps/4, \mathcal{F}_{1,1}, \|.\|_{Q,r})\mathcal{N}(\vps/4, \mathcal{F}_{1,0}, \|.\|_{Q,r})\mathcal{N}(\vps/4, \mathcal{F}_{0,1}, \|.\|_{Q,r})\mathcal{N}(\vps/4, \mathcal{F}_{0,0},\|.\|_{Q,r}).
\end{equation}
We look at the covering number of the respective functional class. According to Lemma 9.8 of \cite{kosorok2008introduction}, the subgraph of the function $\IF (X_{t-1}^{\top}\theta \leq  0)$ is of VC dimension less than $d+2$ since the class $\{x \in \mathbb{R}^d, x^{\top}\theta \leq  0,\theta \in \Theta\}$ is of VC dimension less than $d+2$ (see the proof of Lemma 9.6 of \cite{kosorok2008introduction}). Therefore, we have $\sup_Q\mathcal{N}(\vps/4, \mathcal{F}_{0,1},\|.\|_{Q,r}) \vee \mathcal{N}(\vps/4, \mathcal{F}_{0,0}, \|.\|_{Q,r}) \lesssim (4/(\vps M'))^{d+2}$.

Moreover, we impose the following {Lipschitz} condition on functions $f^{\theta}_{1,1}$ and $f^{\theta}_{1,0}$: For any distinct points $\theta, \theta' \in \Theta$ and a positive constant $M_d$, it holds that $||f^{\theta}_{1,1}-f^{\theta'}_{1,1}||_{Q,r}\leq M_d |\theta - \theta'|_r$ (a similar equality holds for $f^{\theta}_{1,0})$. Then, it falls within the type II class defined in \cite{andrews1994empirical}, so according to the derivation of (A.2) in \cite{andrews1994empirical} we have
\begin{equation}
	\sup_Q\mathcal{N}(\vps M', \mathcal{F}_{1,1}, \|.\|_{Q,r})\leq \sup_{Q} \mathcal{N}(\vps M'/M_d, \Theta, \|.\|_{Q,r}),
\end{equation}
where the latter is the covering number of an Euclidean ball under the norm $\|.\|_{Q,r}$. Thus, according to Equation (5.9) in \cite{wainwright2019high}, $\sup_{Q} \mathcal{N}(\vps M'/M_d, \Theta, \|.\|_{Q,r}) \lesssim (1+ 2M_d/\vps M')^{d}$. Combining the above results, we have $\sup_Q \mathcal{N} (\varepsilon M', \mathcal{H}_t, \|.\|_{Q,r})\lesssim (4/(\vps M'))^{2(d+2)} (1+ 2M_d/(\vps M'))^{2d}$. Finally, with some rearrangement and redefinition of constant terms, we can obtain \eqref{cover}.

Beyond this specific example, verifying Assumption~\ref{cover1} in
more general non-stationary or dependent settings would typically
require either explicit dependence-adjusted entropy bounds  or restrictions on the temporal complexity
of $\mathbf{H}_n$. We therefore impose Assumption~\ref{cover1} as a
primitive structural condition.

\subsection{Proof of Theorem \ref{thm:mds_bound} } \label{proof_MDS_t}
We first show the following lemma. 

\begin{lemma} \label{bound}
Under Assumptions \ref{ass:continuous_Markov} to \ref{ass bounded y}, and \ref{cover1} to \ref{norm}, we have
\beqq \E [|\E{_n} h|_{\mathbf{H_{n}}}]\lesssim M \sqrt{v/n}. \eeqq
\end{lemma}
It shall be noted that the result of the lemma above is of the maximal inequality type and has a standard $\sqrt{n}^{-1}$ rate. The complexity of the function class $v$ also plays a role. This is in line with other results in the literature, such as \cite{kitagawa2018should}.
\begin{proof}
	 $h_t$ denotes a function belonging to the functional class $\mathcal{H}_t$, and  $h = \left\{ h_{1},h_{2},\dots,h_{n}\right\}$. $J_k$ is a cover of the functional class $\mathbf{H_n}$ with radius $2^{-k}M$ with respect to the $\rho_{2,n}(.)$ norm, and $k = 1,\cdots, \overline{K}$. We set $2^{-\overline{K}}\asymp \sqrt{n}^{-1}$, then $\overline{K} \asymp \log(n)$. Recall that $M$ is the constant defined in Assumption \ref{ass bounded y}, which implies $\max_t|h_t|\leq M$. 
	 We define $h^* = \argmax_{h\in \mathbf{H_n}} \E_n h_{.}$. Let $h^{(k)} = \min_{h \in J_k}\rho_{2,n}(h, h^*)$ and $h^{(0)} = (0,\cdots, 0)\in \mathbb{R}^n$, then $\rho_{2,n}(h^{(k)}, h^*) \leq 2^{-k}M$ holds by the definition of $J_k$, and
	\begin{equation} \label{eq: net_width}
		\rho_{2,n}(h^{(k-1)}, h^{(k)})\leq \rho_{2,n}(h^{(k-1)}, h^*) + \rho_{2,n}(h^{(k)}, h^*) \leq 3\cdot 2^{-k}M.
	\end{equation}
	By a standard chaining argument, we express any partial sum of  $h \in \mathbf{H_n}$ as a telescoping sum,
	\begin{equation}
		\sum^n_{t=1} h_t  \leq  |\sum^n_{t=1}h^{(0)}_{t}| + |\sum^{\overline{K}}_{k=1} \sum^n_{t=1}(h_t^{(k)} - h_t^{(k-1)})|+|\sum^n_{t=1}(h_t^{(\overline{K})} - h_t^{*})|.
	\end{equation}
	The inequality $|\sum_t a_t|\leq \sum_t |a_t|\leq |\sum_t a_t^2|^{1/2}\sqrt{n}$ can be applied to the third term. Notice that, by the definition of the $h^{(\overline{K})}$,
	\begin{equation}
		|\sum^n_{t=1}(h_t^{(\overline{K})} - h_t^{*})| \leq |(\sum^n_{t=1}(h_t^{(\overline{K})} - h_t^{*})^{2})^{1/2}|\sqrt{n}\leq {n}2^{-\overline{K}}M.
	\end{equation}
	Thus,
	\begin{eqnarray} \label{chaining}
		 \E(|\E{_n} h|_{\mathbf{H_{n}}}) \leq \sum_k^{\overline{K}-1}\E \max_{f \in J_k , g\in J_{k-1}, \rho_{2,n}(f,g) \leq 3\cdot 2^{-k}\cdot M} |\E{_{n}}(f-g)|+  2^{-\overline{K}}M.
	\end{eqnarray}
Apply Lemma \ref{max} and Assumption \ref{norm} to \eqref{chaining}. The maximal inequality \eqref{Maximal_Freedman} of Lemma \ref{max} is reproduced here:
\begin{equation*}
	\E(\max_{a \in \mathcal{A}} |M_a|)  \lesssim A \log(1+|\mathcal{A}|)+\sqrt{B}\sqrt{\log(1+|\mathcal{A}|)},
\end{equation*}
where for the first term of \eqref{chaining}, we have $\mathcal{A}=\{f-g:f \in J_k , g\in J_{k-1}, \rho_{2,n}(f,g) \leq 3\cdot 2^{-k}\cdot M\}$,  $|\mathcal{A}|=|J_{k}||J_{k-1}| \leq 2 \mathcal{N}^2(2^{-k}M, \mathbf{H_n}, \rho_{2,n}(.))\lesssim_p  2\max_t \sup_Q\mathcal{N}^2(2^{-k}M, \mathcal{H}_t, \|.\|_{Q,2})$, and $A\leq 3M$.
$B$ in \eqref{Maximal_Freedman} is an upper bound of the sum of conditional variances of an MDS. By Assumption \ref{norm}, we have  $B=  \sum_t \E[(f_t-g_t)^2| \mathcal{F}_{t-1}]\leq nL^2 \rho_{2,n}(f,g)^2  \leq  nL^2 (3\cdot 2^{-k}M)^2$ for any pair $(f,g)$ satisfying $f-g\in \mathcal{A}$.
	
	Therefore,  by  Lemma \ref{max}, 
	\begin{eqnarray*}
	n\E(|\E{_n} h|_{\mathbf{H_{n}}}) 
		&& \lesssim \sum^{\overline{K}}_{k=1} (L*3* 2^{-k} M \sqrt{n}) \sqrt{\log (1+ 2\max_t\sup_Q\mathcal{N}^2(2^{-k}M, \mathcal{H}_t, \|.\|_{Q,2}))}\\
		&& + 3* M \sum^{\overline{K}}_{k=1}\log (1+ 2 * \max_t\sup_Q\mathcal{N}^2(2^{-k}M, \mathcal{H}_t, \|.\|_{Q,2}))+o_p(\sqrt{n})\\
		&&\lesssim 6\sqrt{n}\int^1_{0} M \sqrt{\log ( 2^{1/2} \max_t\sup_Q\mathcal{N}(2^{-k}M, \mathcal{H}_t, \|.\|_{Q,2}))}d\varepsilon.
	\end{eqnarray*}
		By Assumption \ref{cover1}, $\max_t \log \sup_Q\mathcal{N}(\varepsilon M, \mathcal{H}_t,\|.\|_{Q,2}) \lesssim  \log(K)+ \log (v +1)+ (v+1)(\log4+ 1)+(crv)\log(\frac{2}{\varepsilon M})$. Thus, the integral in the last row is finite by a standard argument for bracketing numbers (see, e.g., the comment following Theorem 19.4 in \citeauthor{van2000asymptotic}, \citeyear{van2000asymptotic}). Then, we have
	$\E(|\E{_n} h|_{\mathbf{H_{n}}}) \lesssim M \sqrt{v/n}.$
\end{proof}
The next lemma concerns the tail probability bound. It states that, under certain regularity conditions,  $|\E{_n} h|_{\mathbf{H_{n}}}$ is very close to $ \E(|\E{_n} h|_{\mathbf{H_{n}}})$.
\begin{lemma}\label{cont}
Under Assumptions \ref{ass:continuous_Markov} to \ref{ass bounded y}, and \ref{cover1} to \ref{norm},
\beqq
|\E{_n} h|_{\mathbf{H_{n}}} - \E(|\E{_n} h|_{\mathbf{H_{n}}}) \lesssim_p M c_n\sqrt{v/n},
\eeqq
\end{lemma}
where $c_n$ is an arbitrarily slowly growing sequence. 
\begin{proof}
Similar to the above derivation, for a positive constant $\eta_{k}$, with $\sum_k\eta_{k}\leq 1$,
	\begin{eqnarray*}
		&&{\Pr(n^{-1}\sum^n_{t=1} h_t\geq x)}
		\leq  \Pr( n^{-1} |\sum^{\overline{K}}_{k=1} \sum^n_{t=1}h_t^{(k)} - h_t^{(k-1)}|\geq x- \sqrt{n}^{-1}2^{-\overline{K}}M )\\
		&\leq &\sum^{\overline{K}}_{k=1} \Pr(  |n^{-1} \sum^n_{t=1}h_t^{(k)} - h_t^{(k-1)}|\geq \eta_k(x- \sqrt{n}^{-1}2^{-\overline{K}}M) )\\
		&\leq& \sum^{\overline{K}}_{k=1}\exp\{\log \max_t\sup_Q\mathcal{N}^2(2^{-k}M, \mathcal{H}_t, \|.\|_{Q,2}) -\eta_{k}^2(nx- \sqrt{n}2^{-\overline{K}}M) ^2/[2\{(nx- \sqrt{n}2^{-\overline{K}}M)\\&&+ 2((3\cdot 2^{-k}\cdot M)^2n)\}] \}\\
		&\leq & \sum^{\overline{K}}_{k=1}\exp(  \log(K)+ \log (v +1)+ (v+1)(\log4+ 1)+(2v)\log(\frac{2}{2^{-k}M})
		\\&&-\eta_{k}^2(nx- \sqrt{n}2^{-\overline{K}}M) )^2/(2((nx- \sqrt{n}2^{-\overline{K}}M)+ 2((3\cdot 2^{-k}\cdot M)^2n)) ),
	\end{eqnarray*}
where the above derivation is due to the tail probability in Lemma \ref{free}.
	We pick $\eta_k$ and $x$ to ensure the right-hand side converges to zero and $\sum_k\eta_{k} \leq 1$. 
	
	We take  $ b_k=  \log(\overline{K})+ \log (v +1)+ (v+1)(\log4+ 1)+(2v)\log(\frac{2}{2^{-k}M})$, $a_k =2^{-1} (nx- \sqrt{n}2^{-\overline{K}}M) )^2/((nx- \sqrt{n}2^{-\overline{K}}M)+ 2((3*2^{-k}*M)^2n))$. We pick $\eta_k \geq \sqrt{a_k/b_k}$, so that $b_k \leq \eta_k^2 a_k$.
	We also need to choose $x$ to ensure that $\sum_k \eta_k \leq 1$ and $\sum_k \exp(b_k - \eta_k ^2 a_k) \to 0$.	We pick $ c_n \sqrt{v/n}\lesssim x$, and $\eta_k = c_n' \sqrt{b_k/a_k}$, with two slowly growing functions $c_n$ and $c_n'$ such that $c_n'\ll c_n$. We set $x = \E(|\E{_n} h|_{\mathbf{H_{n}}})  + c_n \sqrt{v/n} $. The result then follows.
\end{proof}
Finally, Theorem \ref{thm:mds_bound} follows by combining Lemma \ref{bound}, Lemma \ref{cont}, and $n=T-1$.

\subsection{Proof of Theorem \ref{mean_bound}} \label{Proof_of_mean_bound}
In this proof, we will show the concentration of $I=\frac{1}{T-1}\sum_{t=1}^{T-1}\overline{S}_t(G)+\frac{1}{T-1}\sum_{t=1}^{T-1}\tilde{S}_t(G)$. Under Assumptions \ref{cover1} and \ref{D.4} and with similar arguments as those for Lemma \ref{cont}, the first term of $I$ satisfies:
\begin{eqnarray}
\mbox{sup}_{G\in \mathcal{G}}\left |\frac{1}{T-1}\sum_{t=1}^{T-1}\overline{S}_t(G) \right|\lesssim_p M \sqrt{v}/{\sqrt{T-1}}.
\end{eqnarray}
The concentration rate of the second term of $I$, $\frac{1}{T-1}\sum_{t=1}^{T-1}\tilde{S}_t(G)$, is shown through the following two steps:

(i) For  a finite function class $\mathcal{G}$ with $|\mathcal{G}|= \tilde{M}<\infty$ and  under Assumptions  \ref{D.1}- \ref{D.3}, 
\begin{equation} \label{eq:finite_ex_bound}
\sup_{G\in \mathcal{G}}\left|\frac{1}{T-1} \sum_{t=1}^{T-1} \tilde{S}_{t}(G)\right| \lesssim_p \frac{c_T [\log (\tilde{M}) 2e\gamma ]^{1/\gamma}\sup_G\Phi_{\phi_{\tilde{v}}}(\tilde{S}_{.}(G)) }{\sqrt{T-1}}
\end{equation}
holds with probability $1- \exp(-c_T^\gamma)$,  where $c_T$ is a large enough constant, and $\Phi_{\phi_{\tilde{v}}}(\tilde{S}_{.}(G))$ is defined in Assumption  \ref{D.3}.

(ii) Next, we extend \eqref{eq:finite_ex_bound} to obtain the main result of Theorem \ref{mean_bound}: Let $\mathcal{G}$ be a function class with infinite elements, and its complexity is subject to Assumption \ref{D.4}. Under Assumptions  \ref{D.1}- \ref{D.4}, we have
\begin{equation*}
\sup_{G\in \mathcal{G}} \left|\frac{1}{T-1} \sum_{t=1}^{T-1} \tilde{S}_{t}(G)\right| \lesssim_p \frac{c_T [V\log (T) 2e\gamma ]^{1/\gamma}\sup_G \Phi_{\phi_{\tilde{v}}}(\tilde{S}_{.}(G))}{\sqrt{T-1}}.
\end{equation*}

 \paragraph{Step (i).}
	 By Assumption \ref{ass:continuous_Markov} and  \ref{D.2}(i), we have $\E(S_{t}(G)| \mathcal{F}_{t-2})=\E(S_{t}(G)| X_{t-2})$.  By Assumption \ref{D.1}, we have
  $X_t =g_t(\varepsilon_t,\varepsilon_{t-1},\cdots).$
Define $X_{t-2,l}^*$ as a version of the random variable $X_{t-2}$, in which $\varepsilon_{t-2-l}$ is replaced by $\varepsilon_{t-2-l}^*$:
By definition of $\tilde{S}(G)$, we have $$\Phi_{\phi_{\tilde{v}}}(\tilde{S}_{.}(G))= \sup_{q\geq 2}(\sum_{l\geq 0} \max_t\|\E(S_{t}(G)| X_{t-2})  - \E(S_{t}(G)| X_{t-2,l}^*)\|_{q})/q^{\tilde{v}} .$$ 
	%
	Then, by Assumption \ref{D.3}, we have $\sup_{G}\Phi_{\phi_{\tilde{v}}}(\tilde{S}_{.}(G)) < \infty$. For $\gamma = 1/(1+2\tilde{v})$ and a finite $\mathcal{G}$ with $|\mathcal{G}|=\tilde{M}$, the bound \eqref{eq:finite_ex_bound} follows by 	 Theorem 3 of \cite{wu2016performance}:
	\begin{equation}
		\Pr(\sup_{G\in \mathcal{G}}  \sum_{t=1}^{T-1} \tilde{S}_{t}(G) \geq x) \leq \tilde{M}\exp\left[-\left(\frac{x}{\sqrt{T-1}\sup_G\Phi_{\phi_{\tilde{v}}}(\tilde{S}_{.}(G))}\right)^\gamma\frac{1}{2e\gamma} \right].
	\end{equation}
	%
	Specifically, let $x = c_T [\log (\tilde{M}) 2e \gamma]^{1/\gamma}\Phi_{\phi_{\tilde{v}}}(\tilde{S}_{.}(G)) \sqrt{T-1}$, where $c_T$ is a sufficiently large constant,
	then
	\begin{equation}
		\sup_{G\in \mathcal{G}}\left|\frac{1}{T-1} \sum_{t=1}^{T-1} \tilde{S}_{t}(G)\right|\lesssim c_T [\log (\tilde{M}) 2e\gamma]^{1/\gamma} \sup_{G\in \mathcal{G}}\Phi_{\phi_{\tilde{v}}}(\tilde{S}_{.}(G)) /\sqrt{T-1}
	\end{equation}
	holds with probability $1- \exp(-c_T^\gamma)$.
	\paragraph{Step (ii).}
		We now consider the case where $\mathcal{G}$ is not finite. We define $\mathcal{G}^{(1)\delta}$ to be  a $\delta \max_t\mbox{sup}_Q\|\tilde{F}_{t}\|_{Q,2}$-net within $\tilde{\mathbf{F}}_n$ w.r.t.\ $\mathcal{G}$.
 {We denote $\tilde{S}_{t} (\pi(G))$ as the closest component to $\tilde{S}_{t}(G)$ in the net $\mathcal{G}^{(1)\delta}$.}
	Then,
	\begin{eqnarray*}
		\sup_{G} \left| \frac{1}{T-1} \sum_t \tilde{S}_{t}(G)\right|
		&\leq& \sup_{G\in \mathcal{G}}\left| \frac{1}{T-1} \sum_{t=1}^T\left[ \tilde{S}_{t}(G) - \tilde{S}_{t} (\pi(G))\right] \right|+ \sup_{G\in \mathcal{G}^{(1)\delta}} |\frac{1}{T-1}  \sum_{t=1}^{T-1}( \tilde{S}_{t}(G))  | \\
		&\lesssim_p & \delta \mbox{max}_t\mbox{sup}_Q\|\tilde{F}_{t}\|_{Q,2}+   \sup_{G\in \mathcal{G}^{(1)\delta}} |\frac{1}{T-1} \sum_{t=1}^{T-1}( \tilde{S}_{t}(G))  | \\
		&\lesssim_p& \delta 
  \mbox{max}_t\mbox{sup}_Q     \|\tilde{F}_{t}\|_{Q,2}+  c_T [V\log (1/\delta)2e\gamma]^{1/\gamma}  \sup_{G\in \mathcal{G}}\Phi_{\phi_{\tilde{v}}}(\tilde{S}_{.}(G))/\sqrt{T-1},
	\end{eqnarray*}
	where $V$ and $\delta$ following the first $\lesssim_p$ are the constants in the first statement of Assumption \ref{D.4}. Recall that  $\sup_{G\in \mathcal{G}}\Phi_{\phi_{\tilde{v}}}(\tilde{S}_{.}(G))$ is finite by Assumption \ref{D.3}.
	By setting $\delta = \frac{1}{T}$, we obtain
	\begin{equation}
		\sup_{G \in \mathcal{G}} \frac{1}{T-1} \sum_{t=1} \tilde{S}_{t}(G) \lesssim_p c_T [V\log (T) 2e\gamma ]^{1/\gamma} \sup_{G\in \mathcal{G}}\Phi_{\phi_{\tilde{v}}}(\tilde{S}_{.}(G)) /\sqrt{T-1}.
	\end{equation}

\subsection{Proof of Lemma \ref{lem:lower_RT(G)}} \label{app:proof_uncon_con_bound}
\begin{proof} We have
\begin{align}
        R_{T}(G) & =\int R_{T}(G|x)dF_{X_{T-1}} (x)\nonumber \\
        & =\int_{x\in A(x^{obs},G)}R_{T}(G|x)dF_{X_{T-1}}(x)+\int_{x\not\in A(x^{obs},G)}R_{T}(G|x)dF_{X_{T-1}}(x)\nonumber \\
        & \geq R_{T}(G|x^{obs})\cdot p_{T-1}(x^{obs},G)+0 =R_{T}(G|x^{obs})\cdot p_{T-1}(x^{obs},G).\label{eq:unc_con_bound_con_dis}
    \end{align}

The first inequality follows from the definition of $A(x^{\prime},G)$ and $R_{T}(G|x)$ being non-negative (by Assumption \ref{ass:correct_specify}). Then, Assumption \ref{lower_p_dis_con} yields $  R_{T}(G|x^{obs})  \leq  \frac{1}{\underline{p}} R_{T}(G)$.
\end{proof}

\subsection{Accounting for the Lucas critique} \label{MDP_B}
Here we solve the VAR reduced form of the three-equation New Keynesian model discussed in Section \ref{sec:lucas critique}. 

Recall that $\tilde{Y}_{t}:=\left( \begin{array}{c}
	x_t\\
	\pi_t
\end{array}\right)$ and $d_{t}:=\left(\begin{array}{c}
	v_{t}\\
	\varepsilon_{t}
\end{array}\right)$. Rearranging \eqref{eq:three_equation} yields
\bea
\left(\begin{array}{c}
	\E_{t}x_{t+1}\\
\E_{t}\pi_{t+1}
\end{array}\right)=\left(\begin{array}{cc}
	\kappa/(\sigma\beta)+1 & \delta/\sigma-1/(\sigma\beta)\\
	-\kappa/\beta & 1/\beta
\end{array}\right)\left(\begin{array}{c}
	x_{t}\\
	\pi_{t}
\end{array}\right)+\left(\begin{array}{c}
	v_{t}/\sigma+\varepsilon_{t}/(\sigma\beta)\\
	-\varepsilon_{t}/\beta
\end{array}\right).\label{eq:Keynesian_causal}\eea

Define
\[
N=\left(\begin{array}{cc}
	\kappa/(\sigma\beta)+1 & \delta/\sigma-1/(\sigma\beta)\\
	-\kappa/\beta & 1/\beta
\end{array}\right),\quad
C=-\left(\begin{array}{cc}
1/\sigma & 1/\sigma\beta\\
0 & -1/\beta
\end{array}\right).
\]
By assuming that $N$ is invertible, we can define
$A=N^{-1}$. Recall that  $d_{t}$ is assumed to be an AR(1) process:  $d_{t}=Fd_{t-1}+\eta_{t}$,
where
$
F=\left(\begin{array}{cc}
	\rho & 0\\
	0 & \gamma
\end{array}\right).
$
Then, equation \eqref{eq:Keynesian_causal} can be written as
$\tilde{Y}_{t}=A\E{_{t}}\tilde{Y}_{t+1}+ACd_{t}.$ Solving forward, we obtain $\tilde{Y}_{t}=\lim_{L\to\infty}A^{L}\E{_{t}}(\tilde{Y}_{t+L})+\sum_{l\geq0}A^{l+1}C\E_{t}\left(d_{t+l}\right).$

{ By assuming  $||A||<1$}, we have $\lim_{L\to\infty}A^{L}\E{_{t}}(\tilde{Y}_{t+L})\to_{a.s.}0.$
Thus, $\tilde{Y}_{t}$ can be solved as
\begin{equation}
	\tilde{Y}_{t}=M(\rho)d_{t},\label{eq:Yt_dt}
\end{equation}
where $M(\rho):=\sum_{l\geq 0} A^{l+1}CF^{l}=\sum_{l\geq 0} A^{l+1}C\left(\begin{array}{cc}
	\rho & 0\\
	0 & \gamma
\end{array}\right)^l$. {Assuming that $M(\rho)$ is invertible, by iterating  \eqref{eq:Yt_dt},
  we can solve the
VAR reduced form:
$$ \tilde{Y}_{t}=M\left(\rho\right)F\left[M\left(\rho\right)\right]^{-1}\tilde{Y}_{t-1}+M\left(\rho\right)\eta_{t}.$$
}

\subsection{Proof of Theorem \ref{thm:e_propensity}} \label{sec:proof_e_propensity}

	Under policy $G$, we use the following notation: $\widetilde{\mathcal{W}}(G)$ is defined in \eqref{2intermediate_welfare}; $\widehat{\mathcal{W}}(G)$  represents the estimated welfare defined in \eqref{eq:unconditional_sample}; $\widehat{\mathcal{W}}^{\hat{e}}(G)$  represents the estimated welfare, with the estimated propensity score $\hat{e}(\cdot)$; and
\beaa 
\widehat{\mathcal{W}}^{\hat{e}}(G)
=\frac{1}{T-1}\sum_{t=1}^{T-1} \left[\frac{Y_{t}W_{t}}{\hat{e}_t(X_{t-1})}1(X_{t-1}\in G)+ \frac{Y_{t}(1-W_{t})}{1- \hat{e}_t(X_{t-1})}1(X_{t-1}\notin G)\right].
\eeaa
	Recall that $ G_*$ is the optimal policy defined in \eqref{maxG}.	Let $\hat{G}^{\hat{e}}$ be the optimal policy estimated using the
	estimated propensity score $\hat{e}(\cdot)$, 
	\begin{equation}
		\hat{G}^{\hat{e}}\in \underset{G\in\mathcal{G}}{\text{argmax
		}}\widehat{\mathcal{W}}^{\hat{e}}(G).\label{eq:wp_e_hat}
	\end{equation}
	
	Recall $\tau_{t}=\frac{Y_{t}W_{t}}{e_t(X_{t-1})}- \frac{Y_{t}(1-W_{t})}{1- e_t(X_{t-1})}$ and $\hat{\tau}_{t}=\frac{Y_{t}W_{t}}{\hat{e}_t(X_{t-1})}- \frac{Y_{t}(1-W_{t})}{1- \hat{e}_t(X_{t-1})}$. Similar to (A.29) in the supplementary material for \citet{kitagawa2018should}, we have
	\begin{align}
		& \widetilde{\mathcal{W}}( G_*)-\widetilde{\mathcal{W}}(\hat{G}^{\hat{e}})\nonumber \\
		&
		=\widetilde{\mathcal{W}}( G_*)-\widetilde{\mathcal{W}}(\hat{G}^{\hat{e}})+\left[\widehat{\mathcal{W}}^{\hat{e}}( G_*)-\widehat{\mathcal{W}}^{\hat{e}}( G_*)\right]+\left[\widehat{\mathcal{W}}(\hat{G}^{\hat{e}})-\widehat{\mathcal{W}}(\hat{G}^{\hat{e}})\right]+\left[\widehat{\mathcal{W}}( G_*)-\widehat{\mathcal{W}}( G_*)\right]\nonumber
		\\
		&
		\leq\widetilde{\mathcal{W}}( G_*)-\widetilde{\mathcal{W}}(\hat{G}^{\hat{e}})+\left[\widehat{\mathcal{W}}^{\hat{e}}(\hat{G}^{\hat{e}})-\widehat{\mathcal{W}}^{\hat{e}}( G_*)\right]+\left[\widehat{\mathcal{W}}(\hat{G}^{\hat{e}})-\widehat{\mathcal{W}}(\hat{G}^{\hat{e}})\right]+\left[\widehat{\mathcal{W}}( G_*)-\widehat{\mathcal{W}}( G_*)\right]\nonumber
		\\
		&
		=\left[\widehat{\mathcal{W}}( G_*)-\widehat{\mathcal{W}}^{\hat{e}}( G_*)-\widehat{\mathcal{W}}(\hat{G}^{\hat{e}})+\widehat{\mathcal{W}}^{\hat{e}}(\hat{G}^{\hat{e}})\right]
			+\left[\widetilde{\mathcal{W}}( G_*)-\widetilde{\mathcal{W}}(\hat{G}^{\hat{e}})-\widehat{\mathcal{W}}( G_*)+\widehat{\mathcal{W}}(\hat{G}^{\hat{e}})\right]\nonumber
		\\
		& =I^{\hat{e}}+II^{\hat{e}},\label{eq:decompo W}
	\end{align}
	where the first inequality comes from
	$\widehat{\mathcal{W}}^{\hat{e}}(\hat{G}^{\hat{e}})\geq\widehat{\mathcal{W}}^{\hat{e}}( G_*)$,
	which is implied by the definition \eqref{eq:wp_e_hat}.
	
	For $II^{\hat{e}}$, we know $II^{\hat{e}}\leq2\sup_{G\in\mathcal{G}}|\widetilde{\mathcal{W}}(G)-\widehat{\mathcal{W}}(G)|$. Similar arguments to Section \ref{sec:uncon_con} can then be used to bound it. For $I^{\hat{e}}$, note that for any $G\in\mathcal{G}$, 
	\bea \label{eq:tau_welfare}
	\widehat{\mathcal{W}}(G)
    =\frac{1}{T-1}\sum_{t=1}^{T-1} \left[\tau_t 1(X_{t-1}\in G)+ \frac{Y_{t}(1-W_{t})}{1- e_t(X_{t-1})} \right]. 
	\eea
	Similarly,
	\bea \label{eq:hat_tau_welfare}
	\widehat{\mathcal{W}}^{\hat{e}}(G) 
    =\frac{1}{T-1}\sum_{t=1}^{T-1} \left[\hat{\tau}_t 1(X_{t-1}\in G)+ \frac{Y_{t}(1-W_{t})}{1- \hat{e}_t(X_{t-1})} \right].
	\eea
	
	Combining \eqref{eq:tau_welfare} and \eqref{eq:hat_tau_welfare} with $I^{\hat{e}}$,
	\beaa
			&\widehat{\mathcal{W}}( G_*)-\widehat{\mathcal{W}}(\hat{G}^{\hat{e}})&=\frac{1}{T-1}\sum_{t=1}^{T-1}\tau_{t}\left[\mathbf{1}\{X_{t-1}\in G_*\}-\mathbf{1}\{X_{t-1}\in\hat{G}^{\hat{e}}\}\right] \\
		&\widehat{\mathcal{W}}^{\hat{e}}( G_*)-\widehat{\mathcal{W}}^{\hat{e}}(\hat{G}^{\hat{e}})&=\frac{1}{T-1}\sum_{t=1}^{T-1}\hat{\tau}_{t}\left[
		\mathbf{1}\{X_{t-1}\in G_*\}-\mathbf{1}\{X_{t-1}\in\hat{G}^{\hat{e}}\}\right].
	\eeaa
	Then,
		\begin{align*}
		I^{\hat{e}} & 
		=\frac{1}{T-1}\sum_{t=1}^{T-1}\left[\left(\tau_{t}-\hat{\tau}_{t}\right)\cdot1\{X_{t-1}\in G_*\}-\left(\tau_{t}-\hat{\tau}_{t}\right)\cdot1\{X_{t-1}\in\hat{G}^{\hat{e}}\}\right]\\
		&
		=\frac{1}{T-1}\sum_{t=1}^{T-1}\left[\left(\tau_{t}-\hat{\tau}_{t}\right)\left(1\{X_{t-1}\in G_*\}-1\{X_{t-1}\in\hat{G}^{\hat{e}}\}\right)\right]\leq\frac{1}{T-1}\sum_{t=1}^{T-1}|\tau_{t}-\hat{\tau}_{t}|.
	\end{align*}
		Finally, we have that the rate of convergence is bounded by the accuracy of propensity score estimation and the bound with known propensity scores,
	\begin{align*}
		\E{_{P_{T}}}[\mathcal{W}_T(G_*)-\mathcal{W}_T(\hat{G}^{\hat{e}})] & \leq
		\E{_{P_{T}}}\left[\frac{1}{T-1}\sum_{t=1}^{T-1}|\tau_{t}-\hat{\tau}_{t}|\right]
		+2\E{_{P_{T}}}\left[\sup_{G\in\mathcal{G}}|\widetilde{\mathcal{W}}(G)-\widehat{\mathcal{W}}(G)|\right].
	\end{align*}
	The statement of Theorem \ref{thm:e_propensity} follows from \eqref{eq:ass phi} and Theorem \ref{thm:mds_mean_bound_unconditional}.

\

\subsection{Link to Markov decision problems} \label{MDP_A}

\textcolor{red}{}
In this section, we show the connection between our T-EWM setup and models of the Markov Decision Process (MDP).
For MDP, we adapt the notation of \cite{kallenberg2016markov} (LK hereafter), an online set  of lecture notes by Lodewijk
Kallenberg. {As described in LK, the MDP is the set of models for making decisions for dependent data. An MDP typically has components $\{[p_{ij}(a)]_{i,j}, r_{i}^{t}(a),W_{t-1}\}$. In period $t$, $W_{t-1}$ is the state and $a$ is the action. The agent chooses their decision according to a policy (a map from the state $W_{t-1}$  to an action $a$). They then receive a reward $r_{i}^{t}(a)$. The reward function depends on the transition probabilities of a Markov process, which are determined by the action $a$. Thus their action affects the reward via its effect on the transition probability matrix.  The optimal policy is estimated by optimizing an aggregated reward function.
In this subsection, we show a formal link between our T-EWM framework and an MDP. In particular, we show that the MDP's reward function corresponds to our welfare function, and the optimal mapping between states and actions corresponds to the T-EWM policy in our framework. }

In the following equations,
the left-hand sides are the notations for the MDP in LK, and the right-hand sides are
notations for T-EWM from this paper.
We consider the model of Section \ref{one-p}. Firstly, we link the transition probability with the propensity score: For $i,j,a\in\{0,1\}$ at time $t$,
\bea
p_{ij}^{t}(a) & =\text{Pr}(W_{t}=j|W_{t-1}=i;\text{choosing }W_{t}=a).
\eea
The left-hand side is the Markov transition probability
between states $i$ and $j$ under policy $a$. The right-hand side
is a propensity score under policy $a$:
the probability $W_{t}=j$, conditional on $W_{t-1}=i$, given
$W_{t}=a$. Note that, in this simple model, the state at time $t$
is the previous treatment $W_{t-1}$, and the current policy and the
next-period state are both $W_{t}$. We assume that after time $T-1$, the planner implements a deterministic policy, so the probability only takes values in $\{0,1\}$, i.e.,
\bea
p_{ij}^{t}(a) & =1\text{ if }j=a, \quad
p_{ij}^{t}(a) & =0\text{ if }j\ne a.\label{eq:M_pscore}
\eea
Secondly, we connect the reward function with the expected conditional
counterfactual outcome. We denote the reward associated with action $a$ for state $i$ at time $t$ as
\begin{align}
r_{i}^{t}(a) & =\E\left[Y_{t}(a)|W_{t-1}=i\right].\label{eq:reward}
\end{align}
 The left-hand side is the reward in state $i$ under action $a$. The right-hand side is the conditional
expected counterfactual outcome of $Y_{t}(a)$, (recall $a\in\{0,1\}$)
conditional on $W_{t-1}=i$.

Thirdly, we link the expected reward function and the expected unconditional  counterfactual outcome,
\begin{align*}
 \sum_{i}\beta_{i}r_{i}^{t}(a)=r^{t}(a)  & =\E\left[Y_{t}(a)\right].
\end{align*}
The left-hand side is the expected reward for action
$a$, with $\beta_{i}$ as the initial probability of state $i$.
The right-hand side is the unconditional expected counterfactual outcome.

Finally, we show the link between the total expected reward over a finite horizon (of length 2) and the finite-period welfare function
{\small \begin{align}
  v_{i}^{T:T+1}(R)=\E{_{i,R}}\left[\sum_{k=T}^{T+1}r_{i}^{k}(W_{k})\right]\ \nonumber 
 & =\E\bigg[ Y_{T}(1)p_{i1}^{T}\left(g_{1}(W_{T-1})\right)+Y_{T}(0)p_{i0}^{T}\left(g_{1}(W_{T-1})\right)\nonumber \\
 & +Y_{T+1}(1)p_{i1}^{T+1}\left(g_{2}(W_{T})\right)+Y_{T+1}(0)p_{i0}^{T+1}\left(g_{2}(W_{T})\right)|W_{T-1}=i\bigg]\nonumber \\
 & =\E\left[ Y_{T}(1)g_{1}(W_{T-1})+Y_{T}(0)\left[1-g_{1}(W_{T-1})\right]|W_{T-1}=i\right] \nonumber \\
 & +\E\left[ Y_{T+1}(1)g_{2}\left(W_{T}\right)+Y_{T+1}(0)\left[1-g_{2}\left(W_{T}\right)\right]|W_{T}=g_{1}(i)\right]. \label{eq:two-period_welfare_re}
\end{align}}
The left-hand side is the total expected reward over the
planning horizon from $T$ to $T+1$ under the policy $R=(g_{1},g_{2})$,
with the initial state $i$. The last equality follows from (\ref{eq:M_pscore})
and the exclusion condition.

Comparing \eqref{eq:two-period_welfare_re} with \eqref{eq:multi_welfare}, we can view the population conditional welfare of T-EWM with a finite-period target as the value function of a finite-horizon MDP with a nonstationary solution. According to LK, in this case, the policy is usually obtained by using a backward induction algorithm.


\subsection{Simulation} \label{simul}
In this subsection, we illustrate the accuracy of our method through a simple simulation exercise. We consider the following model
\begin{align}
	Y_{t}& =W_{t}\cdot\mu(Y_{t-1},Z_{t-1})+\phi Y_{t-1}+\varepsilon_{t},\nonumber \\
	\mu(Y_{t-1},Z_{t-1}) & =1(Y_{t-1}<B_{1})\cdot1(Z_{t-1}<B_{2})-1(Y_{t-1}>B_{1}\vee Z_{t-1}>B_{2}),\label{eq:simu-dgp}
\end{align}
where $\mu$ is a function determining the direction of the treatment
effect. The treatment effect
at time $t$ is positive if both $Y_{t-1}<B_{1}$ and $Z_{t-1}<B_{2}$, and is negative otherwise. The optimal
treatment rule is therefore $G_{*}=\{(Y_{t-1},Z_{t-1}):Y_{t-1}<B_{1}\text{ and }Z_{t-1}<B_{2}\}$. We set $\varepsilon_{t}\overset{\text{i.i.d.}}{\sim}N(0,1)$,
$Z_{t-1}\overset{\text{i.i.d.}}{\sim}N(0,1)$, $\phi=0.5$, $B_{1}=2.5,$ and
$B_{2}=0.52$ (approximately the 70\% quantile of the standard
normal distribution). The propensity score $e_{t}(Y_{t-1},Z_{t-1})$ is set to $0.5$.
Our goal is to estimate $G_{*}$.  We
consider the quadrant treatment rules defined as \begin{equation}
	\mathcal{G}\equiv\left\{ \begin{array}{c}
		\left((y_{t-1},z_{t-1}):s_{1}(y_{t-1}-b_{1})>0\;\&\;s_{2}(z_{t-1}-b_{2})>0\right),\\
		s_{1},s_{2}\in\{-1,1\},b_{1},b_{2}\in \mathbb{R}
	\end{array}\right\} .\label{eq:quadrant  treatment}
\end{equation}
It is immediate that $G_{\text{FB}}^*\in \mathcal{G}$. Therefore, we can directly estimate the unconditional treatment rule as described in Section \ref{sec:uncon_con}. Figure \ref{fig:n100 and n1000} illustrates the estimated
bounds and true bounds for sample sizes $n=100$ and $n=1000$. In each case, we draw 100 Monte Carlo samples.
\begin{figure}[H]
	\centering
	\begin{minipage}{0.35\textwidth} 
		\includegraphics[width=70mm]{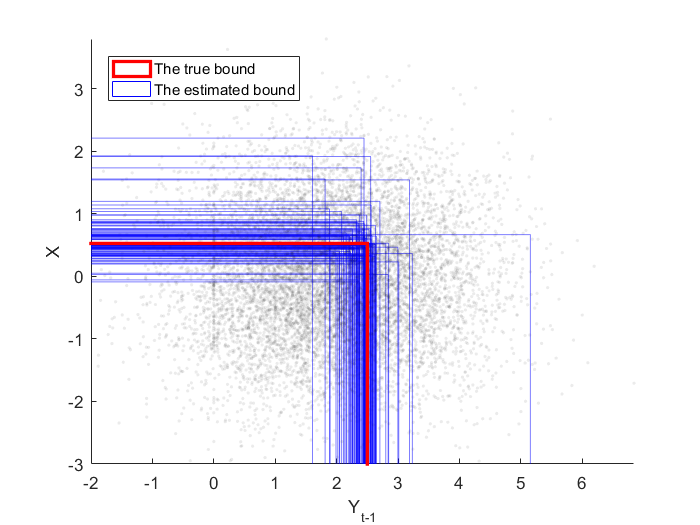}
		\centering
		(a) $n=100$
	\end{minipage}
	\hfil
	\begin{minipage}{0.35\textwidth}
		\includegraphics[width=70mm]{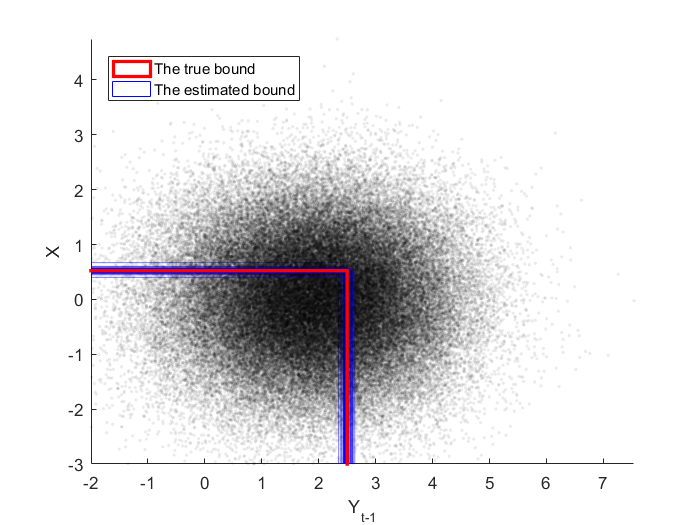}
		\centering
		(b) $n=1000$
		
	\end{minipage}
	\caption{\label{fig:n100 and n1000}The estimated bounds for $n=100$ and $n=1000$.}
\end{figure}

The blue lines are estimated bounds, and the red lines are the
true bounds. For n=100, the majority of the blue lines are close to the red line. As the sample size increases from $100$ to $1000$,
the blue lines become tightly concentrated around the red line. The results in Table \ref{tab:bound_simu} confirm this.
This table presents the Monte Carlo averages $(\hat{\mu}_{B_{1}},\hat{\mu}_{B_{2}})$,
variances $(\hat{\sigma}_{B_{1}}^{2},\hat{\sigma}_{B_{2}}^{2})$, and
MSE of estimated $B_{1}$ and $B_{2}$. We multiply the variances and MSEs by the sample size $n$. The sample sizes are $n=100$, $500$, $1000$,
and 2000. The number of Monte Carlo replications is $500$.

\begin{table}[h]
	\caption{\label{tab:bound_simu}Simulation results for $B_{1}$ and $B_{2}$}
	
	\centering{}{\small{}}%
	
	\begin{tabular}{cr@{\extracolsep{0pt}.}lr@{\extracolsep{0pt}.}lr@{\extracolsep{0pt}.}lr@{\extracolsep{0pt}.}lr@{\extracolsep{0pt}.}lr@{\extracolsep{0pt}.}lr@{\extracolsep{0pt}.}l}
		\toprule
		& \multicolumn{4}{c}{$B_{1}$} & \multicolumn{2}{c}{} & \multicolumn{2}{c}{} &
		\multicolumn{6}{c}{$B_{2}$}\tabularnewline
		\midrule
		$n$ & \multicolumn{2}{c}{$\hat{\mu}_{B_{1}}$} & \multicolumn{2}{c}{$ n\cdot \hat{\sigma}_{B_{1}}^{2}$} &
		\multicolumn{2}{c}{$n\cdot$MSE$_{B_{1}}$} & \multicolumn{2}{c}{} & \multicolumn{2}{c}{$\hat{\mu}_{B_{2}}$} &
		\multicolumn{2}{c}{$n\cdot\hat{\sigma}_{B_{2}}^{2}$} & \multicolumn{2}{c}{$n\cdot$MSE$_{B_{2}}$}\tabularnewline
		\cmidrule{1-7} \cmidrule{2-7} \cmidrule{4-7} \cmidrule{6-7} \cmidrule{10-15} \cmidrule{12-15} \cmidrule{14-15}
		100 & 2&4688 & 14&4331 & 14&5016 & \multicolumn{2}{c}{} & 0&6589 & 18&9382 & 20&7080\tabularnewline
		500 & 2&4919 & 2&3881 & 2&4162 & \multicolumn{2}{c}{} & 0&5433 & 3&3775 & 3&5490\tabularnewline
		1000 & 2&4924 & 1&4676 & 1&5224 & \multicolumn{2}{c}{} & 0&5310 & 1&2620 & 1&3027\tabularnewline
		2000 & 2&4958 & 0&5981 & 0&6327 & \multicolumn{2}{c}{} & 0&5267 & 0&7672 & 0&7767\tabularnewline
		\bottomrule
	\end{tabular}{\small\par}
\end{table}

As the sample size increases, both the $\hat{\mu}_{B_{1}}$ and $\hat{\mu}_{B_{2}}$
converge to their true values, 2.5 and 0.52, respectively. The
variances and MSEs shrink, even after multiplying by the sample size $n$, which
suggests that the convergence rate in this case is faster than $\frac{1}{\sqrt{n}}$. 

\subsection{Additional figures and results for the empirical application} \label{app_app}

\subsubsection{Time series plots of the raw data and estimated propensity scores}

\begin{figure}[H]
	\caption{\label{fig:Cases-and-deaths}Weekly cases and deaths from 4/2020 to
		1/2022}
	
	\centering
	
	\includegraphics[scale=0.5]{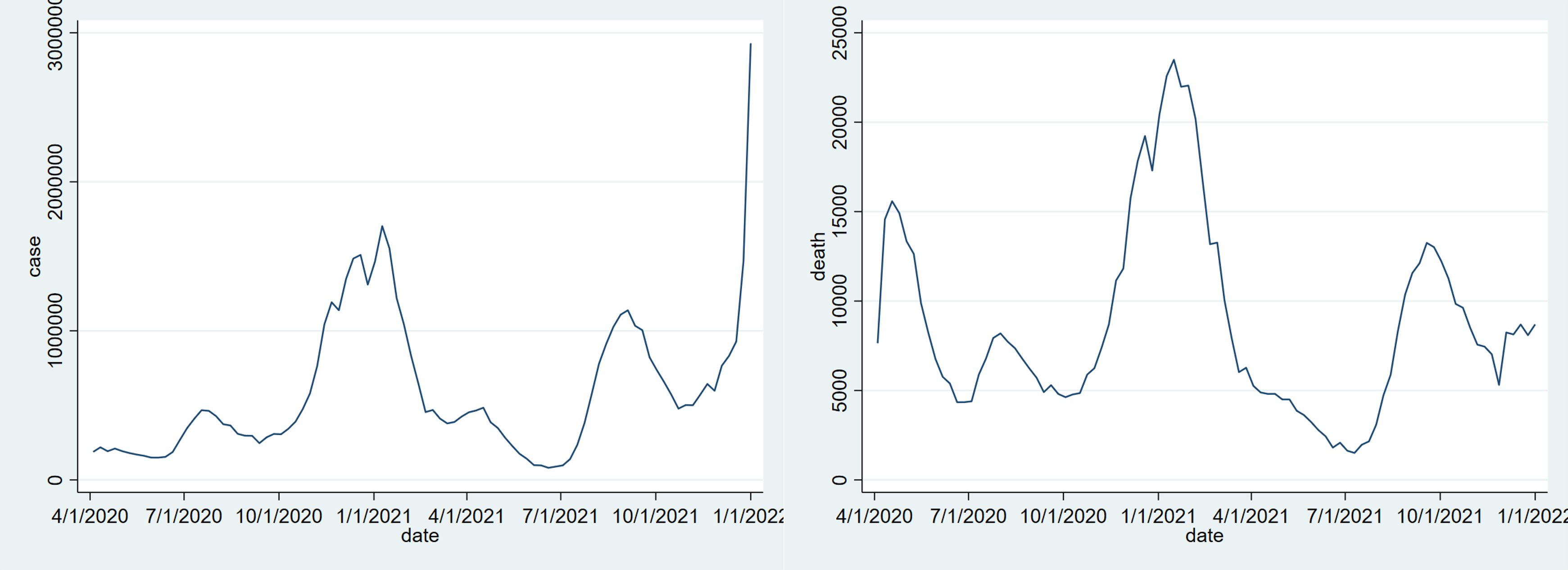}
\end{figure}

\begin{figure}[H]
	\caption{\label{fig:Restriction-econ}Restriction level and economic condition
		from 4/2020 to 1/2022}
	
	\centering
	
	\includegraphics[scale=0.5]{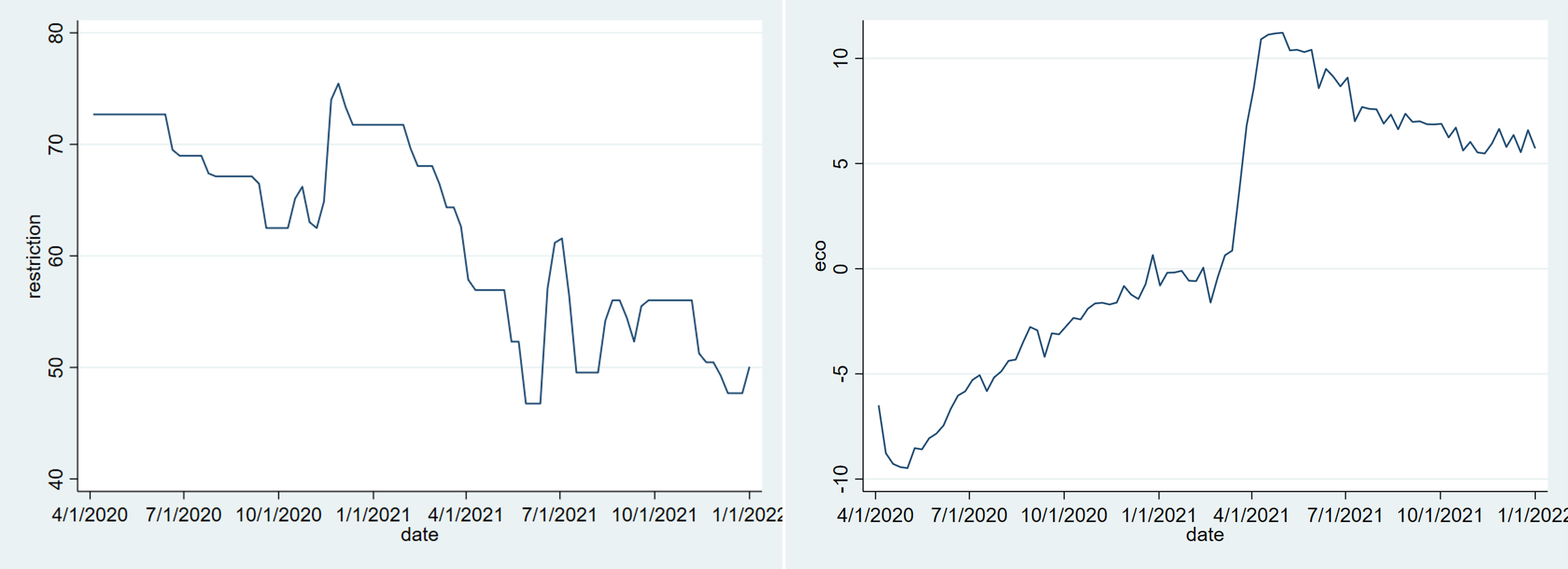}
\end{figure}

%
%

Figure \ref{fig:PS_app} presents the propensity score estimated by \eqref{eq:p_score_app}, in which the set of covariates $X_{t-1}$ is given by \eqref{eq:app_X} (the case of Markov order $q=1$).
%
%

\begin{figure}[h]
	\centering
	\begin{subfigure}[b]{0.4\textwidth} 
		\includegraphics[width=\textwidth]{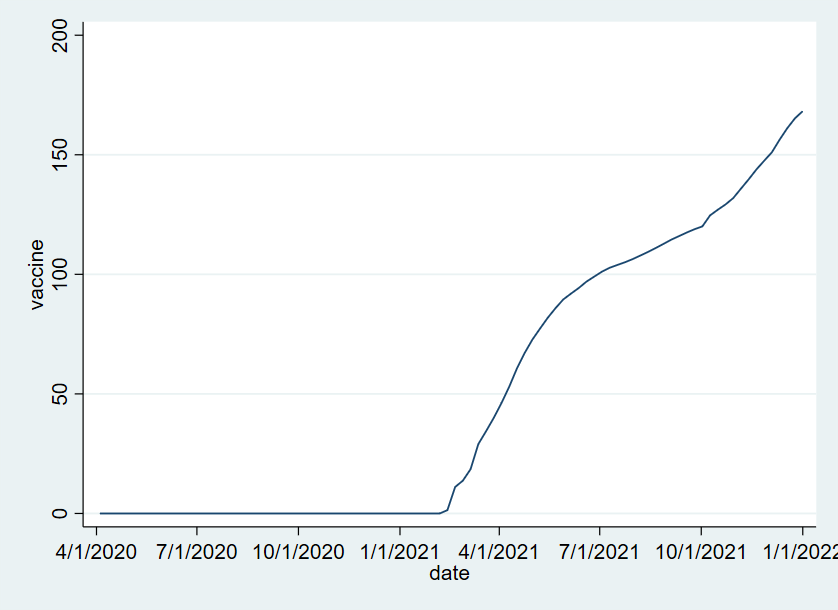} 
		\caption{Vaccine coverage from 4/20 to 1/22} 
		\label{fig:Vaccine} 
	\end{subfigure}
	\begin{subfigure}[b]{0.4\textwidth} 
		\includegraphics[width=\textwidth]{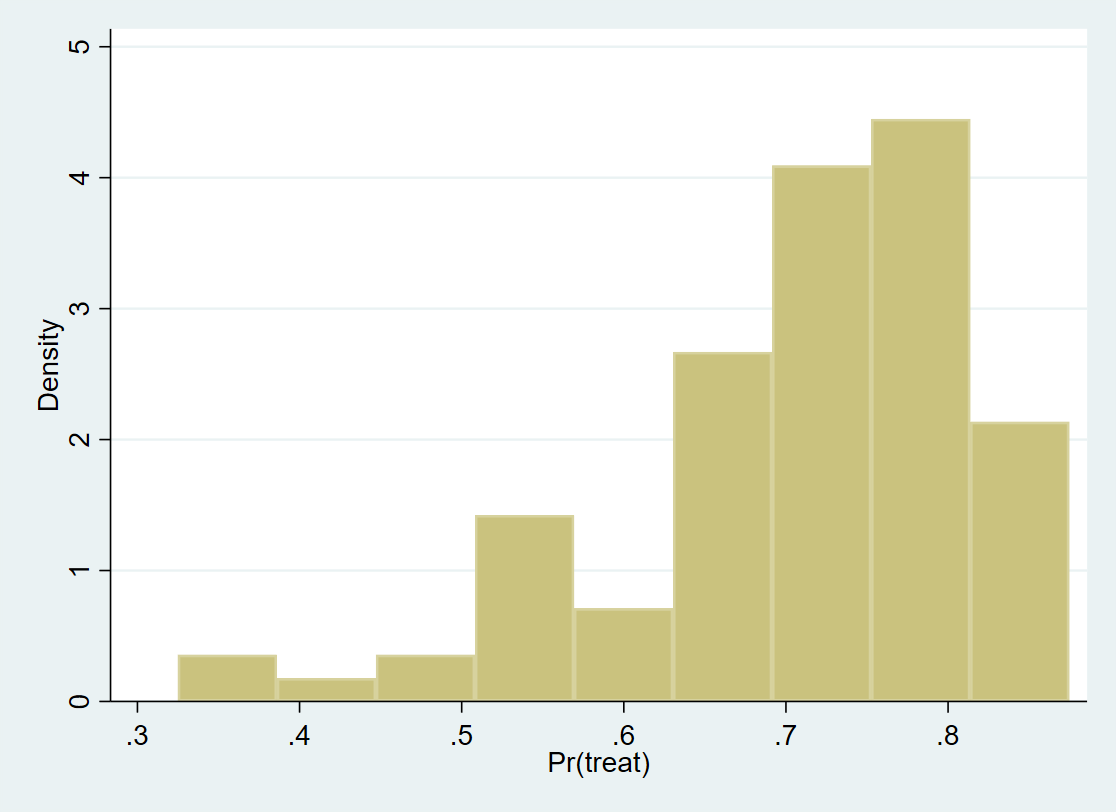} 
		\caption{Estimated PS for $q=1$} 
		\label{fig:PS_app} 
	\end{subfigure}
	\caption{Vaccine coverage and estimated propensity scores} 
\end{figure}

\subsubsection{Policy choices based on an increased number of variables}

\begin{figure}[H]
	\caption{\label{fig:Policy_k=00003D3} Estimated optimal policies based on\\ 
		$X_{T-1}^{P}=\left(\text{change in deaths}_{T-1},\text{restriction stringency}_{T-1},\text{vaccine coverage}_{T-1}\right)$}
	
	\centering
	
	\includegraphics[scale=0.65]{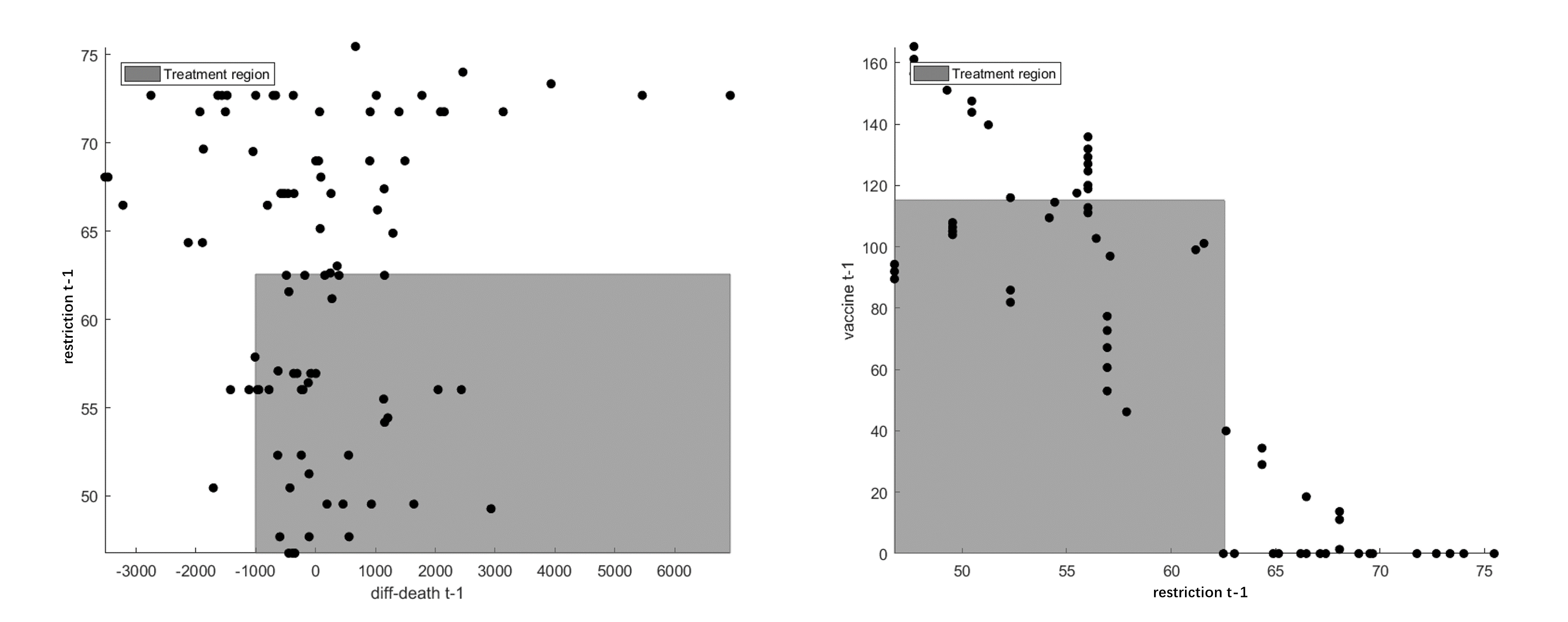}
	
	\noindent\begin{minipage}[t]{1\columnwidth}%
		{\footnotesize{}\raggedleft In the left panel, the $x$-axis is the
			change in deaths at week $T-1$, and the $y$-axis is the stringency of restrictions
			at week $T-1$; in the right panel, the $x$-axis is the stringency of restrictions
			at week $T-1$, and the $y$-axis is the vaccine coverage at week
			$T-1$.}%
	\end{minipage}
\end{figure}

\begin{figure}[H]
	\caption{\label{fig:Policy_k=00003D4} Optimal policy based on \\$X_{T-1}^{P}=(\text{change in deaths}_{T-1},\text{restriction stringency}_{T-1},\text{vaccine coverage}_{T-1},$
		$\text{change in cases}_{T-1})$}
	
	\centering
	
	\includegraphics[scale=0.65]{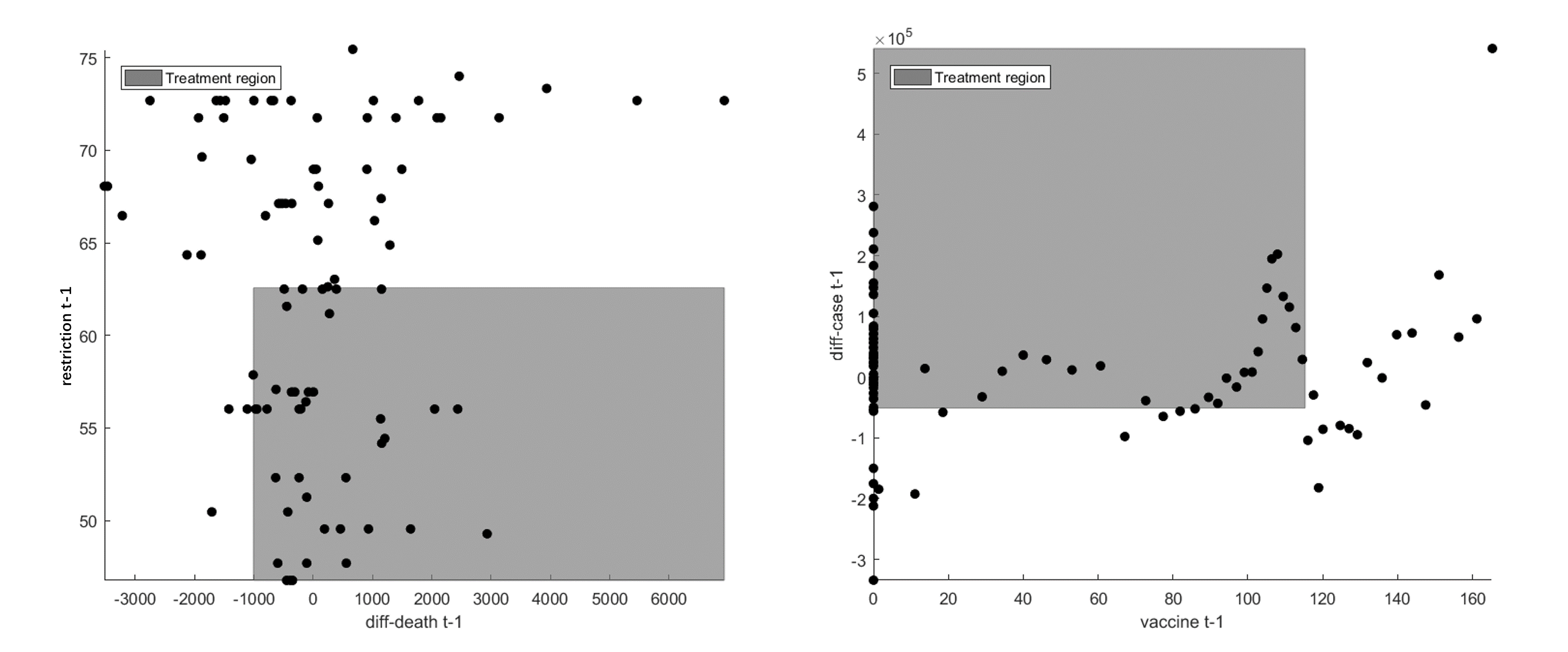}
	
	\noindent\begin{minipage}[t]{1\columnwidth}%
		{\footnotesize{}\raggedleft In the left panel, the $x$-axis is the
			change in deaths at week $T-1$, and the $y$-axis is the stringency of restrictions
			at week $T-1$; in the right panel, the $x$-axis is the vaccine
			coverage at week $T-1$, and the $y$-axis is the change in cases at week
			$T-1$.}%
	\end{minipage}
\end{figure}

\subsubsection{Algorithm details of the T-EWM decision tree} \label{app:emp_tree_algorithm}

The dataset used in the empirical application contains 92 observations, a relatively small size. If a minimum node size is not set, the decision tree may produce nodes with only one or two observations. In these instances, the policy recommendations derived from the tree, based on empirical welfare, become highly variable and can  sometimes be difficult to interpret.

Thus, we set the minimum size for each node to four. During the tree-growing process, if a split determined by the optimal policy variable results in any node having fewer than four observations, the algorithm will disregard this variable. It will then identify a sub-optimal variable from the remaining set of policy variables and execute the split. This process continues until a split is achieved where each resulting node contains at least four observations. If no policy variable can produce a split resulting in nodes with the minimum required size, the algorithm will cease splitting at this node and move to the next. Both trees presented in Figure \ref{T-EWM tree} of the main text and Figure \ref{T-EWM tree2} below are generated using this algorithm.

\subsubsection{Additional results from the T-EWM decision tree} \label{app:emp_tree_result}

In Remark \ref{rem:higher/inf Markov} at the end of Section \ref{discrete_bound}, we extend our theoretical framework to accommodate higher-order Markovian structures. In this subsection, we revisit the empirical application, setting the Markov order $q=2$ under the alternative Assumption \ref{ass:toy_example}{*}, which is introduced in Appendix \ref{app_finite_markov}.
           
For $q=2$, we need to re-estimate the propensity score function. Recall that for the case of $q=1$, the set of covariates of the propensity score is given by \eqref{eq:app_X}: 
\begin{align*}
	\ensuremath{}\ensuremath{X_{t}=} & (\text{cases}_{t},\text{deaths}_{t},\text{change in cases}_{t},\text{change in deaths}_{t},\nonumber \\
	& \text{restriction stringency}_{t},\text{vaccine coverage}_{t},\text{economic conditions}_{t}).
\end{align*}
Now, we shall use both $X_{t-1}$ and $X_{t-2}$ as the covariates for the propensity score. The histograms of the estimated propensity scores for the observed data are presented in Figure \ref{fig:q1_q2}: panel (a) for $q=1$ (replicated from Figure \ref{fig:PS_app}), and panel (b) for $q=2$.
\begin{figure}[h]
	\begin{minipage}{0.35\textwidth} 
		\includegraphics[width=75mm]{pictures/covid-19_app/PS}
		\centering
	{\footnotesize	(a) $q=1$ and $e_t=e_t(X_{t-1})$}
	\end{minipage}
	\hfil
	\begin{minipage}{0.35\textwidth}
		\includegraphics[width=75mm]{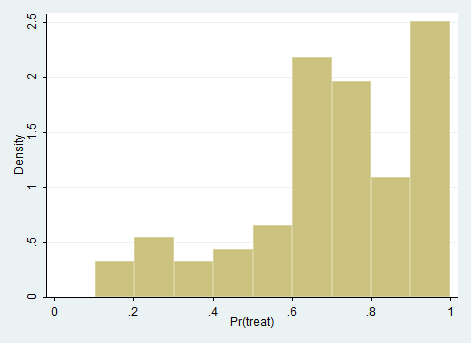}
		\centering
	{\footnotesize	(b)  $q=2$ and $e_t=e_t(X_{t-1},X_{t-2})$}
		
	\end{minipage}
	\caption{\label{fig:q1_q2}The estimated PS for $q=1$ and $q=2$.}
\end{figure}
For Figure \ref{fig:q1_q2}b, we have observed that a non-trivial number of estimated propensity scores are close to 1. This is not necessarily evidence of a violation of the overlap assumption. Overfitting may be the main cause of this pattern, as an increased number of regressors can improve the fit of the predicted values to the binary dependent variables. (Currently, the sample size is 92, and the number of regressors is 14.)

The estimated T-EWM decision tree, as presented in panel (b) of Figure \ref{T-EWM tree2}, is obtained by censoring the estimated propensity scores at 0.025 and 0.975. The set of policy variables is represented by a 14-dimensional vector, i.e.,  $X_{T-1}^{P}=(X_{T-1},X_{T-2})$, where $X_{t}$ is given by \eqref{eq:app_X}.
\begin{figure}[h]
	\begin{minipage}{0.35\textwidth} 
		\includegraphics[width=75mm]{pictures/covid-19_app/tree_q1_p7}
		\centering
		{\footnotesize(a) $q=1$ and $X_{T-1}^{P}=X_{T-1}$}
	\end{minipage}
	\hfil
	\begin{minipage}{0.35\textwidth}
		\includegraphics[width=80mm]{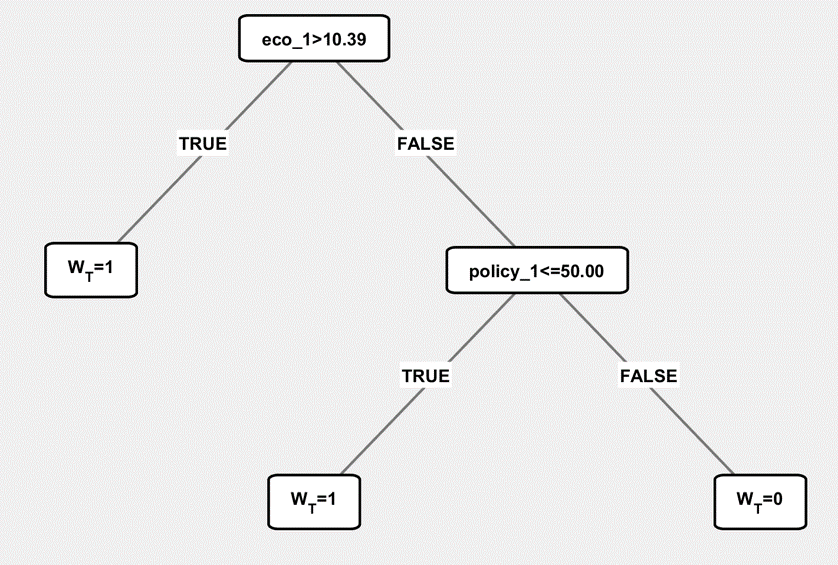}
		\centering
	{\footnotesize(b)   $q=2$ and $X_{T-1}^{P}=(X_{T-1},X_{T-2})$}
		
	\end{minipage}
	\caption{\label{T-EWM tree2}T-EWM decision trees for $q=1$ and $q=2$.}
\end{figure}
Figure \ref{T-EWM tree} in the main text is reproduced here in panel (a). We observe that with a higher Markov order and an enlarged set of policy variables, the treatment region recommended by the T-EWM tree has changed. Both trees have chosen the economic condition as the variable for the first split, but the thresholds selected are slightly different. In the left branch of the second level of panel (b), the algorithm stops splitting the node since it cannot find a variable that can produce two leaves containing at least four observations each. Nevertheless, both trees provide policy recommendations that are reasonable to interpret.
\newpage

\section{Online Appendix: Other Results and Extensions} \label{App_2}

\subsection{On Assumption \ref{ass:correct_specify}} \label{sec:nonpara_conti}


In Section \ref{sec:uncon_con}, we have shown that we can bound conditional regret by  unconditional regret if the unconditional first-best policy is feasible. 
In this Section, we present examples where this method is applicable, and examples where it is not. More examples can be found in Appendix \ref{sec:con_uncon}. When Assumption \ref{ass:correct_specify} does not hold, we proceed to the nonparametric method discussed in Section \ref{sec:conti_kernel}. \\

\subsubsection{Examples} \label{sec:motivation_kernel}

In this subsection, we discuss the relationship between the optimal policy solutions in terms of conditional and unconditional welfare.
This relationship is not straightforward. In some cases, unconditional welfare does bound conditional welfare, and the methods and results in Section \ref{sec:uncon_con} directly apply.
However, in other cases, when the first-best policy is not feasible, we shall use the kernel estimator. This motivates us to present the kernel estimator as an important alternative to the method of Section \ref{sec:uncon_con}.

\begin{exmp}\label{exmp:5.1}
Observing  $X_{T-1}=\left(Y_{T-1}, W_{T-1},Z_{T-1}\right)^{\prime}\in\mathbb{R}\times\{0,1\}\times\mathbb{R}^{2}$ at time $T-1$, the planner chooses $W_T$ based on the last two continuous variables. The feasible policy class is rectangles in $\mathbb{R}^{2}$:
$$\mathcal{G}=\left\{ z\in[a_{1},a_{2}]\times[b_{1},b_{2}]:a_{1},a_{2},b_{1},b_{2}\in\mathbb{R}\right\}.$$
The corresponding unconditional problem is $\underset{G\in\mathcal{G}}{\mbox{max }}\mathcal{W}_{T}(G).$ Suppose the planner is interested in maximizing the welfare conditional on $Z_{T-1}=z:=(z^{(1)},z^{(2)})$. The conditional problem is to find
\beq \label{eq:condition_first}
 \arg \max_{G\in\mathcal{G}}\mathcal{W}_{T}(G|Z_{T-1}=z).
\eeq

We illustrate how the policy solutions can differ between conditional and unconditional welfare functions.
In Figures \ref{fig:fb_equi} and \ref{fig:nofb_noequi}, the \textit{shaded area} represents the region where the conditional average treatment effect is positive. The first best unconditional policy assigns $W_T = 1$ to any value of $Z_{T-1}$ inside this region, and $W_T = 0$ to any point outside it. This policy is also the solution to  \eqref{eq:condition_first}. The \textit{red rectangle} is the best feasible (i.e., rectangular) unconditional policy. This policy assigns $W_T = 1$ to any realization of $Z_{T-1}$ inside the rectangle. The conditional policy concerns what policy to assign only at a particular value of $Z_{T-1}$ corresponding to its realized value in the data (blue point in the right-hand side panel of Figure \ref{fig:fb_equi}). If the best feasible unconditional policy (red rectangle) agrees with the first best unconditional policy (shared area), then the policy choice informed by the unconditional policy is guaranteed to be optimal in terms of the conditional policy at any conditioning value of $Z_{T-1}$.
\begin{figure}[H]\label{fig:1}
        \caption{\label{fig:fb_equi} $G_{\text{FB}}^{*}\in\mathcal{G}$}
        \centering
        \includegraphics[clip,scale=0.5]{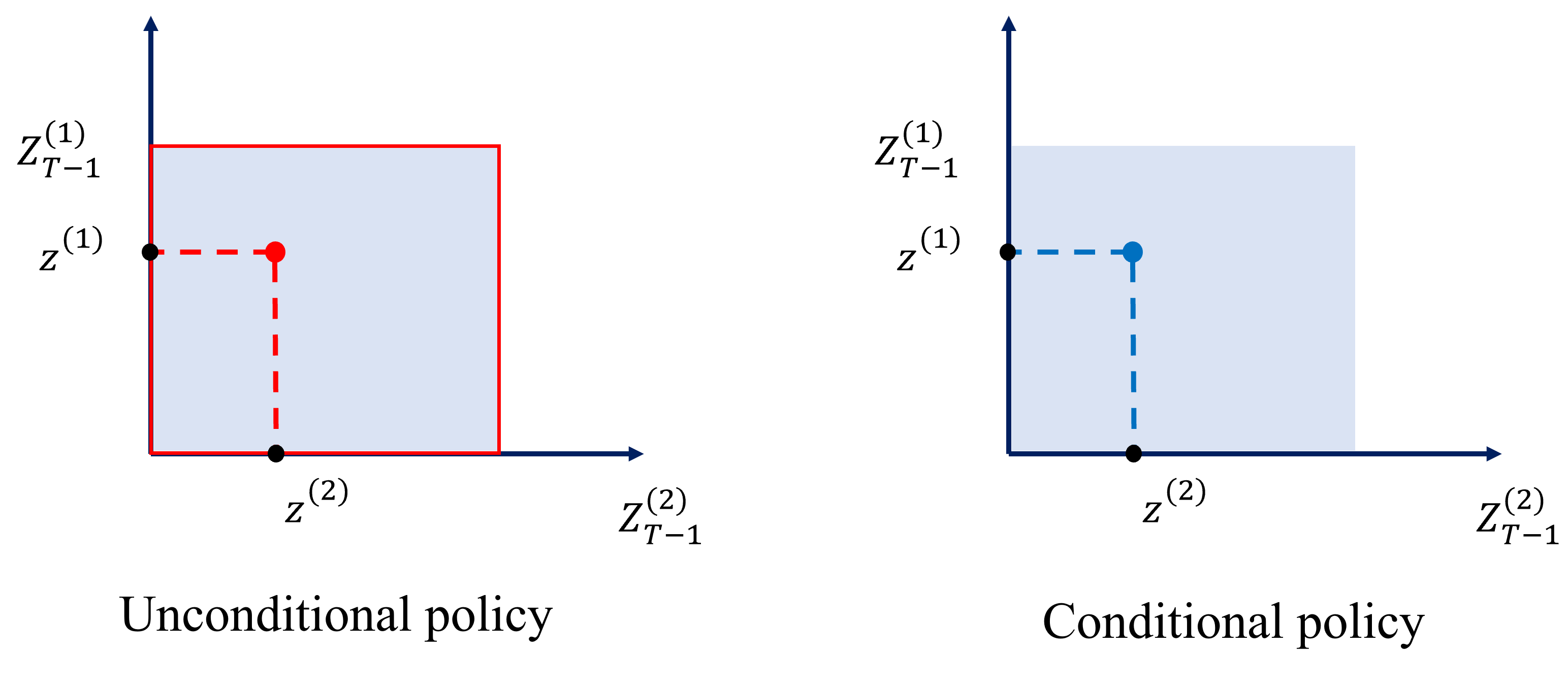}\\
\end{figure}
\begin{figure}[H]
        \caption{\label{fig:nofb_noequi} $G_{\text{FB}}^{*}\protect\notin\mathcal{G}$
                }
        \centering
        \includegraphics[clip,scale=0.5]{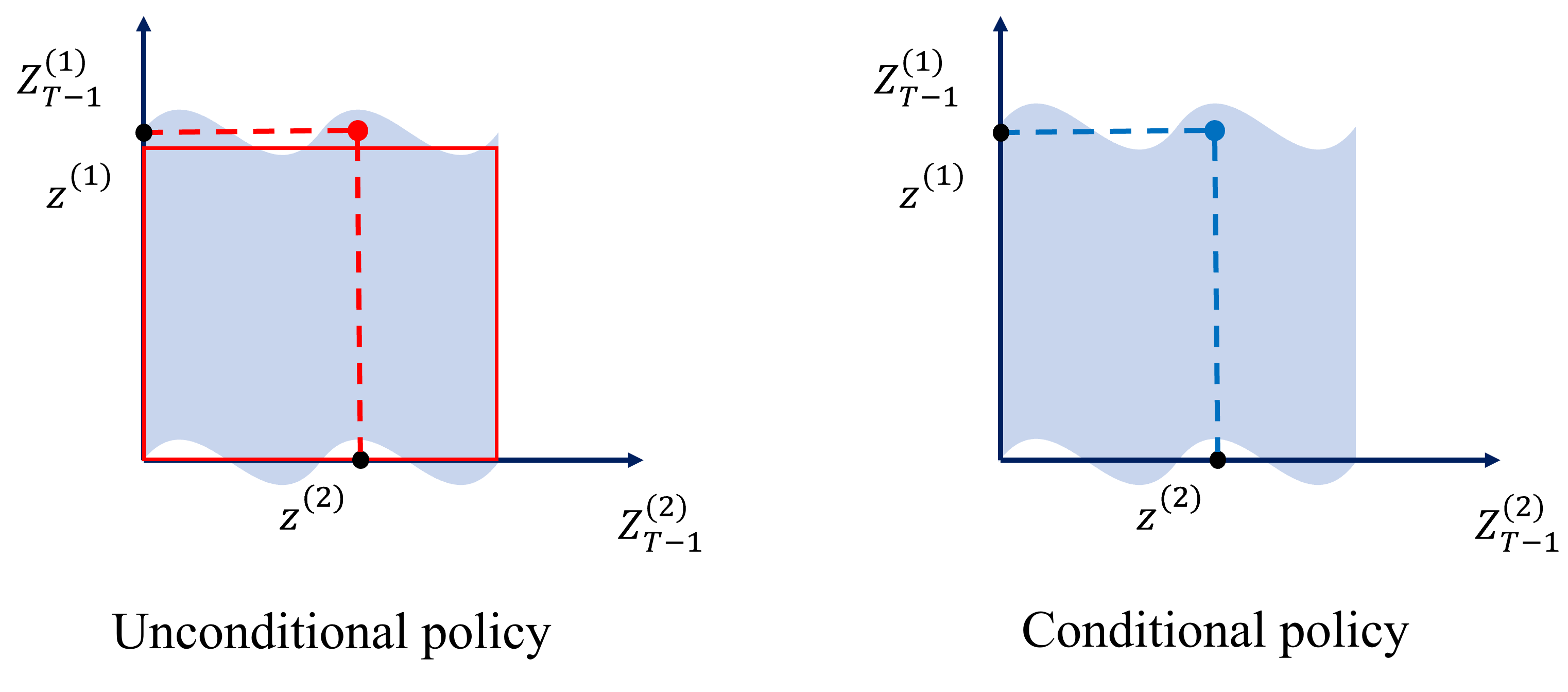}\\
\end{figure}
 \end{exmp}

Figure \ref{fig:nofb_noequi} shows a case  where the first best unconditional policy is not contained in the class of feasible policies: it is not possible to implement the policy choice that coincides with the shaded area. In this case, the policy chosen by the unconditionally optimal feasible policy (red rectangle in the left-hand side panel) does not coincide with the optimal policy choice of the conditional policy. The highlighted point $(z^{(1)}, z^{(2)})$ lies outside the red rectangle, so the best feasible unconditional policy would set $W_t = 0$. However, it lies within the shaded region, so the conditional average treatment effect given $Z_{T-1}=(z^{(1)}, z^{(2)})$ is positive, and the optimal policy conditional on $Z_{T-1}=(z^{(1)}, z^{(2)})$ is to set $W_t = 1$.

These two examples show the importance of the first best unconditional policy: when it is included in the set of feasible unconditional policies, the solution to the conditional problem corresponds to the solution to the unconditional problem. When it is not included, we do not have this correspondence.
 In Appendix \ref{simple}, we show that the feasibility of the first best solution to the unconditional problem is a sufficient condition. A sufficient and necessary condition is given by Assumption \ref{ass:correct_specify}.

\subsection{The relationship between the conditional and unconditional cases} \label{sec:con_uncon}

In this section, we present further examples to illustrate that $G_{\text{FB}}^{*}\in\mathcal{G}$ is a sufficient but not necessary condition for equivalence between the unconditional and conditional problems. Finally, we extend Example \ref{exmp:5.1} to show how Assumption \ref{ass:correct_specify} ensures this equivalence.

\subsubsection{$G_{\text{FB}}^{*}\in\mathcal{G}$ is sufficient: a discrete and a continuous case} \label{simple}

Let $X_{t-1}=W_{t-1}\in\{0,1\}$, and $\mathcal{G}$ be a subclass
of the power set of $\{0,1\}$, $\mathcal{P}=\left\{ \emptyset,\{0\},\{1\},\{0,1\}\right\} $.

For compactness, suppress $W_{T-1}$ in $Y_{T}(W_{T-1},1)$. Unconditional
welfare can be written as:
\begin{align}
\mathcal{W}_{T}(G)= & \mathsf{E}\left[ Y_{T}(1)\mathbf{1}(W_{T-1}\in G)+Y_{T}(0)\mathbf{1}(W_{T-1}\notin G)\right]\nonumber \\
= & \mathsf{E}\left[ \left[Y_{T}(1)-Y_{T}(0)\right]\mathbf{1}(W_{T-1}\in G)+Y_{T}(0)\right] \nonumber \\
= & \mathsf{E}\left[\tau(W_{T-1})\mathbf{1}(W_{T-1}\in G)\right]+\mathsf{E}\left[Y_{T}(0)\right],\label{eq:uncond_disc}
\end{align}

where $\tau(w_{T-1})=\mathsf{E}\left[Y_{T}(1)-Y_{T}(0)|W_{T-1}=w_{T-1}\right]$,
and the last equality follows from the law of iterated expectations.

The first best unconditional policy is
\begin{equation}
G_{\text{FB}}^{*}\equiv\{w_{T-1}\in\{0,1\}:\tau(w_{T-1})\geq0\}.\label{eq:def_fb}
\end{equation}

By the assumption that $G_{\text{FB}}^{*}\in\mathcal{G}$
\begin{equation}
G_{\text{FB}}^{*}=\mbox{argmax}_{G\in\mathcal{G}}\mathcal{W}_{T}(G).\label{eq:uncon_solution}
\end{equation}

The planner's conditional objective function can be written as
\begin{align}
\mathcal{W}_{T}(G|W_{T-1}) & =\mathsf{E}\left[ Y_{T}(1)\mathbf{1}(W_{T-1}\in G)+Y_{T}(0)\mathbf{1}(W_{T-1}\notin G)|W_{T-1}\right] \nonumber \\
 & =\mathsf{E}\left[ \left[Y_{T}(1)-Y_{T}(0)\right]\mathbf{1}(W_{T-1}\in G)+Y_{T}(0)|W_{T-1}\right] \nonumber \\
 & =\mathsf{E}\left[ \left[Y_{T}(1)-Y_{T}(0)\right]\mathbf{1}(W_{T-1}\in G)|W_{T-1}\right] +\mathsf{E}\left[Y_{T}(0)|W_{T-1}\right]\nonumber \\
 & =\mathsf{E}\left[\tau(W_{T-1})\mathbf{1}(W_{T-1}\in G)|W_{T-1}\right]+\mathsf{E}\left[Y_{T}(0)|W_{T-1}\right].\label{eq:tau_con}
\end{align}

To check whether $G_{\text{FB}}^{*}$ is optimal in the conditional
problem, we need to study this problem w.r.t. $\mathcal{P}$
\begin{align}
 & \underset{G\in\mathcal{\mathcal{P}}}{\mbox{max }}\mathcal{W}_{T}(G|W_{T-1}=w_{T-1})\nonumber \\
= & \underset{G\in\mathcal{\mathcal{P}}}{\mbox{max }}\mathsf{E}\left[\tau(W_{T-1})\mathbf{1}(W_{T-1}\in G)|W_{T-1}=w_{T-1}\right]\nonumber \\
= & \underset{G\in\mathcal{\mathcal{P}}}{\mbox{max }}\tau(w_{T-1})\mathbf{1}(w_{T-1}\in G)\nonumber \\
= & \tau(w_{T-1})\mathbf{1}(w_{T-1}\in G_{\text{FB}}^{*}),\label{eq:con_max}
\end{align}


where the first equality follows from (\ref{eq:tau_con}), and the last one follows from the definition of $G_{\text{FB}}^{*}$ in (\ref{eq:def_fb}). The equivalence follows by combining (\ref{eq:uncon_solution})
and (\ref{eq:con_max}).\\

We now turn to the continuous conditioning variable case.

The planner's unconditional welfare function can be rewritten as (suppressing $W_{T-1}$ in $Y_{T}(W_{T-1},1)$)
\begin{align*}
\mathcal{W}_{T}(G)= & \mathsf{E}\left[ Y_{T}(1)\mathbf{1}(X_{T-1}\in G)+Y_{T}(0)\mathbf{1}(X_{T-1}\notin G)\right] \\
= & \mathsf{E}\left[ \left[Y_{T}(1)-Y_{T}(0)\right]\mathbf{1}(X_{T-1}\in G)+Y_{T}(0)\right] \\
= & \mathsf{E}\left[\tau(X_{T-1})\mathbf{1}(X_{T-1}\in G)\right]+\mathsf{E}\left[Y_{T}(0)\right],
\end{align*}

where $\tau(x_{T-1})=\mathsf{E}\left[Y_{T}(1)-Y_{t}(0)|X_{T-1}=x_{T-1}\right]$,
and the last equality follows from the law of iterated expectations.

The first best policy is
\begin{equation}
G_{\text{FB}}^{*}\equiv\{x_{T-1}\in\mathbb{R}^{2}:\tau(x_{T-1})\geq0\}.\label{eq:un_fb_con}
\end{equation}

(For simplicity, we assume that $G_{\text{FB}}^{*}$ is measurable,
i.e., $G_{\text{FB}}^{*}\in\mathfrak{B}(\mathbb{R}^{2})\subset\mathcal{P}(\mathbb{R}^{2})$.
This means that we don't have to deal with the outer probability and expectation.)
By assumption $G_{\text{FB}}^{*}\in\mathcal{G}$, 
\begin{equation}
G_{\text{FB}}^{*}\in \mbox{argmax}_{G\in\mathcal{G}}\mathcal{W}_{T}(G).\label{eq:uncon_solution_gen}
\end{equation}

We now introduce some notations. Let $X_{t-1}=[X_{t-1}^{(1)},X_{t-1}^{(2)}]^{\prime}\in\mathbb{R}^{2}$.
$X_{t-1}^{(1)}$ and $X_{t-1}^{(2)}$ can be continuous or discrete,
e.g., one of them can be the treatment $W_{t-1}.$
Let $\mathcal{G}$ be a class of subsets of $\mathbb{R}^{2}$, and
$\mathcal{G}^{(1)}$ a class of subsets of $\mathbb{R}.$ The planner's conditional objective function can be written as:
\begin{align}
 & \mathcal{W}_{T}(G^{(1)}|X_{T-1}^{(2)})\nonumber \\
= & \mathsf{E}\left[ Y_{T}(1)\mathbf{1}(X_{T-1}^{(1)}\in G^{(1)})+Y_{T}(0)\mathbf{1}(X_{T-1}^{(1)}\notin G^{(1)})|X_{T-1}^{(2)}\right] \nonumber \\
= & \mathsf{E}\left[ \left[Y_{T}(1)-Y_{T}(0)\right]\mathbf{1}(X_{T-1}^{(1)}\in G^{(1)})+Y_{T}(0)|X_{T-1}^{(2)}\right] \nonumber \\
= & \mathsf{E}\left[ \left[Y_{T}(1)-Y_{T}(0)\right]\mathbf{1}(X_{T-1}^{(1)}\in G^{(1)})|X_{T-1}^{(2)}\right] +\mathsf{E}\left[Y_{T}(0)|X_{T-1}^{(2)}\right]\nonumber \\
= & \mathsf{E}\left[ \mathsf{E}\left[\tau(X_{T-1})\mathbf{1}(X_{T-1}^{(1)}\in G^{(1)})|X_{T-1}\right]|X_{T-1}^{(2)}\right] +\mathsf{E}\left[Y_{T}(0)|X_{T-1}^{(2)}\right],\label{eq:con_condi_decomp}
\end{align}

where the last equality follows from the law of iterated expectations.

We also assume that the conditional first-best policy $G_{\text{CFB}}^{*}$
is measurable, i.e., $G_{\text{CFB}}^{*}(X_{T-1}^{(2)}=x_{T-1}^{(2)})\in\mathfrak{B}(\mathbb{R})$,
for every $x_{T-1}^{(2)}\in\mathbb{R}.$ Note that $\tau(x_{t-1})=\tau(x_{t-1}^{(1)},x_{t-1}^{(2)}).$
Following the last row of (\ref{eq:con_condi_decomp}), the optimal
conditional policy is defined to be
\begin{align*}
 & \text{argmax}_{G^{(1)}\in\mathfrak{B}(\mathbb{R})}\mathcal{W}_{T}(G^{(1)}|X_{T-1}^{(2)}=x_{T-1}^{(2)})\\
= & \text{argmax}_{G^{(1)}\in\mathfrak{B}(\mathbb{R})}\mathsf{E}\left\{ \mathsf{E}\left[\tau(X_{T-1})\mathbf{1}(X_{T-1}^{(1)}\in G^{(1)})|X_{T-1}\right]|X_{T-1}^{(2)}=x_{T-1}^{(2)})\right\} \\
= & \text{argmax}_{G^{(1)}\in\mathfrak{B}(\mathbb{R})}\mathsf{E}\left\{ \mathsf{E}\left[\tau(X_{T-1}^{(1)},x_{T-1}^{(2)})\mathbf{1}(X_{T-1}^{(1)}\in G^{(1)})|X_{T-1}^{(1)}\right]\right\} \\
= & \text{argmax}_{G^{(1)}\in\mathfrak{B}(\mathbb{R})}\mathsf{E}\left[\tau(X_{T-1}^{(1)},x_{T-1}^{(2)})\mathbf{1}(X_{T-1}^{(1)}\in G^{(1)})\right].
\end{align*}

The optimal policy conditional on $X_{T-1}^{(2)}=x_{T-1}^{(2)}$
is
\begin{equation}
G_{\text{CFB}}^{*}(X_{T-1}^{(2)}=x_{T-1}^{(2)})=\{x_{T-1}^{(1)}\in\mathbb{R}:\tau(x_{T-1}^{(1)},x_{T-1}^{(2)})\geq0\}.\label{eq:cond_fb_con}
\end{equation}

Comparing (\ref{eq:un_fb_con}) and (\ref{eq:cond_fb_con}), we see
$G_{\text{CFB}}^{*}(X_{T-1}^{(2)}=x_{T-1}^{(2)})$ is given by the intersection
of $G_{\text{FB}}^{*}$ with the line $X_{T-1}^{(2)}=x_{T-1}^{(2)}$.

\subsubsection{$G_{\text{FB}}^{*}\in\mathcal{G}$ is not a necessary condition \label{subsec: relax}}
Here we show that $G_{\text{FB}}^{*}\in\mathcal{G}$ is sufficient but not necessary for correspondence between the optimal conditional policy and the first-best unconditional policy.

\begin{exmp}
\textbf{Univariate covariates}

We set $X_{t-1}=W_{t-1}\in\{0,1\}$ and $\mathcal{G}=\left\{ \emptyset,\{1\}\right\} \subset\mathcal{P}=\left\{ \emptyset,\{0\},\{1\},\{0,1\}\right\} .$
Suppose
\[
\tau(1)=\tau(0)=1>0.
\]

For the unconditional problem (\ref{eq:uncond_disc}), the
first-best policy is then
\[
G_{\text{FB}}^{*}=\{0,1\}.
\]

Note that $G_{\text{FB}}^{*}\notin\mathcal{G}$, so the solution to the
unconditional problem is
\[
G_{*}\equiv\underset{G\in\mathcal{G}}{\mbox{argmax }}\mathcal{W}_{T}(G)=\{1\}.
\]

Consider the conditional problem for $W_{T-1}=1$.
\begin{align}
 & \underset{G\in\mathcal{\mathcal{P}}}{\mbox{max }}\mathcal{W}_{T}(G|W_{T-1}=1)\nonumber \\
= & \underset{G\in\mathcal{\mathcal{P}}}{\mbox{max }}\mathsf{E}\left[\tau(W_{T-1})\mathbf{1}(W_{T-1}\in G)|W_{T-1}=1\right]\nonumber \\
= & \underset{G\in\mathcal{\mathcal{P}}}{\mbox{max }}\tau(1)\mathbf{1}(1\in G)\nonumber \\
= & \tau(1)\mathbf{1}(1\in G_{*}).\label{eq:con_max-1}
\end{align}

Therefore, the solution to the unconditional problem is also the solution to the
conditional problem. (This is only the case for $W_{T-1}=1.$)\\
\end{exmp}
\begin{exmp}
\textbf{ Two-dimensional discrete covariates}

Set $X_{t-1}=\left(W_{t-1},Z_{t-1}\right)^{\prime}\in\{0,1\}\times\{i\}_{i=0}^{10}$$.$
Suppose
\[
G_{\text{FB}}^{*}=\{0,1\}\times\{1,3,5,7,9\}.
\]

and
\[
\mathcal{G}=\left\{ \begin{array}{c}
\left\{ \left(w,z\right):w\in\{0,1\},z\in\{i\}_{i=0}^{10}\text{, and }z\in[0,a)\right\} ,\\
a\in\mathbb{R}^{+}
\end{array}\right\} .
\]

For example, if $G\in\mathcal{G}$ and $a=4.5$, $(0,4)\in G$,
so $(0,0)$, (0,1), (0,2), and (0,3) are also in $G$.

Note that $G_{\text{FB}}^{*}\notin\mathcal{G}$. Suppose the best feasible unconditional policy is
\begin{align} \label{two-non-same}
G_{*} & \equiv\underset{G\in\mathcal{G}}{\mbox{argmax }}\mathcal{W}_{T}(G)=\{0,1\}\times\{0,1,2,3,4,5\},
\end{align}
We can always construct a data-generating process with a certain type of conditional treatment effect $\tau$, such that \eqref{two-non-same} is the best feasible unconditional policy. For example, let $\tau(x_{t-1})=\tau(w_{t-1},z_{t-1})$, set $\text{Pr}(Z_{t-1}=i)=1/10$, $Z_{t-1}\perp W_{t-1}$, and for any
$w\in\{0,1\}$ assume

\[
\tau(w,z)=\begin{cases}
2 & z\in\{1,3,5\}\\
0.1 & z\in\{7,9\}\\
-1.5 & z\text{ is even.}
\end{cases}
\]

For example, a policy $G_{a}$ with $a=5.5$ includes $\{w,0\}$, $\{w,2\}$ and $\{w,4\}$, which
has a welfare cost of $3\times-1.5$.

The conditional problem is
\begin{align*}
\underset{G^{z}\in\mathcal{G}^{z}}{\mbox{max }}\mathcal{W}_{T}(G^{z}|W_{T-1}=w)
\end{align*}

where $\mathcal{G}^{z}=\left\{ z:z\in\{i\}_{i=0}^{10}\text{, and }z\in[0,a),a\in\mathbb{R}^{+}\right\} $.
Here
\begin{align*}
  \underset{G^{z}\in\mathcal{G}^{z}}{\mbox{argmax }}\mathcal{W}_{T}(G^{z}|W_{T-1}=w)
= \{0,1,2,3,4,5\},
\end{align*}

which is the intersection of $G_{*}$ with $\text{Supp}(Z_{T-1})$.
We have the conclusion.

\end{exmp}
\subsubsection{An illustration of Assumption \ref{ass:correct_specify}}

Recall Assumption \ref{ass:correct_specify}
\[ \arg\sup_{G\in \mathcal{G}} \mathcal{W}_{T}(G)\subset \arg\sup_{G\in \mathcal{G}} \mathcal{W}_{T}(G|x). \]
We extend Example \ref{exmp:5.1} to show why Assumption \ref{ass:correct_specify} ensures equivalence between the unconditional and conditional problems.

\textbf{Example \ref{exmp:5.1}, continued.}
\begin{figure}[H]
	\caption{ \label{fig:3} $G_{\text{FB}}^{*}\protect\notin\mathcal{G}$
		but Assumption \ref{ass:correct_specify} is satisfied}
		\centering
	\includegraphics[clip,scale=0.5]{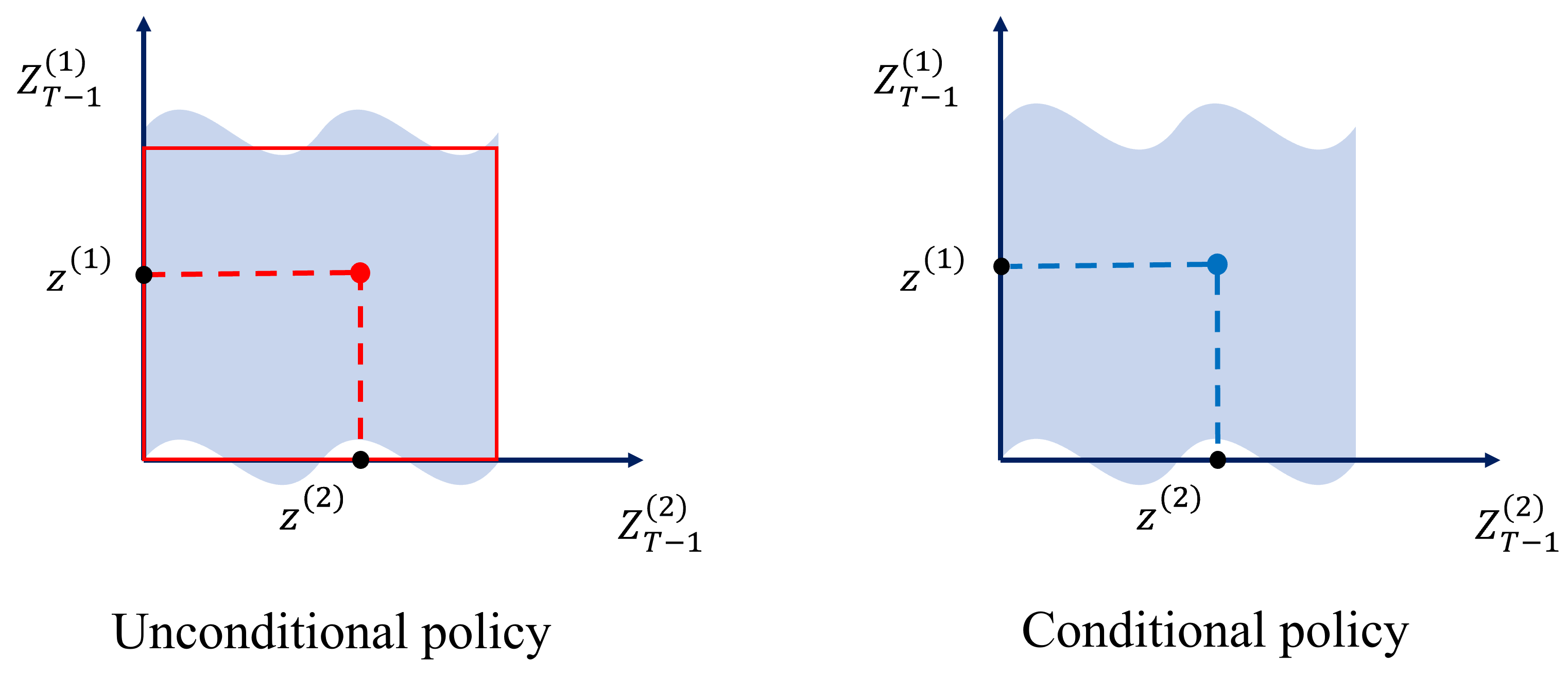}
\end{figure}

Figure \ref{fig:3} shows a case where the solutions coincide even though the first best unconditional policy is not available. The red square differs from the shaded area (the first best), but the red point is inside the red square, and the blue point is inside the shaded area. The social planner will set $W_T = 1$ in both cases. Hence, the feasibility of the first best solution is sufficient but not necessary for conditional and unconditional welfare to coincide.

Thus, there exist situations where the first-best solution is not feasible, but we can still achieve correspondence. This example confirms the validity of Assumption \ref{ass:correct_specify}.

\subsection{Multi-period welfare with continuous covariates}
\label{appendmulti}
In this section, we extend the results in the previous sections to a multi-period setup.  We focus on a simple offline decision problem with deterministic treatment rules.
Namely, we do not update the policy after time $T$, but we allow the welfare function to include the updated realized  observations. A social planner is faced with a finite multi-period
welfare target and a set of continuous policy variables is available. Without loss
of generality, we focus on the case of a two-period policy assignment.

The planner's decision-making procedure can be described as follows. Based
on a sample collected from time 0 to $T-1$, the planner chooses (or estimates)
a \textit{decision rule} that will be implemented on $T$ and $T+1$.
The rule is characterized by two sets defined on $\mathcal{X}$, and
we write them as $G_{1:2}=\{G_{1},G_{2}\}\in\mathcal{X}^{2}$ . At
the beginning of time $T$, the planner makes the decision by $W_{T}=1(X_{T-1}\in G_{1})$,
then she/he will observe the outcome $Y_{T}$ as well as $X_{T}=(Y_{T},W_{T},Z_{T})$
at the end of time $T$. At the beginning of time $T+1$, the planner will
make the decision by $W_{T+1}=1(X_{T}\in G_{2})$. Since we focus
on an offline problem, $G_{2}$ is chosen (estimated) at the end of
$T-1$ and implemented at $T+1$ after  $X_{T}$ is revealed.
Also, the class of subsets of $\mathcal{X}^{2}$ where $G_{1:2}$
is chosen from, i.e., $\mathcal{G}_{1:2}$, is assumed to be of polynomial classification and  with finite
VC dimension.


Recall the notations:
\begin{align*}
S_{t}(G) & =Y_{t}(W_{t-1},1)1(X_{t-1}\in G)+Y_{t}(W_{t-1},0)1(X_{t-1}\notin G),\\
\mathcal{W}_{t}(G|X_{t-1}) &= \E\left[S_{t}(G)|X_{t-1}\right],\\
\widehat{\mathcal{W}}_{t}(G)&=  \frac{Y_{t}W_{t}}{e_{t}(X_{t-1})}1(X_{t-1}\in G)+\frac{Y_{t}\left(1-W_{t}\right)}{1-e_{t}(X_{t-1})}1(X_{t-1}\notin G).
\end{align*}
Let us define\begin{align*}
\mathcal{\bar{W}}_{t}(G) & =\mathsf{E}\left[S_{t}(G)|\mathcal{F}_{t-1}\right].
\end{align*}

We maintain Assumption \ref{ass:continuous_Markov}, the Markovian condition of order 1. Then, the
conditional two-period welfare function is defined as \begin{align} \label{wel_mul_conti}
&\mathcal{W}_{T:T+1}(G_{1:2}|\mathcal{F}_{T-1})  \nonumber\\
  =&\mathcal{W}_{T}(G_{1}|X_{T-1})+\mathsf{E}\left[S_{T+1}(G_{2})|\mathcal{F}_{T-1}\right]\nonumber\\
  =&\mathcal{W}_{T}(G_{1}|X_{T-1})+\mathsf{E}\left[\mathsf{E}\left[S_{T+1}(G_{2})|\mathcal{F}_{T}\right]|\mathcal{F}_{T-1}\right]\nonumber\\
  =&\mathcal{W}_{T}(G_{1}|X_{T-1})+\mathsf{E}\left[ \mathsf{E}\left[S_{T+1}(G_{2})|X_{T}\right]|X_{T-1}\right] \nonumber\\
 =&\mathcal{W}_{T}(G_{1}|X_{T-1})+\mathsf{E}\left[ \mathsf{E}\left[S_{T+1}(G_{2})|Y_{T}(W_{T-1},1(X_{T-1}\in G_{1})),W_{T}=1(X_{T-1}\in G_{1})\right],Z_t|X_{T-1}\right],
\end{align}

where the third equality follows from Assumption \ref{ass:continuous_Markov}.
It shall be noted that in this case, the welfare function is conditional on the information available up to time $T-1$, and we fixed the conditioning value of
$X_{T-1}$ at the observed one.





\subsubsection{Direct estimation of the conditional welfare function. }
With the defined two-period welfare function on hand, we can discuss how to estimate the welfare function.
Similar to that discussed
in Section \ref{sec:uncon_con} for the single-period welfare target, we can either directly estimate the conditional welfare
functions or the regret bounds by the unconditional welfare function.
In this subsection, we directly estimate the two-period conditional welfare function proposed.
Without loss
of generality, we start with the case $X_{t}=\left(Y_{t},W_{t}\right)\in\mathbb{R}\times\{0,1\}$,
i.e., $Y_{t}$ is the only continuous variable in $X_{t}$. It can
be easily extended to the case with other continuous variables $X_{t}=\left(Y_{t},W_{t},Z_{t}\right)\in\mathbb{R}\times\{0,1\}\times\mathbb{R}^{k}$.

Again, we apply the abbreviation: $(\cdot|x)=(\cdot|X_{T-1}=x)$,
where $x=(y,w).$ Now, we have
\begin{align*}
\mathcal{W}_{T:T+1}(G_{1:2}|x) & =\mathcal{W}_{T}(G_{1}|x)+\mathsf{E}\left(\mathcal{W}_{T+1}(G_{2}|X_{T},W_{T}=1(X_{T-1}\in G_{1}))|X_{T-1}=x\right).
\end{align*}

Let $\mathcal{G}_{1:2}$ be the class of feasible policies, which
is a subclass of the class of all the measurable functions defined
on $\mathcal{X}\times\mathcal{X}\to\{0,1\}^{2}$. Conditional on $X_{T-1}=x$,
we define
\[
G_{1:2}^{*}\in\text{argmax}_{G_{1:2}\in\mathcal{G}_{1:2}}\mathcal{W}_{T:T+1}(G_{1:2}|x).
\]
Note $G_{1:2}^{*}$ can depend on $x$, but we suppress this dependence
in the notation.


Similarly to \eqref{eq:sample_kernel}, for any policies $G_{1}$ and $G_{2}$, define
\begin{align*}
\widehat{\mathcal{W}}(G_{1:2}|x) & =\frac{\sum_{t=1}^{T-1}\mathbf{1}\left(W_{t-1}=w\right)K_{h}(Y_{t-1},y)\widehat{\mathcal{W}}_{t}(G_{1})}{\sum_{t=1}^{T-1}\mathbf{1}\left(W_{t-1}=w\right)K_{h}(Y_{t-1},y)}\\
 & +\frac{\sum_{t=1}^{T-2}K_{h}(Y_{t-1},y)\mathbf{1}\left(W_{t-1}=w\right)\mathbf{1}\left(W_{t}=1(X_{t-1}\in G_{1})\right)\widehat{\mathcal{W}}_{t+1}(G_{2})}{\sum_{t=1}^{T-2}K_{h}(Y_{t-1},y)\mathbf{1}\left(W_{t-1}=w\right)\mathbf{1}\left(W_{t}=1(X_{t-1}\in G_{1})\right)},
\end{align*}

where $K_{h}(a,b)=\frac{1}{h}K(\frac{a-b}{h})$ with $K(\cdot)$ assumed
to be a bounded kernel function with a bounded support. 


Let $\mathcal{G}^{2}$ denote the class of feasible unconditional
decision sets, which is a class of subsets of $\mathcal{X}\times\mathcal{X}$,
and
\[
\hat{G}_{1:2}\in\text{argmax}_{G_{1:2}\in\mathcal{G}_{1:2}}\widehat{\mathcal{W}}(G_{1:2}|x).
\]

 Define $\mathcal{{EW}}_{t+1}(G_{1:2}|x)  = \mathsf{E}\left[\mathsf{E}\left[S_{T+1}(G_{2})|Y_{T}(W_{T-1},1(X_{T-1}\in G_{1})),W_{T}=1(X_{T-1}\in G_{1})\right]|X_{T-1}=x\right] $ and   $\mathcal{{EW}}_{t+1}(G_{1:2}|X_{T-1})  = \mathsf{E}\left[ \mathsf{E}\left[S_{T+1}(G_{2})|Y_{T}(W_{T-1},1(X_{T-1}\in G_{1})),W_{T}=1(X_{T-1}\in G_{1})\right]|X_{T-1}\right] $.
First of all, $x$ is a vector of values consisting
of $w$ and $y$.
To construct an MDS, we shall define an intermediate counterpart,
\begin{align*}
\bar{\mathcal{W}}_{h}(G_{1:2}|x) & =\frac{\sum_{t=1}^{T-1}\mathbf{1}\left(W_{t-1}=w\right)K_{h}(Y_{t-1},y)\mathcal{W}_{t}(G_{1}|x)}{\sum_{t=1}^{T-1}\mathbf{1}\left(W_{t-1}=w\right)K_{h}(Y_{t-1},y)}\\
 & +\frac{\sum_{t=1}^{T-2}K_{h}(Y_{t-1},y)\mathbf{1}\left(W_{t-1}=w\right)\mathbf{1}\left(W_{t}=1(X_{t-1}\in G_{1})\right)\mathcal{EW}_{t+1}(G_{1:2}|x)}{\sum_{t=1}^{T-2}K_{h}(Y_{t-1},y)\mathbf{1}\left(W_{t-1}=w\right)\mathbf{1}\left(W_{t}=1(X_{t-1}\in G_{1})\right)}.
\end{align*}

Note that 
\begin{align*}
\E(K_{h}(Y_{t-1},y)\mathbf{1}\left(W_{t-1}=w\right)\mathbf{1}\left(W_{t}=1(X_{t-1}\in G_{1})\right){\widehat{\mathcal{W}}}_{t+1}(G_{2})|\mathcal{F}_{t-1})\\
=K_{h}(Y_{t-1},y)\mathbf{1}\left(W_{t-1}=w\right)\mathbf{1}\left(W_{t}=1(X_{t-1}\in G_{1})\right)\E({\widehat{\mathcal{W}}}_{t+1}(G_{2})|X_{t-1}),
\end{align*}
so $K_{h}(Y_{t-1},y)\mathbf{1}\left(W_{t-1}=w\right)\mathbf{1}\left(W_{t}=1(X_{t-1}\in G_{1})\right)(\widehat{\mathcal{W}}_{t+1}(G_{2})-\mathcal{EW}_{t+1}(G_{2}|X_{T-1}))$
is an MDS with respect to $\mathcal{F}_{t-1}$. Then with similar steps to those in the proof of Theorem \ref{thm: kernel}, we can achieve the closeness between 
 $K_{h}(Y_{t-1},y)\mathbf{1}(W_{t-1}=w)\mathbf{1}(W_{t}=1(X_{t-1}\in G_{1}))\mathcal{EW}_{t+1}(G_{2}|X_{T-1})$ and  $K_{h}(Y_{t-1},y)\mathbf{1}(W_{t-1}=w)\mathbf{1}(W_{t}=1(X_{t-1}\in G_{1}))\mathcal{EW}_{t+1}(G_{2}|X_{T-1}=x)$
by a bias term of order $O_p(h^2)$.
Thus, under assumptions analogous to those in Theorem \ref{thm: kernel}, we can achieve the regret bound of the rate $\sqrt{(T-1)h}^{-1}+{(T-1)}^{-1}+h^{2}$.

\subsubsection{Bounding the conditional regret by the unconditional one}

Similar to Section \ref{sec:uncon_con}, under the correct specification assumption,
we can bound the conditional regret by unconditional regret.
For the multivariate case, the unconditional welfare is defined as
\begin{equation}
\mathcal{W}_{T:T+1}(G_{1:2})=\mathcal{W}_{T}(G_{1})+\E\{\E\left[S_{T+1}(G_{2})|W_{T}=1(X_{T-1}\in G_{1})\right]\}.\label{eq: wel_uncon_conti_mul}
\end{equation}
Note that we have \\
$\E\{\E\left[S_{T+1}(G_{2})|W_{T}=\mathbf{1}(X_{t-1}\in G_{1})\right]\}= \E\left[S_{T+1}(G_{2})\mathbf{1}(W_{T}=\mathbf{1}(X_{t-1}\in G_{1}))\right]/\Pr(W_{T}=1(X_{t-1}\in G_{1}))$. Note that slightly different from the one-period welfare function,  in this unconditional welfare, the second part is still
conditioning on $W_{T}=1(X_{t-1}\in G_{1})$ since the treatment $W_{T}$
is determined by the planner's policy $G_{1}$.

The optimal unconditional policy within the class $\mathcal{G}_{1:2}$
is defined as
\begin{equation} \label{best_G_mul}
G_{1:2}^{*}\in\text{argmax}_{G_{1:2}\in\mathcal{G}_{1:2}}\mathcal{W}_{T:T+1}(G_{1:2}).
\end{equation}

To bound the conditional regret with the unconditional one, we also
need to impose the following assumption.
\begin{assumption}
\label{assu:correct_specification}Let $\mathcal{W}_{T:T+1}(G_{1:2}|x)$
be the conditional welfare defined in \eqref{wel_mul_conti},
\[
\text{argsup}_{G_{1:2}\in\mathcal{G}_{1:2}}\mathcal{W}_{T:T+1}(G_{1:2})\subset\text{argsup}_{G_{1:2}\in\mathcal{G}_{1:2}}\mathcal{W}_{T:T+1}(G_{1:2}|x).
\]
\end{assumption}
\begin{prop}
Under Assumptions \ref{ass:continuous_Markov} and \ref{assu:correct_specification},
\[
G_{1:2}^{*}\in\text{argsup}_{G_{1:2}\in\mathcal{G}_{1:2}}\mathcal{W}_{T:T+1}(G_{1:2}|X_{T-1}).
\]
\end{prop}
The conditional regret is 
\begin{align*}
R_{T:T+1}(G_{1:2}|x) & :=\mathcal{W}_{T:T+1}(G_{1:2}^{*}|x)-\mathcal{W}_{T:T+1}(G_{1:2}|x).
\end{align*}
Similarly, unconditional regret can be expressed as an integral of
conditional regret. Thus, the unconditional welfare and regret are
\begin{align*}
\mathcal{W}_{T:T+1}(G_{1:2}) & =\int\mathcal{W}_{T:T+1}(G_{1:2}|x)dF_{X_{T-1}}(x),\\
R_{T:T+1}(G_{1:2}) & =\mathcal{W}_{T:T+1}(G_{1:2}^{*}|x)-\mathcal{W}_{T:T+1}(G_{1:2})\\
 & =\int R_{T:T+1}(G_{1:2}|x)dF_{X_{T-1}}(x).
\end{align*}

For $x^{\prime}\in\mathcal{X}$
\begin{align*}
A(x^{\prime},G_{1:2}) & =\{x:x\in\mathcal{X}\text{ and }R_{T:T+1}(G_{1:2}|x)\geq R_{T:T+1}(x^{\prime},G_{1:2})\},\\
p_{T-1}(x^{\prime},G_{1:2}) & =\text{Pr}(X_{T-1}\in A(x^{\prime},G_{1:2}))=\int_{x\in A(x^{\prime},G_{1:2})}dF_{X_{T-1}}(x).
\end{align*}

Now, we impose the following assumption.
\begin{assumption} \label{lower_p_dis_con_mul}
For $x^{obs}\in\mathcal{X}$  and any $G_{1:2}\in\mathcal{G}_{1:2}$, there exists a positive constant $\underline{p}$, such that 
\begin{equation}
p_{T-1}(x^{obs},G_{1:2})\geq\underline{p}>0\label{eq:lower_p_dis_con}.
\end{equation}
\end{assumption}

\begin{lemma} \label{lem:lower_RT(G)_mul}
Under Assumption \ref{lower_p_dis_con_mul}
\[
R_{T:T+1}(G_{1:2}|x^{obs})\leq\frac{1}{\underline{p}}R_{T:T+1}(G_{1:2}).
\]
\end{lemma}
Now under Assumptions \ref{lower_p_dis_con_mul} and Lemma \ref{lem:lower_RT(G)_mul}, we can specify the sample analogue,
\begin{align}
 & \widehat{\mathcal{W}}(G_{1:2})=\frac{1}{T}\sum_{t=1}^{T-1}\widehat{\mathcal{W}}_{t}(G_{1})\nonumber  +\frac{1}{T(G_{1})}\sum_{t=1}^{T-2}\mathbf{1}\left(W_{t}=\mathbf{1}(X_{t-1}\in G_{1})\right)S_{t+1}(G_2),\label{eq:mul_cont_uncon_sample}
\end{align}

where $T(G_{1})$ is a random number defined as $T(G_{1})=\#\{1\leq t\leq T-1:W_{t}=\mathbf{1}(X_{t-1}\in G_{1})\}$.
We define the estimated welfare policy, 
\[
\hat{G}_{1:2}\in\text{argmax}_{G_{1:2}\in\mathcal{G}_{1:2}}\widehat{\mathcal{W}}(G_{1:2}).
\]

To prove the bound, we can proceed with similar  steps as in the proof of Theorem \ref{thm:mds_mean_bound_unconditional}.
Namely, we define a conditional welfare function to form an MDS as follows,
\begin{align*}
\mathcal{\bar{W}}(G_{1:2}) & =\frac{1}{T}\sum_{t=1}^{T-1}\mathcal{W}_{t}(G_{1}|\mathcal{F}_{t-2})\\
 & +\frac{1}{\E(T(G_{1}))}\sum_{t=1}^{T-1}\mathsf{E}\{ \left[\mathbf{1}\left(W_{t}=\mathbf{1}(X_{t-1}\in G_{1})\right)S_{t+1}(G_{2})|W_{t}=\mathbf{1}(X_{t-1}\in G_{1})\right]|\mathcal{F}_{t-2}\}.
\end{align*}
Then the unconditional population counterpart is as follows,
\begin{align}
\mathcal{\widetilde{\mathcal{W}}}(G_{1:2}) & =\frac{1}{T}\sum_{t=1}^{T-1}\mathcal{W}_{t}(G_{1})\nonumber \\
 & +\frac{1}{\E(T(G_{1}))}\sum_{t=1}^{T-1}\E(\mathbf{1}\left[W_{t}=\mathbf{1}(X_{t-1}\in G_{1})\right]S_{t+1}(G_{2})).\label{eq:mul_cont_uncon_tlilde}
\end{align}
Now, similar to Assumption \ref{equiv_W_uc}, we impose\\
{\bf Assumption} \ref{equiv_W_uc}'. For $G_{1:2}^{*}$ defined in \eqref{best_G_mul} and any $G_{1:2}\in\mathcal{G}_{1:2}$,
there exists some constant $c$, such that
\[
\mathcal{W}_{T:T+1}(G_{1:2}^{*})-\mathcal{W}_{T:T+1}(G_{1:2})\leq c\left(\mathcal{\widetilde{\mathcal{W}}}(G_{1:2}^{*})-\mathcal{\widetilde{\mathcal{W}}}(G_{1:2})\right).
\]

Then, the conditional regret can be bounded by
\begin{align*}
\mathcal{W}_{T:T+1}(G_{1:2}^{*}|x^{obs})-\mathcal{W}_{T:T+1}(G_{1:2}|x^{obs}) & \leq\frac{1}{\underline{p}}\left[\mathcal{W}_{T:T+1}(G_{1:2}^{*})-\mathcal{W}_{T:T+1}(G_{1:2})\right]\\
 & \leq\frac{c}{\underline{p}}\left[\mathcal{\widetilde{\mathcal{W}}}(G_{1:2}^{*})-\mathcal{\widetilde{\mathcal{W}}}(G_{1:2})\right]\\
 & \leq\frac{c}{\underline{p}}\sup_{G_{1:2}\in\mathcal{\mathcal{G}}_{1:2}}\left[\widehat{\mathcal{W}}(G_{1:2})-\mathcal{\widetilde{\mathcal{W}}}(G_{1:2})\right].
\end{align*}
And
\begin{align*}
\widehat{\mathcal{W}}(G_{1:2})-\mathcal{\widetilde{\mathcal{W}}}(G_{1:2}) & =\left[\mathcal{\bar{W}}(G_{1:2})-\mathcal{\widetilde{\mathcal{W}}}(G_{1:2})\right]+\left[\widehat{\mathcal{W}}(G_{1:2})-\mathcal{\bar{W}}(G_{1:2})\right]\\
 & =I+II
\end{align*}

Similar to Section \ref{continuous}, we can derive  upper bounds for $I$ and $II$.

\end{document}